\documentclass[10pt]{article}
\usepackage{times}
\usepackage[paperwidth=199.8mm,
paperheight=297mm,centering,hmargin=25mm,vmargin=25mm]{geometry}
\usepackage{authblk} 
\usepackage[bottom]{footmisc} 

\usepackage[titletoc,title]{appendix} 
\usepackage{changepage}

\usepackage{hyperref}
\hypersetup{colorlinks=true,citecolor=myblue,linkcolor=myblue,filecolor=myblue,urlcolor=myblue,breaklinks=true}
\usepackage{url}

\usepackage[dvipsnames]{xcolor}
\usepackage{color}
\usepackage{framed}

\definecolor{shadecolor}{rgb}{0.9,0.9,0.9}
\definecolor{mylightgray}{RGB}{100,100,100}

\definecolor{myblue}{RGB}{0, 68, 116}
\definecolor{mycyan}{RGB}{0, 97, 91}
\definecolor{mygreen}{RGB}{2, 102, 1}
\definecolor{myorange}{RGB}{240, 102, 0}
\definecolor{myred}{RGB}{172, 23, 0}
\definecolor{mymagenta}{RGB}{140,16,73}

\usepackage{colortbl}
\usepackage{pifont} 
\usepackage{booktabs}
\usepackage{makecell} 
\usepackage{diagbox}
\usepackage{multirow}

\usepackage{mathtools}
\usepackage{amsmath}
\usepackage{amssymb}
\numberwithin{equation}{section}
\usepackage[shortlabels]{enumitem}
\usepackage{graphicx,epic,eepic,epsfig,latexsym,verbatim}
\usepackage{dsfont}
\usepackage{mathrsfs}
\usepackage{bm}

\usepackage{subcaption}
\usepackage{caption}
\usepackage{tabularx}
\usepackage{float}
\usepackage{tikz}
\usetikzlibrary{arrows.meta}
\usepackage{relsize}
\usepackage{pgfplots}

\usepackage{amsthm}
\usepackage{tcolorbox}
\tcbuselibrary{skins}
\tcbuselibrary{breakable}

\makeatletter
\newcommand{\DefineColoredTheoremStyle}[2]{%
  \newtheoremstyle{#1}
    {3pt}{3pt}
    {\itshape}
    {}
    {\bfseries\color{#2}}
    {\;}
    { }
    {\thmname{##1}\thmnumber{ ##2}\textnormal{\thmnote{ (##3)}}}
}
\makeatother

\DefineColoredTheoremStyle{lemmstyle}{black}
\DefineColoredTheoremStyle{propstyle}{black}
\DefineColoredTheoremStyle{thmstyle}{black}
\DefineColoredTheoremStyle{corostyle}{black}
\DefineColoredTheoremStyle{conjstyle}{black}
\DefineColoredTheoremStyle{problstyle}{black}
\DefineColoredTheoremStyle{plainstyle}{black}

\RequirePackage[ruled,vlined,linesnumbered]{algorithm2e}
\SetKwInput{KwInput}{Input}
\SetKwInput{KwOutput}{Output}
\SetAlgoNlRelativeSize{-1}
\SetAlCapNameFnt{\normalsize}
\SetAlCapFnt{\normalsize}
\SetKwComment{Comment}{$\triangleright$\ }{}
\DontPrintSemicolon
 
\newtheorem{definition}{Definition}[section]
\newtheorem{proposition}[definition]{Proposition}

\newtheorem{lemma}[definition]{Lemma}
\newtheorem{theorem}[definition]{Theorem}

\newtheorem{problem}{Problem}

\newcommand{\nc}{\newcommand}
\nc{\rnc}{\renewcommand}

\nc{\<}{\langle}
\rnc{\>}{\rangle}
\nc{\bra}[1]{\langle#1|}
\nc{\ket}[1]{|#1\rangle}
\nc{\ketbra}[1]{|#1\rangle\!\langle#1|}
\nc{\braket}[2]{\langle#1|#2\rangle}
\nc{\braandket}[3]{\langle #1|#2|#3\rangle}
\nc{\proj}[1]{| #1\rangle\!\langle #1 |}
\nc{\avg}[1]{\langle#1\rangle}

\makeatletter
\newsavebox{\@brx}
\newcommand{\llangle}[1][]{\savebox{\@brx}{\(\m@th{#1\langle}\)}%
  \mathopen{\copy\@brx\kern-0.5\wd\@brx\usebox{\@brx}}}
\newcommand{\rrangle}[1][]{\savebox{\@brx}{\(\m@th{#1\rangle}\)}%
  \mathclose{\copy\@brx\kern-0.5\wd\@brx\usebox{\@brx}}}
\makeatother

\nc{\rank}{\operatorname{Rank}}
\nc{\id}{{\operatorname{id}}}
\nc{\supp}{{\operatorname{supp}}}
\nc{\smfrac}[2]{\mbox{$\frac{#1}{#2}$}}
\nc{\tr}{\operatorname{Tr}}
\nc{\ox}{\otimes}
\nc{\floor}[1]{\lfloor #1 \rfloor}
\nc{\trans}{\mathsf T}
\nc{\img}{\mathbf{i}}

\newcommand{\norm}[2]{{\left\lVert #1 \right\rVert}_{#2}}

\nc{\cA}{{\cal A}}
\nc{\cB}{{\cal B}}
\nc{\cC}{{\cal C}}
\nc{\cD}{{\cal D}}
\nc{\cE}{{\cal E}}
\nc{\cF}{{\cal F}}
\nc{\cG}{{\cal G}}
\nc{\cH}{{\cal H}}
\nc{\cI}{{\cal I}}
\nc{\cJ}{{\cal J}}
\nc{\cK}{{\cal K}}
\nc{\cL}{{\cal L}}
\nc{\cM}{{\cal M}}
\nc{\cN}{{\cal N}}
\nc{\cO}{{\cal O}}
\nc{\cP}{{\cal P}}
\nc{\cQ}{{\cal Q}}
\nc{\cR}{{\cal R}}
\nc{\cS}{{\cal S}}
\nc{\cT}{{\cal T}}
\nc{\cV}{{\cal V}}
\nc{\cU}{{\cal U}}
\nc{\cX}{{\cal X}}
\nc{\cY}{{\cal Y}}
\nc{\cZ}{{\cal Z}}
\nc{\cW}{{\cal W}}

\nc{\RR}{{{\mathbb R}}}
\nc{\CC}{{{\mathbb C}}}
\nc{\FF}{{{\mathbb F}}}
\nc{\NN}{{{\mathbb N}}}
\nc{\ZZ}{{{\mathbb Z}}}
\nc{\QQ}{{{\mathbb Q}}}
\nc{\UU}{{{\mathbb U}}}
\nc{\EE}{{{\mathbb E}}}

\nc{\fF}{{\mathfrak{F}}}
\nc{\fL}{{\mathfrak{L}}}

\nc{\sK}{{{\mathscr{K}}}}
\nc{\sS}{{{\mathscr{S}}}}
\nc{\sT}{{{\mathscr{T}}}}
\nc{\sA}{{{\mathscr{A}}}}
\nc{\sB}{{{\mathscr{B}}}}
\nc{\sC}{{{\mathscr{C}}}}
\nc{\sE}{{{\mathscr{E}}}}
\nc{\sL}{{{\mathscr{L}}}}
\nc{\sH}{{{\mathscr{H}}}}
\nc{\sG}{{{\mathscr{G}}}}
\nc{\sF}{{{\mathscr{F}}}}
\nc{\sP}{{{\mathscr{P}}}}
\nc{\sI}{{{\mathscr{I}}}}
\nc{\sN}{{{\mathscr{N}}}}
\nc{\sM}{{{\mathscr{M}}}}
\nc{\sX}{{{\mathscr{X}}}}
\nc{\sY}{{{\mathscr{Y}}}}
\nc{\sO}{{{\mathscr{O}}}}
\nc{\sR}{{{\mathscr{R}}}}

\newcommand{\bA}{\mathbb{A}}
\newcommand{\bE}{\mathbb{E}}
\newcommand{\bC}{\mathbb{C}}

\newcommand{\bF}{\mathbb{F}}

\nc{\Choi}{Choi-Jamio\l{}kowski }
\nc{\reg}{\infty}
\nc{\amo}{\text{\rm amo}}
\nc{\Renyi}{R\'{e}nyi }

\nc{\conv}{\operatorname{conv}}
\nc{\cvxset}{\mathscr{C}}
\nc{\aff}{\operatorname{aff}}
\nc{\cone}{\operatorname{cone}}
\nc{\diam}{\operatorname{diam}}

\nc{\RM}{{{\mathscr{R}}}}

\nc{\END}{\operatorname{End}}
\nc{\PERM}{\mathfrak{\sigma}}

\nc{\Cone}{\text{\rm Cone}}
\nc{\sep}{{\SEP}}

\nc{\DD}{{{\mathbb D}}}
\nc{\BS}{{\scriptscriptstyle \rm {BS}}}
\nc{\Sand}{{\scriptscriptstyle  \rm S}}
\nc{\Petz}{{\scriptscriptstyle  \rm P}}
\nc{\Hypo}{{\scriptscriptstyle  \rm H}}
\nc{\Meas}{{\scriptscriptstyle \rm M}}
\nc{\Proj}{{{\scriptscriptstyle \rm P}}}

\nc{\suchthat}{\text{\rm s.t.}}

\nc{\pl}{{\scalebox{0.7}{+}}}
\nc{\HERM}{\mathscr{H}}
\nc{\PSD}{\HERM_{\pl}}
\nc{\PD}{\HERM_{\pl\pl}}
\nc{\density}{\mathscr{D}}
\nc{\subdensity}{\mathscr{D}_\bullet}

\nc{\polarPSD}[1]{{#1}_{\pl}^{\circ}}
\nc{\polarPSDre}[1]{{#1}_{\pl}^{\star}}
\nc{\polarPD}[1]{{#1}_{\pl\pl}^{\circ}}

\nc{\distill}{{\operatorname{Distill}}}
\nc{\dilute}{{\operatorname{Dilute}}}

\nc{\PPT}{\text{\rm PPT}}
\nc{\Rains}{\text{\rm Rains}}
\nc{\WD}{\text{\rm WD}}
\nc{\SEP}{\text{\rm SEP}}
\nc{\PSEP}{\text{\rm PSEP}}
\nc{\CPTP}{\text{\rm CPTP}}
\nc{\GPO}{\text{\rm GPO}}
\nc{\GPL}{\text{\rm GPL}}
\nc{\TO}{\text{\rm TO}}
\nc{\GPC}{\text{\rm GPC}}
\nc{\POVM}{\text{\rm POVM}}
\nc{\PVM}{\text{\rm PVM}}
\nc{\CP}{\text{\rm CP}}
\nc{\adv}{\text{\rm adv}}
\nc{\spec}{\text{\rm spec}}
\nc{\poly}{\text{\rm poly}}
\nc{\End}{\operatorname{End}}
\nc{\Par}{\operatorname{Par}}
\nc{\RNG}{\operatorname{RNG}}
\nc{\STAB}{\text{\rm STAB}}
\nc{\epi}{\boldsymbol{\operatorname{epi}}}
\nc{\op}{\boldsymbol{\operatorname{op}}}
\newcommand{\Var}{\mathbb{V}}

\makeatletter
\newcommand*\rel@kern[1]{\kern#1\dimexpr\macc@kerna}
\newcommand*\widebar[1]{%
  \begingroup
  \def\mathaccent##1##2{%
    \rel@kern{0.8}%
    \overline{\rel@kern{-0.8}\macc@nucleus\rel@kern{0.2}}%
    \rel@kern{-0.2}%
  }%
  \macc@depth\@ne
  \let\math@bgroup\@empty \let\math@egroup\macc@set@skewchar
  \mathsurround\z@ \frozen@everymath{\mathgroup\macc@group\relax}%
  \macc@set@skewchar\relax
  \let\mathaccentV\macc@nested@a
  \macc@nested@a\relax111{#1}%
  \endgroup
}
\makeatother

\definecolor{googleblue}{HTML}{4285F4}
\definecolor{googlered}{HTML}{DB4437}
\definecolor{googleyellow}{HTML}{F4B400}
\definecolor{googlegreen}{HTML}{0F9D58}
\definecolor{klevinblue}{HTML}{002FA7}
\definecolor{tiffanyblue}{HTML}{0ABAB5}

\usepackage{xcolor}
\usetikzlibrary{quantikz}

\usepackage[T1]{fontenc}
\usepackage[utf8]{inputenc}

\begin{document}

\title{\Large \textbf{Optimal Distributed Similarity Estimation of Quantum Channels}}

\author[1,2,3]{Congcong Zheng}
\author[4]{Kun Wang\thanks{Corresponding author: \href{mailto:nju.wangkun@gmail.com}{nju.wangkun@gmail.com}}}
\author[1,2,3]{Xutao Yu}
\author[4]{Ping Xu}
\author[5,2,3]{Zaichen Zhang\thanks{Corresponding author: \href{mailto:zczhang@seu.edu.cn}{zczhang@seu.edu.cn}}}

\affil[1]{\small State Key Lab of Millimeter Waves, Southeast University, Nanjing 211189, China}
\affil[2]{\small Purple Mountain Laboratories, Nanjing 211111, China}
\affil[3]{\small Frontiers Science Center for Mobile Information Communication and Security,\protect\\ Southeast University, Nanjing 210096, China}
\affil[4]{\small College of Computer Science and Technology, National University of Defense Technology, Changsha 410073, China}
\affil[5]{\small National Mobile Communications Research Laboratory, Southeast University, Nanjing 210096, China}

\date{\today}
\maketitle

\begin{abstract}
We study the query complexity of estimating the similarity of two unknown quantum channels available at separated devices.
For channels $\cE$ and $\cF$ on a $d$-dimensional system, the target functional is the normalized Choi overlap $\tr[J_\cE J_\cF]/d^2$, where $J_\cE$ and $J_\cF$ are unnormalized Choi operators.
We call this task \emph{distributed similarity estimation of quantum channels} (DSEC).
Our main result is a tight characterization: estimating this overlap to additive error $\varepsilon$ with constant success probability requires and is achievable with
\begin{align*}
\Theta\!\left(\max\left\{\frac{\sqrt d}{\varepsilon},\frac{1}{\varepsilon^2}\right\}\right)
\end{align*}
queries to each channel.
The lower bound holds even for the most general model considered here, allowing arbitrary ancillas, coherent adaptive access, quantum memory within each laboratory, and multi-round classical communication.
The dimension-dependent term follows from a shared-versus-independent Haar-unitary testing problem, while the $\varepsilon$-dependence is obtained through a random-phase channel construction.
Conversely, we give a non-adaptive, ancilla-free, incoherent randomized-measurement protocol for arbitrary channels that attains the same scaling.
The protocol uses only shared classical randomness to coordinate local state preparations and measurement bases.
For unitary channels, our protocol reduces the dimension-dependent upper bound from $d/\varepsilon$ for independent coherent shadows to $\sqrt d/\varepsilon$.
These results identify the optimal information-theoretic cost of distributed Choi overlap estimation and clarify the role of public classical randomness in distributed quantum channel learning.
\end{abstract}

\tableofcontents

\section{Introduction}
\label{sec:introduction}

Estimating scalar functionals of unknown quantum channels is a basic problem in quantum information theory.
When the quantity of interest is a single similarity, fidelity, or noise parameter, full channel tomography is often unnecessary and information-theoretically wasteful.
This motivates a query-complexity viewpoint: which channel functionals can be estimated without learning the channel itself, and what resources determine the optimal scaling?

In this work we study a distributed version of this question, which is illustrated in Fig.~\ref{fig: task description}.
Two separated laboratories are given local query access to two unknown channels $\cE$ and $\cF$ acting on a $d$-dimensional system.
Their goal is to estimate the normalized Choi overlap $\tr[J_\cE J_\cF]/d^2$ using only local operations and classical communication (LOCC).
Here $J_\cE$ denotes the unnormalized Choi operator of $\cE$.
We call this task \emph{distributed similarity estimation of quantum channels} (DSEC).
The functional $\tr[J_\cE J_\cF]/d^2$ is a Hilbert--Schmidt measure of channel similarity and is the channel analogue of the inner product appearing in distributed state comparison.

\begin{figure}[!htbp]
\centering
\includegraphics[width=0.6\linewidth]{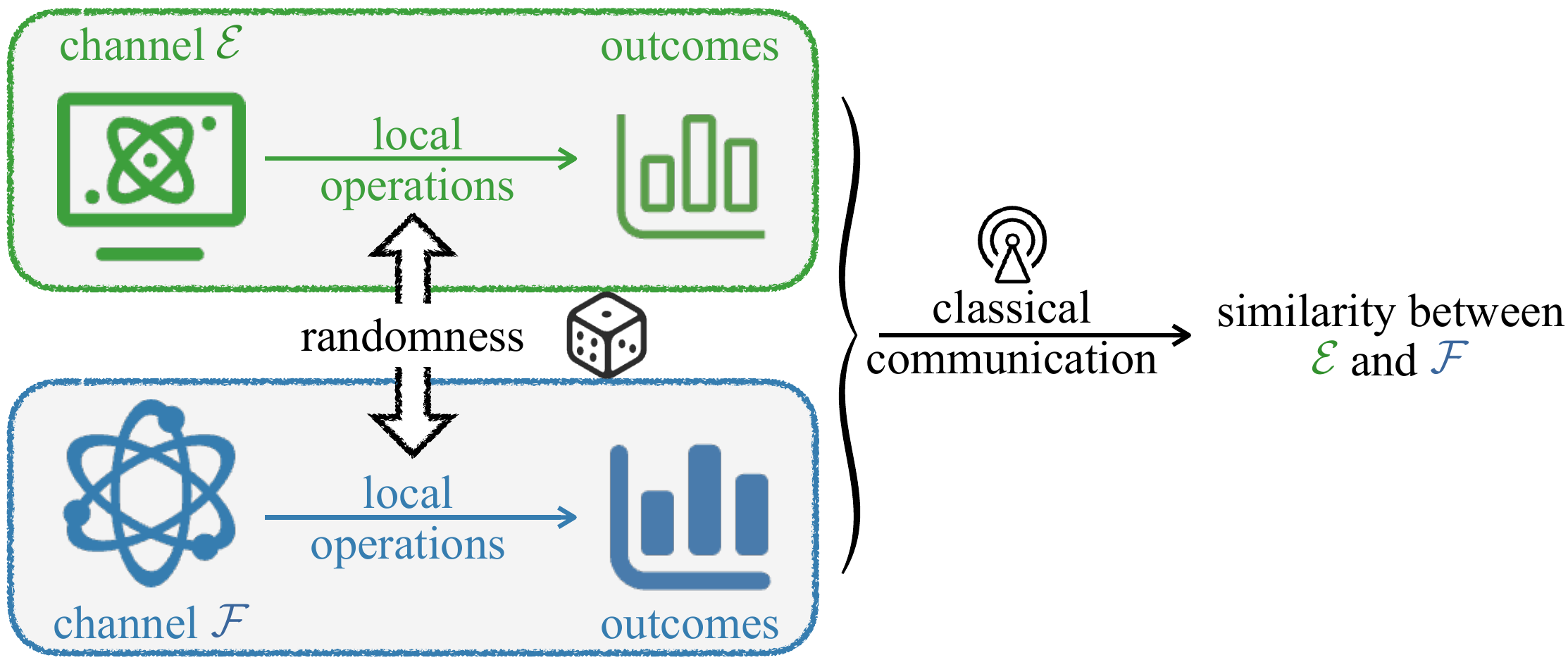}
\caption{\raggedright
Distributed Similarity Estimation of Quantum Channels (DSEC). 
The task is to estimate the normalized Hilbert--Schmidt inner product of the Choi operators of two unknown channels $\cE$ and $\cF$ implemented on different devices, using only local operations and classical communication (LOCC) between the two devices.}
\label{fig: task description}
\end{figure}

DSEC is also the core information-theoretic primitive behind cross-platform verification of quantum dynamics.
Quantum information processing is no longer tied to a single physical architecture: superconducting circuits, trapped ions, neutral atoms, and photonic processors now offer complementary routes to programmable quantum dynamics, often through remote or cloud-based interfaces~\cite{popkin2016quest, preskill2018quantum, brown20245}.
Existing certification protocols often compare an experimentally implemented state or channel with a trusted target description, including direct fidelity estimation~\cite{flammia2011direct, dasilva2011practical}, randomized benchmarking~\cite{emerson2007symmetrized, lu2015experimental, helsen2022general}, and quantum verification~\cite{pallister2018optimala, wang2019optimala, zheng2024efficient, chen2025quantum}.
In many cross-platform settings, however, the two implemented channels are both unknown, even when their target dynamics are specified.
For example, two simulators may target the same evolution $e^{-itH}$ but implement different channels because of control errors and noise that are not fully characterized, where $H$ is the Hamiltonian and $t$ is the evolution time.
Even for known $H$, computing the ideal output state classically can be computationally demanding~\cite{wild2023classical,shaw2024benchmarking}, motivating direct experimental comparisons~\cite{elben2020crossplatform}.

Such direct comparisons assess consistency between independent implementations, which cannot generally be inferred from their individual fidelities to an ideal target~\cite{zheng2024crossplatform}.
Zhu \emph{et al.}~\cite{zhu2022crossplatform} directly compare states from trapped-ion and superconducting processors, while Zheng \emph{et al.}~\cite{zheng2024crossplatform} directly compare gates across IBM and Baidu devices and monitor device stability.
Dalton \emph{et al.}~\cite{dalton2025resourceefficient} and Wu \emph{et al.}~\cite{wu2025state} demonstrate state and unitary verification on modular superconducting processors.
On photonic platforms, Zhan \emph{et al.}~\cite{zhan2025experimental} benchmark state-overlap estimation strategies, while Li \emph{et al.}~\cite{li2026machine} apply direct similarity measurements to quantum-data classification and online learning.
These photonic studies concern state overlaps and include joint measurements, whereas we compare unknown channels under LOCC.
This distributed setting supports consistency checks between devices without quantum links and across experimental runs without retaining quantum outputs.

\paragraph*{Problem setup.}
Let $\cH$ be the Hilbert space of a $d$-dimensional system.
Let $\cE$ and $\cF$ be two unknown channels implemented on different, possibly distant, quantum devices.
Define the Choi operator of $\cE$ as $J_\cE:=(\cE\ox\cI)(\ketbra{\Phi})$, where $\cI$ is the identity channel and $\ket{\Phi}:=\sum_i\ket{ii}$ is the unnormalized maximally entangled state in $\cH^{\ox 2}$~\cite{jamiolkowski1972linear}; define $J_\cF$ similarly.
The DSEC task is to estimate $\tr[J_\cE J_\cF]/d^2$ to additive accuracy $\varepsilon$ with constant success probability, using only LOCC between the two devices.
Variants of this task have recently appeared in quantum information theory~\cite{zheng2024crossplatform, angrisani2025learninga, dalton2025resourceefficient}, with applications to circuit equivalence checking~\cite{sun2022equivalence, wang2022equivalence, tang2024experimental} and quantum channel benchmarking~\cite{eisert2020quantuma, fawzi2023quantum, kliesch2021theory, hu2025quantum}.
The Choi overlap, together with the two self-overlaps, determines several measures of channel similarity and distance, including the max-fidelity-type similarity ratio~\cite{zheng2024crossplatform, liang2019quantuma}
\begin{align}
\frac{\tr[J_\cE J_\cF]}{\max\{\tr[J_\cE^2], \tr[J_\cF^2]\}},
\label{eq:max fidelity similarity}
\end{align}
where $\tr[J_\cE^2]/d^2$ and $\tr[J_\cF^2]/d^2$ describe the unitarities of $\cE$ and $\cF$.
Unitarity estimation has been considered in~\cite{chen2023unitarity}, and our algorithm can also serve as a new tool for estimating such self-overlap quantities.
The ratio in Eq.~\eqref{eq:max fidelity similarity} lies in $[0,1]$ and equals $1$ exactly when $\cE=\cF$.
Thus the cross-overlap and two locally estimable self-overlaps determine a faithful similarity, as well as the squared Hilbert--Schmidt distance
\begin{align*}
\norm{\frac{J_\cE}{d}-\frac{J_\cF}{d}}{2}^2
=\frac{\tr[J_\cE^2]+\tr[J_\cF^2]-2\tr[J_\cE J_\cF]}{d^2}.
\end{align*}
Additive accuracy $\varepsilon$ in this distance requires overlap accuracy $\varepsilon/4$.
The overlap also has an average-input interpretation. 
The Haar second-moment identity gives
\begin{align*}
\frac{\tr[J_\cE J_\cF]}{d^2}
=\frac{d+1}{d}\bE_{\psi\sim\mu_H}\tr[\cE(\ketbra{\psi})\cF(\ketbra{\psi})]
-\tr[\cE(I/d)\cF(I/d)].
\end{align*}
When $\cF$ is unitary, the average gate fidelity is therefore $(\tr[J_\cE J_\cF]/d+1)/(d+1)$.
For unitary channels $\cU,\cV$, the average squared output trace distance in~\cite{zhao2024bounded,rosenthal2024quantum} is
\begin{align*}
\bE_{\psi\sim\mu_H}\left[\left(\frac12\norm{\cU(\ketbra{\psi})-\cV(\ketbra{\psi})}{1}\right)^2\right]
=\frac{d}{d+1}\left(1-\frac{\tr[J_\cU J_\cV]}{d^2}\right).
\end{align*}
Both these terms are determined by the Choi overlap, which is therefore a fundamental quantity in distributed channel comparison.

\paragraph*{Why DSEC is not simply DIPE on Choi states.}
Since $\tr[J_\cE J_\cF]/d^2$ is a bilinear Choi functional, it is tempting to view DSEC as a straightforward channel version of distributed inner product estimation (DIPE) for quantum states.
This viewpoint is misleading in several ways.
First, a literal reduction would apply state DIPE to the normalized Choi states $J_\cE/d$ and $J_\cF/d$, which live in dimension $d^2$.
The known DIPE scaling in dimension $D$ is $\Theta(\max\{\sqrt D/\varepsilon,1/\varepsilon^2\})$~\cite{anshu2022distributed}; with $D=d^2$, this gives a dimension-dependent cost $d/\varepsilon$, not the $\sqrt d/\varepsilon$ scaling proved here.
Second, the lower bound for state DIPE cannot be imported directly: a channel learner has strictly richer oracle access than a state-sample learner, because it may choose input states, use ancillas, preserve quantum memory, and interleave unknown channel calls with known operations.
A lower bound that only rules out passive copies of Choi states does not rule out such adaptive channel-query protocols.
Third, the upper bound is not a Choi-state protocol: our optimal general-channel algorithm never prepares Choi states, uses no quantum memory, and estimates the Choi overlap through correlated local state preparations and randomized output measurements.
Thus the final $\sqrt d/\varepsilon$ dimension dependence should not be read as inherited from DIPE; it is a new channel-query phenomenon.

The fundamental questions addressed in this paper are therefore:
\begin{center}
\emph{How many channel queries per device are necessary for DSEC? \\
Can the same scaling be achieved by an explicit non-adaptive incoherent protocol?}   
\end{center}

\paragraph*{Our results.}
Our main result is a tight information-theoretic characterization of DSEC.
For two unknown $d$-dimensional channels and additive accuracy $\varepsilon$, the optimal query complexity per device is
\begin{align}
\Theta\!\left(\max\left\{\frac{\sqrt{d}}{\varepsilon},\frac{1}{\varepsilon^2}\right\}\right).
\end{align}
Equivalently, the dimension-dependent term is $\sqrt d/\varepsilon$, not $d/\varepsilon$ or a tomography-scale power of $d$.
A query means one use of the corresponding unknown channel; all query complexities below are counted per device, and classical communication is not charged.
We consider $d\geq2$, with $0<\varepsilon\leq\varepsilon_0$ for the lower bounds, where $\varepsilon_0>0$ is a sufficiently small universal constant.
The following informal theorem summarizes the result.

\paragraph*{Main theorem (informal).}
For any fixed success probability greater than $1/2$, estimating $\tr[J_\cE J_\cF]/d^2$ to additive error $\varepsilon$ requires
\begin{align}
\Omega\!\left(\max\left\{\frac{\sqrt{d}}{\varepsilon},\frac{1}{\varepsilon^2}\right\}\right)
\end{align}
queries to each channel, even against adaptive coherent distributed algorithms with arbitrary ancillas and multi-round LOCC.
Conversely, there is a non-adaptive, ancilla-free, incoherent randomized-measurement protocol for arbitrary channels using
\begin{align}
\cO\!\left(\max\left\{\frac{\sqrt{d}}{\varepsilon},\frac{1}{\varepsilon^2}\right\}\right)
\end{align}
queries to each channel.

The lower-bound proof has three ingredients.
First, we prove a constant-error hardness result for distinguishing a shared Haar-random unitary pair $(U,U)$ from an independent pair $(U,V)$.
This step already requires $\Omega(\sqrt d)$ queries and holds even with coherent adaptive access.
This $\Omega(\sqrt d)$ lower bound also applies to constant-gap testing of the diamond distance and the trace distance between normalized Choi states.
Second, we embed this constant-gap testing problem into an $\varepsilon$-separated channel-overlap estimation problem through a random-phase channel construction.
The random phase removes cross terms of different degrees in the hidden Haar unitaries.
We bound the total weight of high-degree terms in the actual outcome probabilities and compare the remaining nonnegative terms by Haar moment inequalities, yielding $T\varepsilon=\Omega(\sqrt d)$.
Third, the independent $\Omega(1/\varepsilon^2)$ term follows from a two-point distinguishing argument for noisy channels, using exact simulation from their Choi states.

The upper bound is achieved by randomized measurements.
For unitary channels, we first give a coherent protocol with query complexity $\cO(\max\{\sqrt d/\varepsilon,1/\varepsilon^2\})$.
For arbitrary channels, our main algorithm is operationally simpler: in each shot, each learner prepares a local state, applies its unknown channel once, and immediately measures the output.
Nevertheless, this non-adaptive incoherent protocol achieves the same optimal query complexity.

A conceptual consequence is that coherent quantum control is not the decisive asymptotic resource for DSEC.
The optimal incoherent protocol uses shared classical randomness to coordinate random seeds for state preparation and measurement, but it uses no shared entanglement, no quantum memory, and no adaptive feedback.
To isolate the role of this resource, we compare our protocol with independent channel classical shadows~\cite{kunjummen2023shadow, levy2024classical, li2025nearly}, where the two devices choose their SPAM settings independently.
The independent incoherent channel-shadow estimator analyzed in Section~\ref{sec:comparison_classical_shadow} has query complexity $\cO(\max\{d^2/\varepsilon,1/\varepsilon^2\})$ for unital channels.
On the other hand, the independent coherent channel-shadow estimator has query complexity $\cO(\max\{d/\varepsilon,1/\varepsilon^2\})$ for unitary channels.
Our incoherent protocol reduces the dimension-dependent term in these upper bounds to $\sqrt d/\varepsilon$ and applies to arbitrary channels.
The obtained results are summarized in Table~\ref{tab:summarization}.

\begin{table*}[t]
\centering
\renewcommand{\arraystretch}{2} 
\setlength{\tabcolsep}{4pt} 
\setlength\heavyrulewidth{0.3ex}  
\begin{tabular}{@{}ccccc@{}}
\toprule
\textbf{Access Model} & \textbf{Lower Bound} & \textbf{Algorithm} & \textbf{Query Complexity} & \textbf{Applicable Channel} \\ \midrule
\multirow{2}{*}{Coherent} & \multirow{4}{*}{\raisebox{-5pt}{\makecell[c]{
$\Omega(\max\{\sqrt d/\varepsilon,1/\varepsilon^2\})$\\[2pt]
{\footnotesize Worst case over general channels}\\[2pt]
{\footnotesize (Theorem~\ref{the:lower bound for DSEC})}
}}} & \cellcolor{black!6}RM with SR 
& \cellcolor{black!6}\makecell[c]{$\cO(\max\{\sqrt d/\varepsilon,1/\varepsilon^2\})$\\[2pt]{\footnotesize (Theorem~\ref{the:coherent upper bound})}} & \cellcolor{black!6}Unitary \\
 &  & Independent CS 
& \makecell[c]{$\cO(\max\{d/\varepsilon,1/\varepsilon^2\})$\\[2pt]{\footnotesize (Proposition~\ref{the:independent classical shadow upper bound})}} & Unitary \\[1pt]
\multirow{2}{*}{Incoherent} &  & \cellcolor{black!6}RM with SR 
& \cellcolor{black!6}\makecell[c]{$\cO(\max\{\sqrt d/\varepsilon,1/\varepsilon^2\})$\\[2pt]{\footnotesize (Theorem~\ref{the:incoherent upper bound})}} & \cellcolor{black!6}General \\
 &  & Independent CS 
& \makecell[c]{$\cO(\max\{d^2/\varepsilon,1/\varepsilon^2\})$\\[2pt]{\footnotesize (Proposition~\ref{the:independent classical shadow upper bound})}} & Unital \\ \bottomrule
\end{tabular}
\caption{Summary of query-complexity results for distributed similarity estimation of quantum channels (DSEC). Here $d$ is the channel dimension, $\varepsilon$ is the additive error, RM denotes randomized measurements, SR denotes shared randomness, and CS denotes classical shadows. Query complexities are per device, up to universal constants.
The ``Applicable Channel'' column specifies the channel class supported by each algorithm.}
\label{tab:summarization}
\end{table*}

\paragraph*{Related works.}
We briefly discuss the closest related lines of work.
\begin{enumerate}
\item \textbf{\emph{Distributed Inner Product Estimation.}}
Our work is motivated by DIPE for quantum states~\cite{elben2020crossplatform, knorzer2023crossplatform, anshu2022distributed, qian2024multimodal, zheng2025distributed, knorzer2025distributed, hinsche2025efficient, wu2025state, gong2024sample, arunachalam2024distributed, dalton2025resourceefficient}, but the channel problem requires different ideas.
Directly applying the general state DIPE algorithm to Choi states gives a query upper bound with a $d/\varepsilon$ dimension-dependent term, and the state DIPE lower bound does not directly apply to adaptive channel-query algorithms.
Our lower bound therefore analyzes arbitrary distributed channel learners directly, while our upper bound exploits channel structure and shared randomness without preparing Choi states.

\item \textbf{\emph{Channel Unitarity Estimation.}}
Chen \emph{et al.}~\cite{chen2023unitarity} show that unitarity estimation has query complexity $\Theta(1/\varepsilon^{2})$ with coherent access.
With incoherent access, they give an upper bound $\cO(\sqrt{d}/\varepsilon^{2})$ and a lower bound $\Omega(\sqrt{d}+1/\varepsilon^{2})$.
The closest connection to our work is that unitarity corresponds to the special self-overlap case, whereas DSEC estimates the cross-overlap $\tr[J_\cE J_\cF]/d^2$ of two remote unknown channels. 
The main technical difference is in the coherent lower bound. 
Our coherent lower bound proves a dimension-dependent distributed hardness: even with coherent access, distinguishing a shared Haar random unitary from two independent Haar random unitaries requires $\Omega(\sqrt{d})$ queries. 

\item \textbf{\emph{Channel Tomography and Classical Shadow.}}
Several existing techniques can potentially address DSEC, including channel tomography and classical shadows for channels.
Channel tomography is highly resource-intensive, requiring $\cO(d^4)$ queries for general channels and $\cO(d^2)$ for unitary channels~\cite{mele2025optimal, chen2025quantumc, chen2026quantum}. 
Classical shadows for channels can be sample-efficient for many prediction tasks~\cite{kunjummen2023shadow, levy2024classical, li2025nearly}, but the independent-shadow estimator analyzed in Section~\ref{sec:comparison_classical_shadow} gives a larger dimension-dependent upper bound for DSEC compared to our algorithms.
\end{enumerate}

The rest of the paper is organized as follows. Section~\ref{sec:preliminaries} fixes notation and reviews the learning models. Section~\ref{sec:lower_bound} proves the lower bound under both coherent and incoherent access. Section~\ref{sec:upper_bound} gives matching algorithms. Section~\ref{sec:comparison_classical_shadow} compares these algorithms with independent channel classical shadows. The remaining proofs are collected in the appendices.

\section{Preliminaries}
\label{sec:preliminaries}

\subsection{Notation}
Let $\cH$ be the Hilbert space of a $d$-dimensional system.
The set of Hermitian operators on $\cH$ is denoted by $\cB(\cH)$, and the set of density matrices on $\cH$ is denoted by $\cD(\cH)$. 
For a quantum channel $\cE$, we define its Choi operator as~\cite{jamiolkowski1972linear} 
\begin{align}
J_\cE := (\cE\ox \cI)(\ketbra{\Phi}), \label{eq:choi state}
\end{align}
where $\ket{\Phi} := \sum_i \ket{ii}$ is the unnormalized maximally entangled state. 
The Choi operator provides the useful identity
\begin{align}
\tr\left[O\cE(\rho)\right] = \tr\left[\left(O\ox \rho^T\right)J_\cE\right]. 
\end{align}
A quantum channel $\cE$ can equivalently be represented by a set of Kraus operators $\{E_i\}_{i}$~\cite{nielsen2010quantum}, 
\begin{align}
\cE(\rho) = \sum_{i} E_i \rho E_i^\dagger.
\end{align}
For two quantum channels $\cE$ and $\cF$, with Kraus operators $\{E_i\}_i$ and $\{F_j\}_j$, the inner product of their Choi operators can be represented as
\begin{align}
\tr\left[J_\cE J_\cF\right] 
= \sum_{i,j} \left| \tr\left[E_i^\dagger F_j\right] \right|^2.
\label{eq:choi and kraus}
\end{align}
For any operator $A\in\cB(\cH)$, we use the notation $|A\rrangle$ to represent the corresponding vectorized operator.
For example, 
\begin{align}
|\ket{\psi}\!\bra{\phi} \rrangle := \ket{\psi}\ox\ket{\phi^*}, \quad 
|ABC^\dagger\rrangle = A\ox C^* | B\rrangle. 
\end{align}
The inner product between two vectorized operators is defined as $\llangle A|B\rrangle = \tr[A^\dagger B]$. 
For operators $A$ and $B$, we write $A \succeq B$ if $A - B$ is positive semidefinite.

\subsection{Channel Learning Models}

We use two axes to specify the learning model: \emph{channel learning access} and \emph{shared randomness}.

\paragraph*{Learning access.}
Following the classifications in~\cite{huang2022quantum, chen2022exponential, chen2023unitarity, caro2024learning, li2025nearly, jeon2025query}, we consider two types of learning access, illustrated in Fig.~\ref{fig: learning models}(a). 
(i) \emph{Incoherent access}, also known as learning without quantum memory:
In each query, the learner prepares an arbitrary quantum state, applies the unknown channel, and immediately measures the output state.
The key restriction is that the unknown channel is queried exactly once before each measurement.
(ii) \emph{Coherent access}, also known as learning with quantum memory:
In each experimental round, the learner prepares an arbitrary quantum state, may interleave multiple queries of the unknown channel with arbitrary intermediate quantum channels, 
and then applies an arbitrary measurement on the final output state.
This allows quantum information to be preserved across multiple channel queries.

\paragraph*{Shared randomness.}
Beyond learning access, shared randomness plays a crucial role in distributed learning with randomized measurements~\cite{elben2023randomized}. 
In many learning algorithms, SPAM settings are essential components.
Within the framework of randomized measurements, both initial states and measurement settings are sampled randomly to enhance practicality and sample efficiency. 
In the distributed scenario, allowing devices to share part or all of this randomness can significantly improve performance, as demonstrated in DIPE~\cite{anshu2022distributed}. 
We consider three levels of shared randomness, illustrated in Fig.~\ref{fig: learning models}(b):
(i) shared randomness for both SPAM settings, 
(ii) shared randomness only for state preparation, and 
(iii) no shared randomness between devices.
Physically, shared randomness should be viewed as a classical synchronization resource: before the experiment, the two laboratories agree on the same random seed, so that corresponding experimental shots use correlated input states or measurement bases.
A key finding of this work is that shared randomness provides substantial advantages in completing DSEC.

\begin{figure}[t]
\centering
\includegraphics[width=0.6\linewidth]{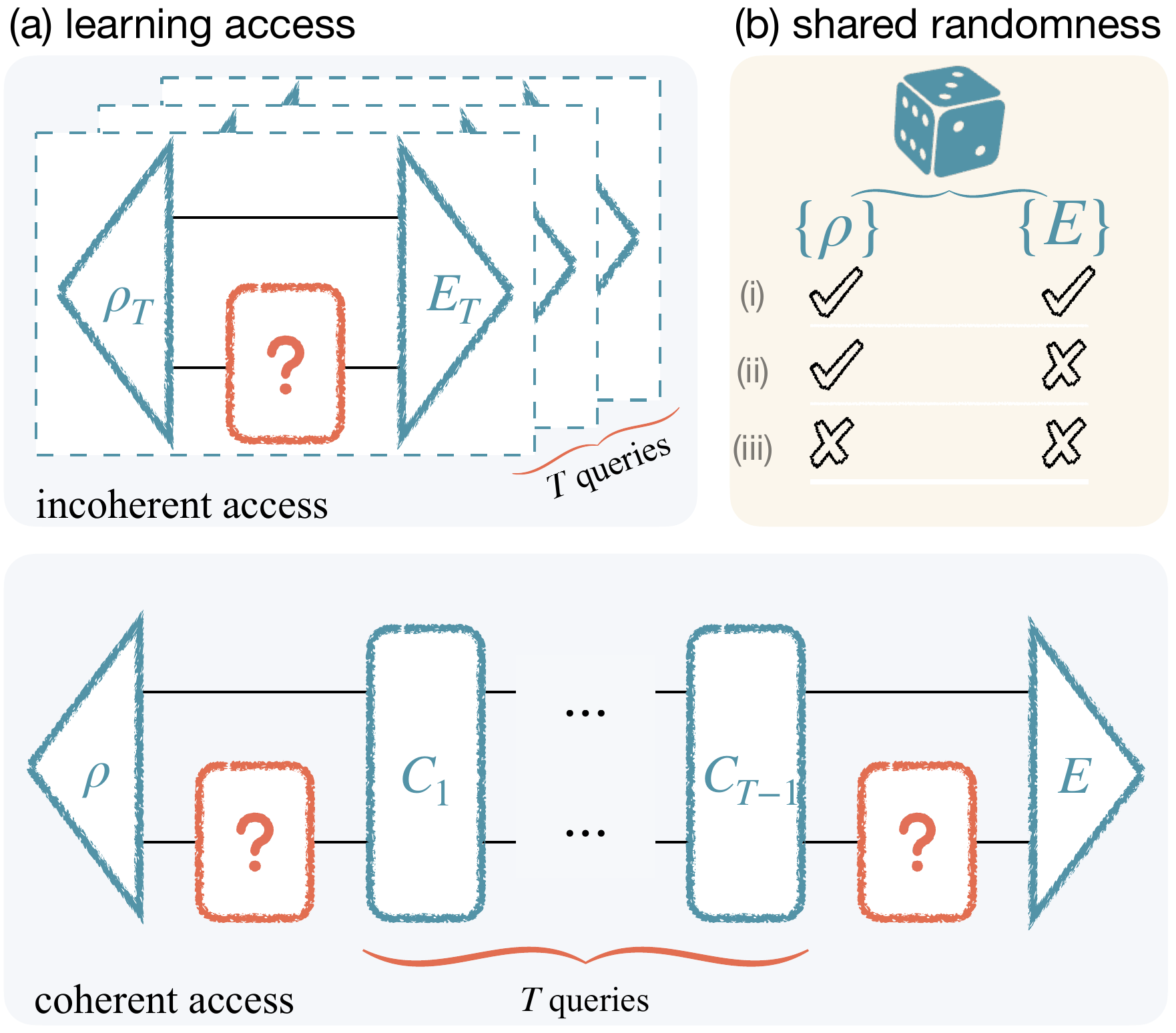}
\caption{\raggedright
Learning models for distributed channel learning.
(a) Learning access:
(i) Incoherent access: the unknown channel is queried once before each measurement.
(ii) Coherent access: the unknown channel is queried $T$ times before each measurement, with arbitrary intermediate quantum channels $\{\cC_t\}$.
(b) Shared randomness: 
(i) for both state preparation and measurement (SPAM) settings.
(ii) only for state preparation.
(iii) no shared randomness.}
\label{fig: learning models}
\end{figure}

\subsection{Haar Measure and Permutation Operators}

Let $\mu_H$ be the Haar measure over the unitary group.
The Haar random channel is defined as 
\begin{align}
\cE_H^{(k)}(A) := \bE_{U\sim\mu_H} U^{\ox k} A U^{\dagger \ox k}, \quad 
A \in \cB(\cH^{\ox k}). 
\end{align}
To describe the Haar random channel, we introduce permutation operators.
Let $\cS_k$ be the symmetric group on $k$ elements.
For $\pi\in\cS_k$, the corresponding permutation operator $P_\pi$ acts on product states as
\begin{align}
P_\pi \ket{\psi_1}\ox \cdots\ox \ket{\psi_k} 
= \ket{\psi_{\pi^{-1}(1)}} \ox \cdots \ox \ket{\psi_{\pi^{-1}(k)}}, 
\quad \forall\;\psi_1,\cdots\psi_k\in\cD(\cH).
\end{align}
With the above definition, we have the following lemma to describe the Haar random channel.

\begin{lemma}[Weingarten Calculus~\cite{mele2024introduction, schuster2025random}]
Let $\cH$ be a Hilbert space with dimension $d$, $\cS_k$ be the symmetric group, and $A\in\cB(\cH^{\ox k})$. Then, we have 
\begin{align}
\cE_H^{(k)} (A) 
= \sum_{\sigma, \tau\in \cS_k} {\rm Wg}_{\sigma, \tau}(d) \tr[AP^\dagger_\sigma] P_\tau, \quad A \in\cB(\cH^{\ox k}). 
\end{align}
Here, ${\rm Wg}_{\sigma, \tau}(d)$ are the elements of the $k!\times k!$ Weingarten matrix, which is defined as the Moore--Penrose inverse of the Gram matrix, $G_{\sigma, \tau}(d) = \tr(P_\sigma P_\tau^\dagger)$, given by the inner products of the permutation operators.
Consequently, the Choi operator of Haar random channel can be represented as 
\begin{align}
J_H^{(k)} 
= \sum_{\sigma, \tau \in\cS_k} {\rm Wg}_{\sigma, \tau}(d) P_\tau \ox P_\sigma. 
\label{eq:choi state of haar random channel}
\end{align}
\end{lemma}

For the special cases $k=1$ and $k=2$, the Haar twirl admits simple closed forms.

\begin{lemma}[Special Cases of Weingarten Calculus, $k=1$ and $k=2$ cases~\cite{mele2024introduction}]
\label{lem:special cases of weingarten calculus}
Let $\cH$ be a Hilbert space with dimension $d$, $A\in\cB(\cH)$, and $B \in\cB(\cH^{\ox 2})$, we have 
\begin{align}
\cE_H^{(1)}(A) = \frac{\tr[A]}{d} I, \quad
\cE_H^{(2)}(B) = \frac{d\tr[B] - \tr[\bF B]}{d(d^2-1)}I \ox I + \frac{d\tr[\bF B] - \tr[B]}{d(d^2-1)} \bF, 
\end{align}
where $\bF$ is the SWAP operator. 
\end{lemma}

For quantum states in $\cH$, we define the Haar measure on quantum states as~\cite{mele2024introduction}
\begin{align}
\bE_{\ket{\psi}\sim\mu_H} \ketbra{\psi}^{\ox k} 
:= \bE_{U \sim\mu_H} U^{\ox k} \ketbra{\phi}^{\ox k} U^{\dagger \ox k}
= \frac{\Pi_{\rm sym}^{(d, k)}}{\kappa_k}, \quad 
\ketbra{\phi}\in\cD(\cH), 
\end{align}
where $\kappa_k := \binom{d+k-1}{k}$ and $\Pi_{\rm sym}^{(d,k)}$ is the orthogonal projector onto the symmetric subspace of $\cH^{\ox k}$, defined as
\begin{align}
\Pi_{\rm sym}^{(d,k)} := \frac{1}{k!} \sum_{\pi\in\cS_k} P_\pi.
\end{align}
A pure-state ensemble $\cA$ is a state $k$-design if
\begin{align}
\bE_{\ket{\psi}\sim\cA} \ketbra{\psi}^{\ox k} 
= \bE_{\ket{\psi}\sim\mu_H} \ketbra{\psi}^{\ox k} 
= \frac{\Pi_{\rm sym}^{(d, k)}}{\kappa_k}.
\end{align}
Similarly, a unitary ensemble $\cA$ is said to form a unitary $k$-design if and only if 
\begin{align}
\cE_\cA^{(k)} (A)
:= \bE_{U\sim\cA} U^{\ox k} A U^{\dagger\ox k} 
= \cE_H^{(k)}(A), \quad \forall\; A\in\cB(\cH^{\ox k}). 
\end{align}
Given a state $(s+2)$-design $\{\ket{\phi_j}\}_{j=1}^{L}$, the symmetric collective measurement $\cM_s$ on $\cH^{\ox s}$ is defined as the following POVM
~\cite{grier2024sampleoptimala, li2025nearly}: 
\begin{align}
\cM_{s} 
:= \left\{\frac{\kappa_s}{L} \ketbra{\phi_j}^{\ox s} \right\}_{j=1}^L 
\cup \left\{I - \Pi_{\rm sym}^{(d, s)}\right\}.
\label{eq:symmetric collective measurement} 
\end{align}
Based on the properties of state designs, we will use the following moment identities for the symmetric collective measurement $\cM_s$. 

\begin{lemma}[Lemmas~13 and~14 in~\cite{grier2024sampleoptimala}]
\label{lem:expectation and second moment of collective measurement result}
Let $\cH$ be a Hilbert space with dimension $d$.
For the measurement $\cM_s$ on the pure state $\ketbra{\psi}^{\ox s}$, the expectation of the measurement result is 
\begin{align}
\bE \ketbra{\phi} 
= \frac{I + s\ketbra{\psi}}{d + s}. 
\end{align}
Additionally, we have 
\begin{align}
\bE \ketbra{\phi}^{\ox 2} 
= \frac{2}{(d+s)(d+s+1)}\left[(I + s\ketbra{\psi})^{\ox 2} - \frac{s(s+1)}{2}\ketbra{\psi}^{\ox 2}\right] \Pi_{\rm sym}^{(d, 2)}. 
\end{align}
\end{lemma}

\subsection{Approximate Unitary Designs}
\label{sec:approximate_unitary_designs}
A unitary ensemble $\cA$ is an $\varepsilon$-approximate unitary $k$-design if~\cite{schuster2025random}
\begin{align}
(1 - \varepsilon)\cE_{\cA}^{(k)} 
\preceq \cE_H^{(k)} 
\preceq (1 + \varepsilon)\cE_{\cA}^{(k)},  
\end{align}
where $\varepsilon$ is called relative error.
Here $\cE_1 \preceq \cE_2$ denotes that $\cE_2 - \cE_1$ is a completely-positive map.
We use the completely positive (CP) reference map~\cite{schuster2025random}
and its Choi operator, defined as
\begin{align}
\cE_a^{(k)} (A) 
:= \frac{1}{d^k}\sum_{\sigma\in\cS_k} \tr[A P_\sigma^\dagger] P_\sigma, \quad 
J_a^{(k)} 
:= \frac{1}{d^k} \sum_{\pi\in\cS_k} P_\pi\ox P_\pi. 
\label{eq:cE_a}
\end{align}
This map is generally not trace preserving.
For $k^2\le d$, the above CP-order inequalities hold with
$\cE_{\cA}^{(k)}$ replaced by $\cE_a^{(k)}$ and
\begin{align}
\varepsilon = \frac{k^2 / 2d}{1- k^2/2d}.
\end{align}
We also use the following comparison between $J_a^{(k)}$ and the Haar random Choi operator $J_H^{(k)}$.
\begin{lemma}[Lemma 2 of~\cite{chen2023unitarity}]
\label{lem:lemma 2 of chen2023unitarity}
Suppose $\frac{d^k}{d(d+1)\cdots (d+k-1)} > \frac{1}{2}$, then, we have  
\begin{align}
J^{(k)}_H 
\succeq \frac{1}{d(d+1)\cdots (d+k-1)} \sum_{\sigma\in\cS_k} P_\sigma \ox P_\sigma
= \frac{d^k}{d(d+1)\cdots(d+k-1)} J_a^{(k)}.
\end{align}
For $k^2\leq d/4$, we have~\cite{schuster2025random}
\begin{align}
J_H^{(k)}\preceq \frac{1}{1-k^2/(2d)}J_a^{(k)}.
\label{eq:haar moment upper comparison}
\end{align}
\end{lemma}

\section{Lower Bound}
\label{sec:lower_bound}

In this section, we prove lower bounds for DSEC under both incoherent and coherent access.
That is, we answer the question of how many queries are necessary to complete DSEC, even for algorithms in the strongest model considered here.
The proof architecture is as follows.
First, we show that distinguishing a shared Haar-random unitary pair $(U,U)$ from an independent pair $(U,V)$ requires $\Omega(\sqrt d)$ queries, even with coherent adaptive access.
Second, we embed this constant-gap testing problem into an $\varepsilon$-separated Choi overlap estimation problem through a random-phase channel construction.
For $T$ queries per device, averaging over the random phases expresses the output probabilities as sums of nonnegative terms with a weighted average of at most $2T\varepsilon$ factors of $U$ and $V$.
Comparing these terms in the shared and independent cases gives $T=\Omega(\sqrt d/\varepsilon)$.
Third, for two suitably chosen noisy channels, one Choi state exactly simulates each query, so even adaptive protocols inherit the independent $\Omega(1/\varepsilon^2)$ lower bound.

We first consider two unitary channels $\cU = U(\cdot)U^\dagger$ and $\cV = V(\cdot)V^\dagger$ acting on $d$-dimensional subsystems of a Hilbert space $\cH\simeq \cH_{\rm main} \ox \cH_{\rm aux}$, 
where $\cH_{\rm main}$ is the main system with dimension $d$ and 
$\cH_{\rm aux}$ is an auxiliary system with dimension $d'$. 

Our approach is to reduce the DSEC problem to a \emph{two-hypothesis distinguishing problem}. In particular, to derive a lower bound for completing DSEC, we consider the following distinguishing problem:

\begin{problem}[Distinguishing Problem]
\label{pro:distinguish problem}
We want to distinguish the following two cases:
\begin{enumerate}
\item Two quantum devices perform the same unitary $U$, which is a Haar random unitary;
\item Two quantum devices independently perform two unitaries $U$ and $V$, which are independent Haar random unitaries. 
\end{enumerate}
\end{problem}

In Case~1, both the diamond distance and the trace distance between normalized Choi states are zero; in Case~2, both are bounded below by a universal positive constant with high probability.
Consequently, an algorithm that estimates either distance to sufficiently small constant additive error, or tests zero distance against a constant separation, also solves Problem~\ref{pro:distinguish problem}, so the lower bounds below apply to these tasks as well.

Define the Choi operators of $U$ and $V$ as $J_U$ and $J_V$, respectively.
To bound the success probability of solving the distinguishing problem, we invoke Le Cam's two-point method as follows.

\begin{lemma}[Le Cam's Two-Point Method~\cite{Yu1997}]
\label{lem:Le Cam's two-point method}
Let $p^{U,U}(\ell)$ and $p^{U,V}(\ell)$ be the probabilities of obtaining measurement outcome $\ell$ under Cases~1 and~2 of Problem~\ref{pro:distinguish problem}, respectively.
Then, for any binary test between the two averaged outcome distributions, the average success probability is at most
\begin{align}
\frac{1}{2} + \frac{1}{2}\norm{\bE_{U\sim \mu_H} p^{U, U} - \bE_{U,V\sim \mu_H} p^{U, V}}{{\rm TV}},
\end{align}
where $\norm{\cdot}{{\rm TV}}$ is the total variation distance (TVD), defined as
\begin{align}
\norm{\bE_{U\sim \mu_H} p^{U, U} - \bE_{U,V\sim \mu_H} p^{U, V}}{{\rm TV}} := 
\frac{1}{2}\sum_{\ell} \left\vert \bE_{U\sim \mu_H} p^{U, U}(\ell) - \bE_{U,V\sim \mu_H} p^{U, V}(\ell) \right\vert. 
\end{align}
\end{lemma}
Therefore, for reliable hypothesis testing, this TVD must be at least \emph{constantly} large. 
The rest of the section can therefore be read as a sequence of indistinguishability bounds. 
In the coherent model, we first represent an arbitrary coherent algorithm in Choi-tester form and then compare the resulting outcome distributions using symmetric-subspace identities.

\subsection{Incoherent Access}

Here, we prove the lower bound for completing DSEC with incoherent access. 
In this model the learner never stores a quantum output for later coherent processing; all memory is classical.
First, we introduce the tree representation of channel learning with incoherent access, which is a powerful tool in proving lower bounds for learning tasks~\cite{chen2022exponential, chen2023unitarity, chen2024optimal, hu2025ansatzfreea, chen2025efficient}. 
Briefly, the learning tree can represent any adaptive incoherent algorithm. 
Each root-to-leaf path records all classical information available to the learner, including the outcomes that determine future state preparations and measurements. 
Thus bounding the leaf distribution under different channel cases is equivalent to bounding the statistical information that the algorithm can extract from its queries.
The formal definition of the learning tree is as follows.

\begin{definition}[Tree Representation for Learning Quantum Channels~\cite{chen2022exponential}]
\label{def:learning tree}
Consider two channels $\cU$ and $\cV$ acting on $\cH_{\rm main}$.
A learning algorithm with incoherent access can be represented as a rooted tree $\cT$ of depth $T$ such that each node encodes all measurement outcomes the algorithm has received thus far. The tree has the following properties: 
\begin{itemize}
\item Each node $u$ has an associated probability $p^{U, V}(u)$. 
\item The root of the tree $r$ has an associated probability $p^{U, V}(r) = 1$. 
\item At each non-leaf node $u$, we prepare a state $\ket{\phi_{u,1}}\ox\ket{\phi_{u,2}}$ on $\cH^{\ox 2}$, apply channels $\cU$ and $\cV$ onto two $d$-dimensional subsystems, and measure a rank-$1$ POVM $\{w_v^u d^2 d'^2\ketbra{\psi^{u}_{v,1}} \ox \ketbra{\psi^{u}_{v,2}}\}_v$ (which can depend on $u$) on the entire system to obtain a classical outcome $v$\footnote{As shown in~\cite[Remark 4.19]{chen2022exponential}, we only need to consider rank-$1$ POVMs on two quantum systems.}. 
Each child node $v$ of the node $u$ corresponds to a particular POVM outcome $v$ and is connected by the edge $e_{u,v}$. 
We denote the set of child nodes of node $u$ as ${\rm child}(u)$. 
\item If $v$ is a child node of $u$, then 
\begin{align}
p^{U, V}(v) 
= p^{U, V}(u) w_v^u d^2 d'^2 \tr\left[\left(\bigotimes_{i=1}^2 \ketbra{\psi^{u}_{v,i}}\right) (\cU\ox \cI_{\rm aux} \ox \cV\ox \cI_{\rm aux}) \left(\bigotimes_{i=1}^2 \ketbra{\phi_{u,i}}\right)\right]. \notag 
\end{align}
\item Each root-to-leaf path is of length $T$. 
For a leaf corresponding to node $\ell$, $p^{U, V}(\ell)$ is the probability that the classical memory is in state $\ell$ after the learning procedure. 
\end{itemize}
\end{definition}

Based on the above tree representation and distinguishing problem, we are ready to prove the lower bound for completing DSEC with incoherent access.
We introduce the intermediate leaf distribution $p^\cD(\ell)$ obtained when both channels are replaced by the completely depolarizing channel $\cD(A):=\tr[A]I/d$. 
Then, by the triangle inequality, we have the following upper bound
\begin{align}
\norm{\bE_{U\sim \mu_H} p^{U, U} - \bE_{U,V\sim \mu_H} p^{U, V}}{{\rm TV}} 
&\leq \norm{\bE_{U\sim \mu_H} p^{U, U} - p^\cD}{{\rm TV}} + \norm{p^\cD - \bE_{U,V\sim \mu_H} p^{U, V}}{{\rm TV}}, 
\label{eq:incoherent lower bound two terms}
\end{align}
Thus, we can focus on analyzing the two terms in the upper bound of TVD in Eq.~\eqref{eq:incoherent lower bound two terms} separately.

\begin{lemma}
\label{lem:upper bound for two terms in incoherent}
For the two terms defined in Eq.~\eqref{eq:incoherent lower bound two terms}, we have the following upper bounds, 
\begin{align}
\norm{\bE_{U\sim \mu_H} p^{U, U} - p^\cD}{{\rm TV}}
&\leq \frac{4 T^2}{d} , \quad 
\norm{p^\cD - \bE_{U,V\sim \mu_H} p^{U, V}}{{\rm TV}}
\leq \frac{2T^2}{d}.
\end{align}
\end{lemma}

The proof of Lemma~\ref{lem:upper bound for two terms in incoherent} is deferred to Appendix~\ref{app:upper bound for two terms in incoherent lower bound}.
The lemma shows that each case performs similarly to the depolarizing channel unless $T=\Omega(\sqrt d)$; equivalently, constant total variation distance requires $\Omega(\sqrt d)$ queries.
This immediately implies the following lower bound for completing DSEC with incoherent access.

\begin{theorem}
\label{the:incoherent lower bound}
Any algorithm with incoherent access, possibly utilizing multi-round LOCC, ancillary systems, and adaptive strategies, requires $\Omega(\sqrt{d})$ queries to complete Problem~\ref{pro:distinguish problem}. 
\end{theorem}
\begin{proof}[Proof of Theorem~\ref{the:incoherent lower bound}]
Applying Lemma~\ref{lem:upper bound for two terms in incoherent} to the two terms in Eq.~\eqref{eq:incoherent lower bound two terms} gives
\begin{align}
\norm{\bE_{U\sim \mu_H} p^{U, U} - \bE_{U,V\sim \mu_H} p^{U, V}}{{\rm TV}}
&\leq \frac{6 T^2}{d},
\end{align}
which implies that we require $T=\Omega(\sqrt{d})$. 
Therefore, we complete the proof of Theorem~\ref{the:incoherent lower bound}. 
\end{proof}

Because the tree representation already captures adaptive multi-round LOCC algorithms with ancillary systems, Theorem~\ref{the:incoherent lower bound} gives an $\Omega(\sqrt d)$ incoherent lower bound for DSEC.
The scaling agrees with unitarity estimation~\cite{chen2023unitarity}, which is the self-overlap special case.

\subsection{Coherent Access}
\label{sec:coherent lower bound}

We now turn to coherent access, where the learner may interleave multiple queries with arbitrary quantum channels before measurement. 
In various quantum learning tasks, coherent access often provides substantial advantages over incoherent access~\cite{chen2022exponential, chen2023unitarity, li2025nearly, jeon2025query}. 
A central question is whether such advantages exist for DSEC.
Remarkably, we prove that even with this significantly more powerful learning access, completing Problem~\ref{pro:distinguish problem} still requires $\Omega(\sqrt{d})$ queries. 
This result holds even with multi-round LOCC, arbitrarily large ancillary systems, and adaptive SPAM settings, highlighting the fundamental hardness of DSEC. 

\paragraph*{Characterization of coherent access.}
It is nontrivial to obtain this lower bound.
One main technical difficulty is to represent an arbitrary coherent algorithm, since such an algorithm may interleave channel queries with general quantum processing. 
Following the tester representation of~\cite{schuster2025random}, the outcome probabilities of a coherent channel learning algorithm can be expressed as traces involving the $T$-fold Choi operator and positive semidefinite operators determined by the algorithm.
Equivalently, all intermediate processing and ancillas are absorbed into these operators, and the unknown unitary appears only through its $T$-fold Choi operator.

Consider two unknown unitary channels $\cU$ and $\cV$ acting on two separate quantum systems.
Condition on shared randomness and fix a complete classical transcript $\ell$ of all measurement outcomes and communicated messages. 
Every adaptive control is then fixed.
On these devices, we input the initial state $\ket{0}\otimes\ket{0}$ on $\cH^{\otimes2}$ and interleave the channel queries with the local operations fixed by $\ell$.
Assuming the auxiliary system is large enough, each local data-processing channel admits a unitary representation~\cite{nielsen2010quantum}. 
For fixed $\ell$, $W_{t,i}^{(\ell)}$ combines these unitaries with projections onto the recorded intermediate outcomes, without renormalization, to retain their probabilities. The adaptive controls use only earlier outcomes and messages in $\ell$. Thus $W_{t,i}^{(\ell)}$ need not be unitary and is independent of $U,V$. 
After $T$ queries, the resulting states on the two devices can then be expressed as
\begin{align}
\ket{\psi_1} 
&= \prod_{t=1}^T \left[\left(U\ox I_{\rm aux} \right) W_{t,1}^{(\ell)}\right] \ket{0}, \quad 
\ket{\psi_2}
= \prod_{t=1}^T \left[\left(V\ox I_{\rm aux} \right) W_{t,2}^{(\ell)}\right] \ket{0}.
\end{align}
These two terms can also be written as 
\begin{align}
\ket{\psi_1} 
&= \left(I_{\cH}\ox\bra{\Phi}\right) \left(I_{\cH} \ox U^{\ox T}\ox I_{\rm main}^{\ox T}\right) \ket{\Psi_{I, 1}}, \\
\ket{\psi_2} 
&= \left(I_{\cH}\ox\bra{\Phi}\right) \left(I_{\cH} \ox V^{\ox T}\ox I_{\rm main}^{\ox T}\right)  \ket{\Psi_{I, 2}}, 
\end{align}
where $\ket{\Phi} = \sum_{i=0}^{d^T-1} \ket{ii}$ is the maximally entangled state on system $\cH_{\rm main}^{\ox T} \ox \cH_{\rm main}^{\ox T}$ and 
\begin{align}
\ket{\Psi_{I, i}} 
= \sum_{\substack{x_1,\cdots,x_T \in \{0,\cdots, d-1\} \\ y_1,\cdots,y_T \in \{0,\cdots, d-1\}}} \prod_{t=1}^T \overbrace{\left[\left(\ket{y_t}\!\bra{x_t}\ox I_{\rm aux}\right)  W_{t, i}^{(\ell)}\right] \ket{0}}^{\text{$\cH$}} \ox \overbrace{\ket{x_1\cdots x_T} \ox \ket{y_1\cdots y_T}}^{\text{$\cH_{\rm main}^{\ox T}\ox \cH_{\rm main}^{\ox T}$}}. 
\label{eq:definition of Psi_I}
\end{align}
The states $\ket{\psi_i}$ and $\ket{\Psi_{I,i}}$ inherit this dependence on $\ell$, which is omitted from their notation.
Lastly, let $E_{\ell,1}\ox E_{\ell,2}$ be the final POVM element on this branch and write $\Psi_{I,i}:=\ketbra{\Psi_{I,i}}$.
The joint state on this fixed branch is $\ket{\psi_1}\ox\ket{\psi_2}$.
Both states and measurements may depend on $\ell$.
This representation gives the following two outcome probabilities:
\begin{align}
p^{U,V}(\ell) 
&:= \tr\left[\left(E_{\ell, 1}\ox J_{U^\dagger}^{(T)}\right)\Psi_{I,1}\right] \tr\left[\left(E_{\ell, 2}\ox J_{V^\dagger}^{(T)}\right)\Psi_{I,2}\right], \\
p^{U,U}(\ell)
&:= \tr\left[\left(E_{\ell, 1}\ox J_{U^\dagger}^{(T)}\right)\Psi_{I,1}\right] \tr\left[\left(E_{\ell, 2}\ox J_{U^\dagger}^{(T)}\right)\Psi_{I,2}\right], 
\end{align}
where $J_U^{(T)}$ is the Choi operator of unitary $U^{\ox T}$, as defined in Eq.~\eqref{eq:choi state}, and similarly for $J_V^{(T)}$. 
Haar invariance under taking adjoints permits relabeling $U^\dagger,V^\dagger$ as $U,V$ in the averaged expressions below.
With the definition of Haar random channel, we have 
\begin{align}
\bE_{U,V\sim\mu_H} p^{U,V}(\ell) 
= \tr\left[\left(E_{\ell, 1}\ox J_H^{(T)}\right)\Psi_{I,1}\right] \tr\left[\left(E_{\ell, 2}\ox J_H^{(T)}\right)\Psi_{I,2}\right], 
\end{align}
where $J_H^{(T)}$ is the Choi operator of the Haar random channel, as defined in Eq.~\eqref{eq:choi state of haar random channel}. 
Consequently, with the above characterization, we can write two probability distributions as a function of Choi operators without loss of generality, which is more convenient for analyzing the TVD between them.

\paragraph*{Comparison with a CP reference map.}
The second challenge is how to analyze the TVD between $\bE_{U\sim \mu_H} p^{U,U}(\ell)$ and $\bE_{U,V\sim \mu_H} p^{U,V}(\ell)$, which is hard to analyze directly.
It suffices to consider $T^2\leq d/16$, since larger $T$ already gives the claimed lower bound.
To overcome this challenge, we introduce nonnegative reference weights $p^a(\ell)$, defined as
\begin{align}
p^a(\ell) 
:= \tr\left[\left(E_{\ell, 1}\ox J_a^{(T)}\right)\Psi_{I,1}\right] \tr\left[\left(E_{\ell, 2}\ox J_a^{(T)}\right)\Psi_{I,2}\right]. 
\end{align}
These weights need not sum to one. The two actual outcome distributions are normalized, so
\begin{align}
&\norm{\bE_{U\sim \mu_H} p^{U,U}-\bE_{U,V\sim \mu_H} p^{U,V}}{{\rm TV}}
=\sum_{\ell:\,\bE_{U,V}p^{U,V}(\ell)>\bE_Up^{U,U}(\ell)}
\left[\bE_{U,V}p^{U,V}(\ell)-\bE_Up^{U,U}(\ell)\right].
\label{eq:coherent lower bound two terms}
\end{align}
We next bound this with the following lemma.

\begin{lemma}
\label{lem:upper bound for two terms in coherent}
If $T^2\leq d/16$, we have
\begin{align}
\bE_{U\sim\mu_H}p^{U,U}(\ell)
\geq\left(1-\frac{4T^2}{d}\right) \bE_{U,V\sim\mu_H}p^{U,V}(\ell).
\label{eq:upper bound for coherent lower bound first term}
\end{align}
\end{lemma}
\begin{proof}[Proof of Lemma~\ref{lem:upper bound for two terms in coherent}]
Define the local positive semidefinite operators
\begin{align}
M_{\ell_i}
&:=\tr_1\left[\left(E_{\ell,i}\ox I\right)\Psi_{I,i}\right]
=\tr_1\left[(E_{\ell,i}^{1/2}\ox I)\Psi_{I,i}(E_{\ell,i}^{1/2}\ox I)\right]\succeq0,\quad i=1,2.
\end{align}
By the definition of the partial trace,
\begin{align*}
\tr[M_{\ell_i}J_U^{(T)}]
=\tr[(E_{\ell,i}\ox J_U^{(T)})\Psi_{I,i}].
\end{align*}
The same identity holds with $J_a^{(T)}$ in place of $J_U^{(T)}$, so
\begin{align}
p^a(\ell)=\tr[M_{\ell_1}J_a^{(T)}]\tr[M_{\ell_2}J_a^{(T)}].
\end{align}
Applying Lemma~\ref{lem:choi state cloning} gives
\begin{align}
\bE_{U\sim\mu_H}p^{U,U}(\ell)
&= \bE_{U\sim\mu_H} \tr\left[M_{\ell_1} J_U^{(T)}\right] \tr\left[M_{\ell_2} J_U^{(T)}\right]\notag \\
&\geq\frac{d^{2T}}{d(d+1)\cdots(d+2T-1)}
\tr[M_{\ell_1}J_a^{(T)}]\tr[M_{\ell_2}J_a^{(T)}]\notag\\
&=\frac{d^{2T}}{d(d+1)\cdots(d+2T-1)} p^a(\ell).
\end{align}
Applying Eq.~\eqref{eq:haar moment upper comparison} gives
\begin{align}
\bE_{U,V\sim\mu_H}p^{U,V}(\ell)
\leq\frac{p^a(\ell)}{(1-T^2/(2d))^2}.
\label{eq:upper bound for coherent lower bound second term}
\end{align}
Combining the preceding two bounds eliminates the reference weight $p^a(\ell)$, and we have
\begin{align}
\frac{d^{2T} (1-T^2/(2d))^2}{d(d+1)\cdots(d+2T-1)}
&\geq\left(1-\frac{2T^2}{d}\right)\left(1-\frac{T^2}{d}\right)
\geq1-\frac{4T^2}{d}.
\end{align}
Here we used $(1+x)^{-1}\geq1-x$ and $\prod_j(1-x_j)\geq1-\sum_jx_j$ for $0\leq x_j\leq1$.
\end{proof}

Consequently, based on the above analysis, we have the following theorem for the lower bound of completing Problem~\ref{pro:distinguish problem} with coherent access.

\begin{theorem}
\label{the:coherent lower bound}
Any algorithm with coherent access, possibly utilizing multi-round LOCC, ancillary systems, and adaptive strategies, requires $\Omega(\sqrt{d})$ queries to complete Problem~\ref{pro:distinguish problem}. 
\end{theorem}
\begin{proof}[Proof of Theorem~\ref{the:coherent lower bound}]
Using Eq.~\eqref{eq:coherent lower bound two terms} and Lemma~\ref{lem:upper bound for two terms in coherent}, we have
\begin{align}
\norm{\bE_{U,V\sim\mu_H}p^{U,V}-\bE_{U\sim\mu_H}p^{U,U}}{\rm TV}
\leq\frac{4T^2}{d},
\end{align}
which implies that we require $T = \Omega(\sqrt{d})$.
For any multi-round LOCC protocol, the representation above ensures that
Eq.~\eqref{eq:upper bound for coherent lower bound first term}
holds for each complete classical transcript $\ell$.
Since the transcript distributions are normalized, the same TVD bound applies.
\end{proof}

Likewise, this theorem implies that any algorithm with coherent access requires at least $\Omega(\sqrt{d})$ queries to complete DSEC.
Importantly, we will show that coherent access offers no asymptotic advantage for completing DSEC: our incoherent algorithm matches this lower bound.

\subsection{Lower Bound for DSEC}
\label{sec:lower bound for DSEC}

We now lift the constant-error lower bound to additive error $\varepsilon$. 
The result is shown in the following theorem.

\begin{theorem}
\label{the:lower bound for DSEC}
Any algorithm, possibly utilizing multi-round LOCC, ancillary systems, coherent access, and adaptive strategies, requires $\Omega(\max\{\sqrt{d}/\varepsilon, 1/\varepsilon^2\})$ queries per device to complete DSEC, for $d\geq2$ and $0<\varepsilon\leq\varepsilon_0$, where $\varepsilon_0>0$ is a sufficiently small universal constant. 
\end{theorem}
\begin{proof}
The lower bound has two independent components. 

\paragraph*{The $\varepsilon$-dependent lower bound.}
We use a Choi-state distinguishing problem following the proof of Theorem~5 in~\cite{chen2023unitarity}.
Fix the generalized Pauli phase operator $Z=\sum_{j=0}^{d-1}e^{2\pi i j/d}\ketbra{j}$, and, 
for $p\in[0,1]$, set $\cE_p(X)=(1-p)X+pZXZ^\dagger$ and $\cF=\cI$.
Let $\ket{\phi_0}=\ket{\Phi}/\sqrt d$ be the normalized maximally entangled state and $\ket{\phi_1}=(Z\otimes I)\ket{\phi_0}$.
These two states are orthogonal, so the normalized Choi state and the target overlap are
\begin{align}
\rho_p:=\frac{J_{\cE_p}}d
&=(1-p)\ketbra{\phi_0}+p\ketbra{\phi_1},\qquad
\frac{\tr[J_{\cE_p}J_{\cI}]}{d^2}=1-p.
\label{eq:accuracy hard choi}
\end{align}
For $0<\varepsilon\leq1/8$, take $p_\pm=1/2\pm2\varepsilon$ and write $\rho_\pm=\rho_{p_\pm}$.
Consider the following two-hypothesis problem.
\begin{problem}[Distinguishing two noisy channels]
\label{pro:accuracy distinguishing}
The two devices are promised to implement one of the following channel pairs, with equal prior probabilities:
\begin{enumerate}
\item $(\cE,\cF)=(\cE_{p_+},\cI)$;
\item $(\cE,\cF)=(\cE_{p_-},\cI)$.
\end{enumerate}
Using at most $T$ queries per device, decide which case holds with success probability at least $2/3$.
\end{problem}
The two target overlaps are $1/2\mp2\varepsilon$, separated by $4\varepsilon$.
Thus an additive-$\varepsilon$ DSEC estimator with success probability at least $2/3$ solves Problem~\ref{pro:accuracy distinguishing} by thresholding its estimate at $1/2$, using the same queries.

These Pauli channels are \emph{teleportation-covariant}, so they can be simulated by teleporting the input state over their Choi states~\cite[Proposition~2]{pirandola2017fundamental}.
Their teleportation corrections are the same for all $p$.
For this family, the Choi-state simulation identity can also be verified explicitly as follows.
Let $\sigma$ be any joint state of the channel input and an arbitrary reference system.
One copy of $\rho_p$ exactly simulates one query to $\cE_p$ by applying the fixed unitary $W=(I-\ketbra{\phi_1})\otimes I+\ketbra{\phi_1}\otimes Z$
to the Choi register and the channel input, then discarding the Choi register.
This operation is independent of $p$ and leaves the reference unchanged:
\begin{equation}
\begin{quantikz}[row sep=0.4em, column sep=0.8em]
\lstick[2]{\(\rho_p\)}\qw & \gate[3]{W} & \qw & \qw\rstick[wires=2]{trace out}\\
\qw & & \qw & \qw\\
\lstick[2]{\(\sigma\)}\qw & \qw & \qw & \qw\\
\qw & \qw & \qw & \qw
\end{quantikz}
\quad=\quad
\begin{quantikz}[row sep=0.8em, column sep=0.8em]
\lstick[2]{\(\sigma\)}\qw & \gate{\cE_p} & \qw & \qw\\
\qw & \qw & \qw & \qw
\end{quantikz}.
\label{eq:accuracy circuit simulation}
\end{equation}
The two upper wires on the left contain $\rho_p$; the bottom wire on each side carries the reference.
Orthogonality of $\phi_0,\phi_1$ gives the exact simulation identity
\begin{align}
&\tr_{\mathrm{Choi}}\!\left[(W\otimes I)(\rho_p\otimes\sigma)(W^\dagger\otimes I)\right]
=(1-p)\sigma+p(Z\otimes I)\sigma(Z^\dagger\otimes I)
=(\cE_p\otimes\cI)(\sigma),
\label{eq:accuracy exact simulation}
\end{align}
where the additional identity factors act on the reference and the partial trace discards the Choi register.
Replacing every call to $\cE_{p_\pm}$ in a protocol for Problem~\ref{pro:accuracy distinguishing} by this fixed simulation gives a quantum measurement on $\rho_\pm^{\otimes T}$; the calls to $\cF=\cI$ are known operations.
This replacement preserves the full adaptive experiment, including its quantum memories and classical messages.
The two Choi states satisfy
\begin{align*}
\norm{\sqrt{\rho_+}\sqrt{\rho_-}}1^2
=\left(\sqrt{(1-p_+)(1-p_-)}+\sqrt{p_+p_-}\right)^2
=1-16\varepsilon^2.
\end{align*}
The Helstrom bound for success probability at least $2/3$, together with the fidelity--trace-distance inequality~\cite{nielsen2010quantum}, gives
\begin{align}
\frac13
\leq\frac12\norm{\rho_+^{\otimes T}-\rho_-^{\otimes T}}1
\leq\sqrt{1-(1-16\varepsilon^2)^T}
\leq4\varepsilon\sqrt T.
\label{eq:accuracy choi distance}
\end{align}
Thus $T\geq1/(144\varepsilon^2)$, proving the $\Omega(\varepsilon^{-2})$ lower bound for general channels.
It remains to prove the dimension-dependent lower bound $\Omega(\sqrt d/\varepsilon)$.

\paragraph*{Distinguishing problem for estimation.}
We enlarge the Hilbert space as $\cH=\cH_A \ox \cH_B$ for convenience, which does not change the form of our lower bound, where $\cH_{A} \cong \bC^2$ and $\cH_{B} \cong \bC^d$. 
The constructed channel has total dimension $2d$, so the final lower bound in the original dimension changes only by a universal constant factor.
For $\theta\in[0,2\pi]$ and a unitary $U$ on $\cH_B$, define two Kraus operators on $\cH$ by
\begin{align}
E_{\theta,U}
&:= \sqrt{1-\varepsilon}e^{i\theta}\left(\ketbra{0}_A\ox I_B\right)
+ \sqrt{\varepsilon}\left(\ket{1}\!\bra{0}_A\ox U_B\right), \quad 
F :=\ketbra{1}_A\ox I_B,
\end{align}
and define the channel
\begin{align}
\cE_{\theta,U}(X)
:=E_{\theta,U}X E_{\theta,U}^\dagger+F X F^\dagger .
\end{align}
This is a CPTP map from $\cB(\cH)$ to $\cB(\cH)$ because
\begin{align}
E_{\theta,U}^\dagger E_{\theta,U}=\ketbra{0}_A\ox I_B,
\quad
F^\dagger F=\ketbra{1}_A\ox I_B, 
\quad \Rightarrow\quad
E_{\theta,U}^\dagger E_{\theta,U}+F^\dagger F= I_{AB}.
\end{align}
We consider the following distinguishing problem.
\begin{problem}[Distinguishing Problem for Estimation]
\label{pro:distinguish problem 2}
Given $T$ queries to each of two unknown channels on $\cH=\cH_A \ox \cH_B$, distinguish the following two cases:
\begin{enumerate}
\item The two devices implement $\cE_{\theta_1,U}$ and $\cE_{\theta_2,U}$, respectively, where $U$ is Haar-random on $\cH_B$ and $\theta_1,\theta_2$ are independent uniform phases in $[0,2\pi]$.
\item The two devices implement $\cE_{\theta_1,U}$ and $\cE_{\theta_2,V}$, respectively, where $U,V$ are independent Haar-random unitaries on $\cH_B$ and $\theta_1,\theta_2$ are independent uniform phases in $[0,2\pi]$.
\end{enumerate}
The goal is to decide which case holds with success probability at least $2/3$.
\end{problem}

Assume, toward a reduction, that an algorithm $\bA$ estimates $\tr[J_\cE J_\cF]/4d^2$ to additive error $\varepsilon/100$ with probability at least $0.99$ using $T$ queries per device and $\varepsilon\leq 0.01$.
We show that $\bA$ solves Problem~\ref{pro:distinguish problem 2} with constant success probability.
Let $J_{\theta,U}$ denote the Choi operator of the channel $\cE_{\theta,U}$.
For the present construction, the normalized Choi overlap is
\begin{align}
\frac{\tr[J_{\theta_1,U}J_{\theta_2,V}]}{4d^2}
= \frac{1}{4}\left[
1+\left|(1-\varepsilon)e^{i(\theta_2-\theta_1)}+\varepsilon \frac{\tr(U^\dagger V)}{d}\right|^2
\right].
\end{align}
In Case~1 of Problem~\ref{pro:distinguish problem 2}, we have
\begin{align}
f_1
:= \bE_{U} \frac{\tr[J_{\theta_1,U}J_{\theta_2,U}]}{4d^2}
=
\frac{1}{2}\left[
1 - \varepsilon + \varepsilon^2 + \varepsilon(1-\varepsilon) \cos(\theta_1 - \theta_2)
\right].
\end{align}
In Case~2, let $f_2:=\tr[J_{\theta_1,U}J_{\theta_2,V}]/(4d^2)$ denote the normalized Choi overlap of the sampled channel pair. We have
\begin{align}
\bE_{U,V}\left|\frac{\tr(U^\dagger V)}{d}\right|^2
=
\frac{1}{d^2}, \quad\Rightarrow\quad
\bE_{U,V}f_2
= \frac{1}{4} + \frac{1}{4}(1-\varepsilon)^2 + \frac{\varepsilon^2}{4d^2}.
\end{align}
Therefore, in Case~1, the overlap fluctuates as $(1-\varepsilon+\varepsilon\cos\gamma)/2$ up to lower-order terms for a uniform phase $\gamma$, while in Case~2 it concentrates near $[1 + (1-\varepsilon)^2]/4$.
Based on this gap, we can design a distinguishing procedure with the estimation algorithm $\bA$, as shown in Lemma~\ref{lem:solve problem2 with algorithm A} for sufficiently large $d$.
Therefore, the above estimator gives a $T$-query algorithm
for solving Problem~\ref{pro:distinguish problem 2}.
Constant-factor rescaling of the estimation accuracy and
constant-repetition amplification from success probability
$2/3$ to $0.99$ preserve the asymptotic lower bound.
It remains to lower bound the query complexity of Problem~\ref{pro:distinguish problem 2}.

\paragraph*{Phase averaging in Problem~\ref{pro:distinguish problem 2}.}
We first expand the phase-averaged Choi operator in Problem~\ref{pro:distinguish problem 2}.
In Section~\ref{sec:coherent lower bound}, we show that for any coherent $T$-query algorithm, all adaptive operations and ancillas can be absorbed into a fixed tester acting on the $T$-fold Choi operator of the unknown channel.
Let
\begin{align}
\ket{E_{\theta,U}}:=(E_{\theta,U}\ox I_{A'B'})\ket{\Phi_{2d}}, 
\quad
\ket{F}:=(F\ox I_{A'B'})\ket{\Phi_{2d}}.
\end{align}
where $\ket{\Phi_{2d}}$ is the unnormalized maximally entangled state on $\cH\ox\cH'$, $\cH' = \cH_{A'} \ox \cH_{B'}$, $\cH_{A'} \cong \bC^{2}$, and $\cH_{B'} \cong \bC^d$.
Then
\begin{align}
J_{\theta,U}=\ketbra{E_{\theta,U}}+\ketbra{F}.
\end{align}
Expanding $J_{\theta,U}^{\ox T}$ gives a sum over the subset $R\subseteq[T]$ of query positions in which the Kraus operator $E_{\theta,U}$ appears; the complementary positions contain only the known vector $\ket{F}$.
The phase average removes cross terms with different active counts of $U$, and we have 
\begin{align}
\bE_{\theta}
\bigotimes_{r\in R}\ketbra{E_{\theta,U}}_r
=
\sum_{t=0}^{|R|}
\binom{|R|}{t}(1-\varepsilon)^{|R|-t}\varepsilon^t
\ketbra{\Phi_{U,R,t}},
\label{eq:phase average same dimensional}
\end{align}
where
\begin{align}
\ket{\Phi_{U,R,t}}
:=
\ket{0}_{A', R} \ox \frac{1}{\sqrt{\binom{|R|}{t}}}
\sum_{\substack{S\subseteq R\\ |S|=t}}
\left( \ket{0}_{A, R\setminus S}\ox \ket{1}_{A, S} \right)
\ox \left(U_{B, S}\ox I_{B'}^{\ox |R|} \ket{\Phi_{d}}^{\ox |R|} \right).
\label{eq:definition of phi_U,R,t}
\end{align}
Here $U_{B, S}$ means applying $U$ on the $\cH_B^{\ox |R|}$ indexed by $S$ and identity on the other active systems.
Therefore, the averaged Choi state of one device can be written as 
\begin{align}
\frac{1}{(2d)^T}\bE_\theta J_{\theta, U}^{\ox T} 
&= \frac{1}{2^T}\sum_{R\subseteq[T]}
\sum_{t=0}^{|R|}
\binom{|R|}{t}(1-\varepsilon)^{|R|-t}\varepsilon^t
\rho_{U,R,t}, 
\label{eq:normalized phase averaged input}
\end{align}
where 
\begin{align}
\rho_{U,R,t}
:= \frac{1}{d^T} J_{U,R,t}
= \frac{1}{d^T} \ketbra{\Phi_{U,R,t}}
\ox \ketbra{F}_{\bar{R}}.
\end{align}
Indeed, the binomial weight records the number of informative $U$ branches inside the positions in $R$.

The active normalized state $d^{-|R|/2}\ket{\Phi_{U,R,t}}$ can be prepared using exactly $t$ queries to the Haar-random unitary $U$ by Lemma~\ref{lem:simulation of phi_U,t}, with the inactive normalized known branch $d^{-(T-|R|)/2}\ket{F}_{\bar R}$ appended independently.

\paragraph*{Lower bound for estimation.}
Suppose that an algorithm $\bA'$ solves Problem~\ref{pro:distinguish problem 2} with success probability at least $2/3$ using at most $T$ queries per device.
We bound the total variation distance between its output distributions in the two cases.
For fixed $U,V$, let $p^{U,V}(\ell)$ denote the probability of a complete classical record $\ell$, including the final decision, averaged over the independent phases $\theta_1,\theta_2$.
The two distributions to be compared are therefore $\bE_U p^{U,U}$ and $\bE_{U,V}p^{U,V}$.

To expand $p^{U,V}(\ell)$, we first describe the branches contributing to this record.
Recall that each channel defined above has two Kraus operators, $E_{\theta,U}$ and $F$, which we label $E$ and $F$, respectively.
We refine the record $\ell$ by appending the Kraus label at every channel call, in the order of the calls.
Write $\lambda$ for the resulting refined branch.
We also include the local instrument Kraus labels and the shared-random seed in $\lambda$.
For example, displaying only the channel Kraus labels, we may write
\begin{align*}
\lambda
= (\ell,\, \underbrace{E,F,\ldots,F}_{\text{sequence of Kraus labels}}).
\end{align*}
The notation $\lambda\mapsto\ell$ means retaining the observed record and discarding the additional labels; in this example, $(\ell,E,F,\ldots,F)\mapsto\ell$.
Different refined branches can therefore give the same observed record $\ell$.

We group the contributions by their degrees $(t,t')$, where $t$ and $t'$ count the factors of $U$ and $V$ in the expansion of the channel Kraus operators.
As before, $J_U^{(t)}$ denotes the Choi operator of $U^{\ox t}$, and similarly for $V$.
At degree zero, the Choi space is one-dimensional and $J_U^{(0)}=J_V^{(0)}=1$.
By Lemma~\ref{lem:phase averaged branch probabilities}, there exist local positive semidefinite operators $M_{\lambda,1,t},M_{\lambda,2,t'}$ such that
\begin{align*}
p^{U,V}(\ell)
=\sum_{\lambda:\lambda\mapsto\ell}\sum_{t,t'=0}^{T}
\tr\left[M_{\lambda,1,t}J_U^{(t)}\right]
\tr\left[M_{\lambda,2,t'}J_V^{(t')}\right].
\end{align*}

Likewise, we introduce a reference weight. 
For a positive integer $K$, define
\begin{align*}
p_{\leq K}^{U,V}(\ell)
:=\sum_{\lambda:\lambda\mapsto\ell}
\sum_{\substack{0\leq t,t'\leq T\\t+t'\leq K}}
\tr[M_{\lambda,1,t}J_U^{(t)}]
\tr[M_{\lambda,2,t'}J_V^{(t')}].
\end{align*}
We do not renormalize these weights.
Then, we have 
\begin{align}
\norm{\bE_{U,V}p^{U,V}-\bE_U p^{U,U}}{\rm TV}
&\leq\norm{\bE_{U,V}p^{U,V}-\bE_{U,V}p_{\leq K}^{U,V}}{\rm TV} + \norm{\bE_U p^{U,U}-\bE_{U,V}p_{\leq K}^{U,V}}{\rm TV}.
\label{eq:phase lower bound split}
\end{align}
For the first term, Lemma~\ref{lem:phase averaged branch probabilities} shows that the weights in the expansion of $p^{U,V}$ sum to one over all $\lambda,t,t'$ and have a $(t+t')$-weighted sum at most $2T\varepsilon$.
Markov's inequality therefore gives
\begin{align}
\sum_\ell\left(p^{U,V}(\ell)-p_{\leq K}^{U,V}(\ell)\right)
\leq\frac{2T\varepsilon}{K}.
\label{eq:phase degree tail}
\end{align}
Since $p_{\leq K}^{U,V}(\ell)\leq p^{U,V}(\ell)$, averaging over $U,V$ yields
\begin{align}
\norm{\bE_{U,V}p^{U,V}-\bE_{U,V}p_{\leq K}^{U,V}}{\rm TV} 
= \frac12\sum_\ell\bE_{U,V}\left[p^{U,V}(\ell)-p_{\leq K}^{U,V}(\ell)\right]
\leq\frac{T\varepsilon}{K}.
\label{eq:phase independent truncation distance}
\end{align}
For the second term, assume $K^2\leq d/16$.
Fix $\lambda,t,t'$ with $t+t'\leq K$.
To compare terms of different degrees, we extend both operators to $K$ copies.
Define
\begin{align*}
M'_{\lambda,1,t}
&:=M_{\lambda,1,t}\ox(I_{d^2}/d)^{\ox(K-t)}, \quad
M'_{\lambda,2,t'}
:=M_{\lambda,2,t'}\ox(I_{d^2}/d)^{\ox(K-t')}.
\end{align*}
Since $\tr[(I_{d^2}/d)J_U]=\tr[J_U]/d=1$, we have
\begin{align*}
\tr[M'_{\lambda,1,t}J_U^{(K)}]
&=\tr[M_{\lambda,1,t}J_U^{(t)}]
\bigl(\tr[(I_{d^2}/d)J_U]\bigr)^{K-t}
=\tr[M_{\lambda,1,t}J_U^{(t)}].
\end{align*}
The same calculation applies to $M'_{\lambda,2,t'}$, for both $V$ and $U$.
Applying Lemma~\ref{lem:choi state cloning} and Eq.~\eqref{eq:haar moment upper comparison}, respectively, gives
\begin{align*}
\bE_U\tr[M'_{\lambda,1,t}J_U^{(K)}]\tr[M'_{\lambda,2,t'}J_U^{(K)}] 
&\geq\frac{d^{2K}}{d(d+1)\cdots(d+2K-1)}
\tr[M'_{\lambda,1,t}J_a^{(K)}]\tr[M'_{\lambda,2,t'}J_a^{(K)}], \\
\bE_{U,V}\tr[M'_{\lambda,1,t}J_U^{(K)}]\tr[M'_{\lambda,2,t'}J_V^{(K)}] 
&\leq\frac{1}{(1-K^2/(2d))^2}
\tr[M'_{\lambda,1,t}J_a^{(K)}]\tr[M'_{\lambda,2,t'}J_a^{(K)}].
\end{align*}
Thus, we have
\begin{align}
\bE_U p^{U,U}(\ell)
\geq\bE_U p_{\leq K}^{U,U}(\ell)
\geq\left(1-\frac{4K^2}{d}\right) \bE_{U,V}p_{\leq K}^{U,V}(\ell).
\label{eq:phase retained comparison}
\end{align}
Combining Eqs.~\eqref{eq:phase degree tail} and~\eqref{eq:phase retained comparison}, we obtain
\begin{align}
\norm{\bE_U p^{U,U}-\bE_{U,V}p_{\leq K}^{U,V}}{\rm TV} 
& =\sum_\ell\max\left\{\bE_{U,V}p_{\leq K}^{U,V}(\ell)-\bE_U p^{U,U}(\ell),0\right\}
+ \frac12\left(1-\sum_\ell\bE_{U,V}p_{\leq K}^{U,V}(\ell)\right)\notag\\
&\quad\leq\frac{4K^2}{d}\sum_\ell\bE_{U,V}p_{\leq K}^{U,V}(\ell)+\frac{T\varepsilon}{K}
\leq\frac{4K^2}{d}+\frac{T\varepsilon}{K}.
\label{eq:phase shared truncation distance}
\end{align}
Combining the two bounds in Eq.~\eqref{eq:phase lower bound split} gives
\begin{align}
\norm{\bE_{U,V}p^{U,V}-\bE_U p^{U,U}}{\rm TV}
&\leq\frac{T\varepsilon}{K}+\left(\frac{4K^2}{d}+\frac{T\varepsilon}{K}\right)
=\frac{4K^2}{d}+\frac{2T\varepsilon}{K}.
\label{eq:phase lower bound TV}
\end{align}
For sufficiently large $d$, choose
\begin{align}
K:=\left\lfloor\frac{\sqrt d}{4}\right\rfloor.
\end{align}
Combining this with Eq.~\eqref{eq:phase lower bound TV} gives
\begin{align}
\frac13
\leq\norm{\bE_{U,V}p^{U,V}-\bE_U p^{U,U}}{\rm TV}
\leq\frac14+\frac{2T\varepsilon}{K}, 
\quad\Rightarrow\quad
T\geq\frac{K}{24\varepsilon}
=\Omega\left(\frac{\sqrt d}{\varepsilon}\right).
\end{align}

Lastly, we extend the lower bound to arbitrary channel dimension $D\geq2$.
If $D=2d$ is even, the hard channel already has the required dimension.
If $D=2d+1$ is odd, decompose $\bC^D=\bC^{2d}\oplus\bC$ and define
\begin{align*}
\widetilde{\cE}_{\theta,U}
\left(\left[\begin{matrix}X_{11}&X_{12}\\X_{21}&X_{22}\end{matrix}\right]\right)
:=\left[\begin{matrix}\cE_{\theta,U}(X_{11})&0\\0&X_{22}\end{matrix}\right].
\end{align*}
Thus the original hard channel acts on the first $2d$ dimensions, and the remaining one-dimensional block is unchanged.
Writing $\widetilde J_{\theta,U}$ for the extended Choi operator, we have 
\begin{align*}
\frac{\tr[\widetilde J_{\theta_1,U}\widetilde J_{\theta_2,V}]}{D^2}
=\left(\frac{2d}{D}\right)^2
\frac{\tr[J_{\theta_1,U}J_{\theta_2,V}]}{4d^2}
+\frac1{D^2}.
\end{align*}
Each extended-channel query can be simulated locally with at most one original-channel query by measuring the input block.
An estimate of the extended overlap therefore gives an estimate of the original overlap by subtracting $1/D^2$ and multiplying by $(D/2d)^2$.
Since $(D/2d)^2\leq9/4$, this preserves the success probability and increases the additive error by at most a factor of $9/4$.
Thus, after a constant rescaling of the estimation accuracy, the preceding lower bound and $d=\Theta(D)$ give $T=\Omega(\sqrt D/\varepsilon)$ for sufficiently large $D$.
For bounded $D$, we have $\sqrt D/\varepsilon=O(\varepsilon^{-2})$ for $0<\varepsilon\leq\varepsilon_0$, so the independent $\Omega(\varepsilon^{-2})$ lower bound covers the remaining dimensions.

Therefore, combining with the independent $\Omega(1/\varepsilon^2)$ lower bound, completes the proof of Theorem~\ref{the:lower bound for DSEC}:
\begin{align}
T = \Omega\left(
\max\left\{
\frac{\sqrt d}{\varepsilon},
\frac{1}{\varepsilon^2}
\right\}
\right).
\end{align}
\end{proof}


\subsection{Technical Lemma}

\begin{lemma}
\label{lem:choi state cloning}
Assume $T^2\leq d/16$. For positive semidefinite operators $\rho,\sigma$ on $\cH^{\ox T}\ox\cH^{\ox T}$, where $\cH$ has dimension $d$, we have
\begin{align}
\bE_{U\sim\mu_H}\tr\left[\rho J_U^{(T)}\right]\tr\left[\sigma J_U^{(T)}\right]
\geq\frac{d^{2T}}{d(d+1)\cdots(d+2T-1)}
\tr\left[\rho J_a^{(T)}\right]\tr\left[\sigma J_a^{(T)}\right].
\end{align}
\end{lemma}

\begin{proof}[Proof of Lemma~\ref{lem:choi state cloning}]
Let
\begin{align*}
\Pi_{\rm sym}'^{(d,T)}
:=\frac1{T!}\sum_{\pi\in\cS_T}P_\pi\ox P_\pi
\end{align*}
be the projector onto the symmetric subspace of paired input-output registers~\cite{schuster2025random}.
First, project the operators onto this subspace:
\begin{align*}
\rho' :=\Pi_{\rm sym}'^{(d,T)}\rho\Pi_{\rm sym}'^{(d,T)}, \quad
\sigma' :=\Pi_{\rm sym}'^{(d,T)}\sigma\Pi_{\rm sym}'^{(d,T)}.
\end{align*}
Both projected operators are positive semidefinite and need not have unit trace.
Since $J_U^{(T)}$ is supported on this subspace, we have
\begin{align*}
\tr[\rho J_U^{(T)}]=\tr[\rho'J_U^{(T)}], \quad
\tr[\sigma J_U^{(T)}]=\tr[\sigma'J_U^{(T)}].
\end{align*}
We define the following permutation map, which pairs the input and output registers:
\begin{align}
\cR:\quad
\left(\bigotimes_{i=1}^T\ket{x_i}\right)\ox
\left(\bigotimes_{i=1}^T\ket{y_i}\right)
\longmapsto\bigotimes_{i=1}^T(\ket{x_i}\ox\ket{y_i}),
\label{eq:definition-map-R}
\end{align}
This register permutation gives
\begin{align}
\cR\left(\Pi_{{\rm sym}}'^{(d,T)}\right)=\Pi_{\rm sym}^{(d^2,T)}.
\end{align}
Thus, $\cR(\rho')$ and $\cR(\sigma')$ are supported on the symmetric subspace of $T$ registers of dimension $d^2$.
Let $\bF_{2,3}$ swap the second and third $\cH^{\ox T}$ blocks in $(\cH^{\ox T}\ox\cH^{\ox T})^{\ox2}$.
Then, our target expression can be rewritten as
\begin{align}
\bE_{U\sim\mu_H}\tr\left[\rho J_U^{(T)}\right]\tr\left[\sigma J_U^{(T)}\right]
&=\tr\left[(\rho'\ox\sigma')\bF_{2,3}J_H^{(2T)}\bF_{2,3}\right] \\
&\geq\frac{d^{2T}}{d(d+1)\cdots(d+2T-1)}
\tr\left[(\rho'\ox\sigma')\left(\bF_{2,3}J_a^{(2T)}\bF_{2,3}\right)\right],
\end{align}
where we use Lemma~\ref{lem:lemma 2 of chen2023unitarity}.
By the definitions of $J_a^{(2T)}$ and $\cR$, we have
\begin{align*}
(\cR\ox\cR)\left(\bF_{2,3}J_a^{(2T)}\bF_{2,3}\right)
&=\frac1{d^{2T}}\sum_{\pi\in\cS_{2T}}P_\pi^{(d^2)}
=\frac{(2T)!}{d^{2T}}\Pi_{\rm sym}^{(d^2,2T)}.
\end{align*}
Here $P_\pi^{(d^2)}$ permutes the $2T$ registers of dimension $d^2$, each formed by pairing an output register with its reference register.
Applying the same register permutation to both operators inside the trace then gives
\begin{align}
\bE_{U\sim\mu_H}\tr\left[\rho J_U^{(T)}\right]\tr\left[\sigma J_U^{(T)}\right]
\geq \frac{(2T)!}{d(d+1)\cdots(d+2T-1)}
\tr\left[\Pi_{\rm sym}^{(d^2,2T)}\bigl(\cR(\rho')\ox\cR(\sigma')\bigr)\right].
\label{eq: convert J_a^2T to Pi}
\end{align}

We now evaluate the trace against the symmetric projector by grouping permutations according to how many registers they exchange between the two $T$-register blocks.
Let $\rho'_s,\sigma'_s$ be the reduced operators on the first $s$ registers of $\cR(\rho'),\cR(\sigma')$, respectively, with $\rho'_0=\tr[\rho']$ and $\sigma'_0=\tr[\sigma']$.
If a permutation $\pi$ sends $s$ registers from the first block to the second, it also sends $s$ registers in the reverse direction.
By relabeling registers within each block before and after the permutation, it can be written as $P_\pi^{(d^2)}=P_\alpha W_sP_\beta$,
where $\alpha,\beta\in\cS_T\times\cS_T$ act within the two blocks and $W_s$ swaps their first $s$ registers.
Because both projected operators are supported on their respective symmetric subspaces, every $\tau\in\cS_T$ satisfies
\begin{align*}
P_\tau\cR(\rho')=\cR(\rho')P_\tau=\cR(\rho'),\qquad
P_\tau\cR(\sigma')=\cR(\sigma')P_\tau=\cR(\sigma').
\end{align*}
The blockwise permutations therefore drop out of the trace, leaving
\begin{align*}
\tr\left[P_\pi^{(d^2)}\bigl(\cR(\rho')\ox\cR(\sigma')\bigr)\right]
=\tr\left[W_s\bigl(\cR(\rho')\ox\cR(\sigma')\bigr)\right]
=\tr\left[\bF(\rho'_s\ox\sigma'_s)\right]
=\tr(\rho'_s\sigma'_s).
\end{align*}
The second equality traces out the unexchanged registers, and $\bF$ swaps the two remaining $s$-register systems.

To count the permutations in this class, choose the image of the first block: it contains $T-s$ positions in the first block and $s$ in the second, giving $\binom{T}{T-s}\binom Ts=\binom Ts^2$ choices.
For each choice, the two blocks can be mapped bijectively to their respective image sets in $T!$ ways each.
Thus, the class contains $(T!)^2\binom Ts^2$ permutations.
Using the permutation expansion of the symmetric projector, we obtain
\begin{align}
\tr\left[\Pi_{\rm sym}^{(d^2,2T)}\bigl(\cR(\rho')\ox\cR(\sigma')\bigr)\right]
&=\frac1{(2T)!}\sum_{\pi\in\cS_{2T}}
\tr\left[P_\pi^{(d^2)}\bigl(\cR(\rho')\ox\cR(\sigma')\bigr)\right]\notag\\
&=\frac1{\binom{2T}{T}}\sum_{s=0}^T\binom Ts^2\tr(\rho'_s\sigma'_s)
\geq\frac{\tr[\rho']\tr[\sigma']}{\binom{2T}{T}}.
\label{eq:symmetric states partial swap}
\end{align}
Here each term is nonnegative, since $\rho'_s,\sigma'_s\succeq0$, and the $s=0$ term equals $\tr[\rho']\tr[\sigma']$.
Since $J_a^{(T)}=(T!/d^T)\Pi_{\rm sym}'^{(d,T)}$, we have $\tr[\rho J_a^{(T)}]=(T!/d^T)\tr[\rho']$, and similarly for $\sigma$.
Substituting into Eq.~\eqref{eq: convert J_a^2T to Pi} gives
\begin{align}
\bE_{U\sim\mu_H}\tr\left[\rho J_U^{(T)}\right]\tr\left[\sigma J_U^{(T)}\right]
&\geq\frac{(T!)^2}{d(d+1)\cdots(d+2T-1)}\tr[\rho']\tr[\sigma']\notag\\
&=\frac{d^{2T}}{d(d+1)\cdots(d+2T-1)}
\tr[\rho J_a^{(T)}]\tr[\sigma J_a^{(T)}].
\label{eq:coherent access:mp term}
\end{align}
\end{proof}

\begin{lemma}
\label{lem:solve problem2 with algorithm A}
Suppose $d$ is sufficiently large and $0<\varepsilon\leq 0.01$ and an algorithm $\bA$ estimates the normalized Choi overlap to additive error $\varepsilon/100$ with probability at least $0.99$.
Then, we can solve Problem~\ref{pro:distinguish problem 2} with high probability in either case by the following procedure:
\begin{enumerate}
\item Obtain the estimation result $f$ with the estimation algorithm $\bA$. 
\item Output Case~2 if $f \in [c_{\rm ref} - \varepsilon/50, c_{\rm ref} + \varepsilon/50]$, and output Case~1 otherwise, where $c_{\rm ref} := [1 + (1-\varepsilon)^2]/4$.
\end{enumerate}
\end{lemma}
\begin{proof}[Proof of Lemma~\ref{lem:solve problem2 with algorithm A}]
The performance of the estimation algorithm $\bA$ guarantees that for each fixed channel pair in Case~$i (i = 1, 2)$, $|f - f_i| \leq \varepsilon/100$ holds with probability at least $0.99$. 
Then, if $f$ is from Case~1, we have
\begin{align}
&\Pr\left\{f \in [c_{\rm ref} - \varepsilon/50, c_{\rm ref} + \varepsilon/50]\right\} \notag \\
\leq&\; \Pr\left\{f \in [c_{\rm ref} - \varepsilon/50, c_{\rm ref} + \varepsilon/50] \wedge  |f - f_1| \leq \frac{\varepsilon}{100}\right\} + 0.01 \\
\leq&\; \Pr\left\{f_1 \in [c_{\rm ref} - 3\varepsilon/100, c_{\rm ref} + 3\varepsilon/100]\right\} + 0.01 \\
=&\; \Pr_{\gamma\sim[0,2\pi]} \left\{\cos(\gamma) \in [-(0.06+\varepsilon/2)/(1-\varepsilon), (0.06 - \varepsilon/2)/(1-\varepsilon)]\right\} + 0.01.
\end{align}
With the assumption $\varepsilon\leq 0.01$, we have 
\begin{align}
\Pr\left\{f \in [c_{\rm ref} - \varepsilon/50, c_{\rm ref} + \varepsilon/50]\right\}
&\leq \Pr_{\gamma\sim[0,2\pi]} \left\{\cos(\gamma) \in [-0.066, 0.061] \right\} + 0.01
\leq 0.06. 
\end{align}
Thus, the above algorithm outputs Case~1 with probability at least $0.94$ when the input is from Case~1.
On the other hand, if the input is from Case~2, we have
\begin{align}
\Pr\left\{f \in [c_{\rm ref} - \varepsilon/50, c_{\rm ref} + \varepsilon/50]\right\} 
\geq&\; 
\Pr\left\{\left|f_2-c_{\rm ref}\right|\leq\frac{\varepsilon}{100}
\;\wedge\;
\left|f-f_2\right|\leq\frac{\varepsilon}{100}\right\} \notag\\
\geq&\;
\Pr\left\{\left|\frac{\tr(U^\dagger V)}{d}\right|\leq\frac{1}{200}
\;\wedge\;
\left|f-f_2\right|\leq\frac{\varepsilon}{100}\right\} \notag\\
\geq&\;
\Pr\left\{\left|\frac{\tr(U^\dagger V)}{d}\right|\leq\frac{1}{200}\right\}
+\Pr\left\{\left|f-f_2\right|\leq\frac{\varepsilon}{100}\right\}
-1 \notag\\
\geq&\;
\Pr\left\{\left|\frac{\tr(U^\dagger V)}{d}\right|\leq\frac{1}{200}\right\}-0.01.
\end{align}
Therefore, for sufficiently large $d$, the above algorithm outputs Case 2 with high probability when the input is from Case 2.
\end{proof}

\begin{lemma}
\label{lem:simulation of phi_U,t}
For any subset \(R\subseteq[T]\) and any \(0\leq t\leq |R|\), the normalized state \(d^{-|R|/2}\ket{\Phi_{U,R,t}}\) in Eq.~\eqref{eq:definition of phi_U,R,t} can be prepared using exactly \(t\) queries to the unknown unitary \(U\).
\end{lemma}
\begin{proof}[Proof of Lemma~\ref{lem:simulation of phi_U,t}]
Let $m:=|R|$.
It suffices to prove the claim for the active registers indexed by $R$, since all registers outside $R$ are fixed and independent of $U$.
Relabel the active registers as $1,\ldots,m$ for the duration of the proof.
Let $B_1,\ldots,B_m$ and $B'_1,\ldots,B'_m$ denote the two halves of the maximally entangled state on $\cH_B^{\ox m}\ox\cH_{B'}^{\ox m}$, and let $A_1,\ldots,A_m$ and $A'_1,\ldots,A'_m$ be the output and reference flag registers.
For each subset $S\subseteq[m]$ with $|S|=t$, write
\begin{align}
\ket{f_S}_A:=\ket{0}_{A,\bar S}\ox\ket{1}_{A,S}.
\end{align}
Prepare the reference flag registers in $\ket{0}_{A'}^{\ox m}$, and prepare the output flag registers and the maximally entangled registers in the state
\begin{align}
\ket{0}_{A'}^{\ox m}\ox
\frac{1}{d^{m/2} \sqrt{\binom{m}{t}}}
\sum_{\substack{S\subseteq[m]\\ |S|=t}}
\ket{f_S}_A\ox\ket{\Phi_d}^{\ox m}.
\end{align}
Let $P_S$ be a known permutation unitary on $B_1,\ldots,B_m$ that moves the registers indexed by $S$ to the first $t$ positions.
Equivalently, for $U_{[t]}:=U^{\ox t}\ox I^{\ox(m-t)}$ acting on $B_1,\ldots,B_m$, we have
\begin{align}
P_S^\dagger U_{[t]} P_S = U_S.
\end{align}
Apply the controlled permutation
\begin{align}
C_P:=\sum_{\substack{S\subseteq[m]\\ |S|=t}}\ketbra{f_S}_A\ox P_S
\end{align}
on the fixed-weight flag subspace, extending it as the identity on its orthogonal complement. Query the unknown unitary $U$ on the first $t$ $B$-registers, and then apply $C_P^\dagger$.
The resulting normalized vector is
\begin{align}
&\ket{0}_{A'}^{\ox m}\ox
\frac{1}{d^{m/2}\sqrt{\binom{m}{t}}}
\sum_{\substack{S\subseteq[m]\\ |S|=t}}\ket{f_S}_A
\ox \left(P_S^\dagger U_{[t]}P_S\ox I_{B'}^{\ox m}\right)\ket{\Phi_d}^{\ox m} \notag \\
=&\; \frac{1}{d^{m/2}\sqrt{\binom{m}{t}}} 
\ket{0}_{A'}^{\ox m}\ox
\sum_{\substack{S\subseteq[m]\\ |S|=t}}
\ket{f_S}_A
\ox \left(U_S\ox I_{B'}^{\ox m}\right)\ket{\Phi_d}^{\ox m} 
= d^{-m/2}\ket{\Phi_{U,R,t}},
\end{align}
Undoing the temporary relabeling gives the normalized state \(d^{-m/2}\ket{\Phi_{U,R,t}}\); equivalently, the normalized density operator is \(d^{-m}\ketbra{\Phi_{U,R,t}}\).
The only $U$-dependent operation is the application of $U_{[t]}$, so the construction uses precisely $t$ queries to the unitary oracle.
All registers corresponding to the known Kraus operator $F$ in the proof of Theorem~\ref{the:lower bound for DSEC} are independent of $U$ and can be appended by known operations without using any additional query.
\end{proof}

\begin{lemma}
\label{lem:phase averaged branch probabilities}
Consider an adaptive LOCC protocol in Problem~\ref{pro:distinguish problem 2} using at most $T$ queries per device.
There exist local positive semidefinite operators $M_{\lambda,1,t}$ and $M_{\lambda,2,t'}$, determined by the protocol, the refined branch $\lambda$, and $\varepsilon$, such that for every fixed $U,V$ and complete classical record $\ell$,
\begin{align}
p^{U,V}(\ell)
=\sum_{\lambda:\lambda\mapsto\ell}\sum_{t,t'=0}^{T}
\tr\left[M_{\lambda,1,t}J_U^{(t)}\right]
\tr\left[M_{\lambda,2,t'}J_V^{(t')}\right].
\label{eq:actual phase degree expansion}
\end{align}
The same operators satisfy
\begin{align}
\sum_{\lambda,t,t'}
\tr\left[M_{\lambda,1,t}J_U^{(t)}\right]
\tr\left[M_{\lambda,2,t'}J_V^{(t')}\right]
&=1,\notag\\
\sum_{\lambda,t,t'}(t+t')
\tr\left[M_{\lambda,1,t}J_U^{(t)}\right]
\tr\left[M_{\lambda,2,t'}J_V^{(t')}\right]
&\leq2T\varepsilon.
\label{eq:phase degree mass}
\end{align}
\end{lemma}
\begin{proof}[Proof of Lemma~\ref{lem:phase averaged branch probabilities}]
Fixing a refined branch $\lambda$ specifies every local Kraus operator and adaptive control.
Let $\ket{\Psi_\lambda}$ be the resulting unnormalized branch vector, with the square root of the shared-seed probability included in its first local factor.
Its squared norm is the probability weight of this branch, and hence
\begin{align*}
p^{U,V}(\ell)
=\sum_{\lambda:\lambda\mapsto\ell}
\bE_{\theta_1,\theta_2}\|\ket{\Psi_\lambda}\|^2.
\end{align*}
Expand each $E_{\theta_1,U}$ and $E_{\theta_2,V}$ into its phase and unitary terms.
The factors $e^{i\theta_1}$ and $e^{i\theta_2}$ extracted from the corresponding $E$ operators form an overall phase of $\ket{\Psi_\lambda}$, which does not affect its squared norm.
In the remaining expansion, each factor of $U$ carries $e^{-i\theta_1}$, and each factor of $V$ carries $e^{-i\theta_2}$.
Collect the coefficients of all terms with degrees $(t,t')$ into $\ket{v_{\lambda,t,t'}(U,V)}$.
This vector is independent of the phases and is homogeneous of degrees $t,t'$ in the entries of $U,V$; it includes the factors of $\sqrt\varepsilon$ and $\sqrt{1-\varepsilon}$ from the channel calls.
Since each device is queried at most $T$ times, both degrees are at most $T$, and hence
\begin{align*}
\|\ket{\Psi_\lambda}\|^2
=\left\|
\sum_{t,t'=0}^{T}e^{-i(t\theta_1+t'\theta_2)}
\ket{v_{\lambda,t,t'}(U,V)}
\right\|^2.
\end{align*}
For an integer $k$, a uniform phase satisfies
\begin{align*}
\bE_\theta e^{ik\theta}
=\frac{1}{2\pi}\int_0^{2\pi}e^{ik\theta}\,d\theta
=\begin{cases}1,&k=0,\\0,&k\ne0.\end{cases}
\end{align*}
The independence of $\theta_1,\theta_2$ thus gives $\bE_{\theta_1,\theta_2} e^{i[(s-t)\theta_1+(s'-t')\theta_2]} = \delta_{t,s}\delta_{t',s'}$. 
Expanding the squared norm and averaging therefore gives
\begin{align*}
\bE_{\theta_1,\theta_2}\|\ket{\Psi_\lambda}\|^2
&=\sum_{t,t',s,s'=0}^{T}
\delta_{t,s}\delta_{t',s'}
\langle v_{\lambda,s,s'}(U,V)\mid v_{\lambda,t,t'}(U,V)\rangle
=\sum_{t,t'=0}^{T}\|\ket{v_{\lambda,t,t'}(U,V)}\|^2.
\end{align*}
Thus, the probability of the observed record $\ell$ can be written as
\begin{align}
p^{U,V}(\ell)
=\sum_{\lambda:\lambda\mapsto\ell}\sum_{t,t'=0}^{T}
\|\ket{v_{\lambda,t,t'}(U,V)}\|^2.
\label{eq:phase averaged branch probabilities}
\end{align}

We now express each $\|\ket{v_{\lambda,t,t'}(U,V)}\|^2$ in terms of $J_U^{(t)}$ and $J_V^{(t')}$.
For a fixed refined branch, all operations act locally on the initial product state, so $\ket{v_{\lambda,t,t'}(U,V)}=\ket{v^{(1)}_{\lambda,t}(U)}\ox\ket{v^{(2)}_{\lambda,t'}(V)}$.
The two local factors have degrees $t$ and $t'$ in $U$ and $V$, respectively.
Expanding each occurrence of $U$ in the computational basis yields products of $t$ matrix entries, each multiplying a vector independent of $U$.
We thus obtain the following expansion and define $L_{\lambda,1,t}$ on the corresponding tensor-product basis:
\begin{align*}
\ket{v^{(1)}_{\lambda,t}(U)}
&=\sum_{i_1,j_1,\ldots,i_t,j_t=0}^{d-1}
\left(\prod_{r=1}^t U_{i_rj_r}\right)
\ket{c_{\lambda,t;i_1j_1,\ldots,i_tj_t}},\\
L_{\lambda,1,t}\left(\bigotimes_{r=1}^t\ket{i_rj_r}\right)
&:=\ket{c_{\lambda,t;i_1j_1,\ldots,i_tj_t}}.
\end{align*}
Since $|U\rrangle=\sum_{i,j}U_{ij}\ket{ij}$, linearity gives $\ket{v^{(1)}_{\lambda,t}(U)}=L_{\lambda,1,t}|U\rrangle^{\ox t}$.
Constructing $L_{\lambda,2,t'}$ in the same way gives
\begin{align*}
\ket{v_{\lambda,t,t'}(U,V)}
=\left(L_{\lambda,1,t}|U\rrangle^{\ox t}\right)
\ox\left(L_{\lambda,2,t'}|V\rrangle^{\ox t'}\right).
\end{align*}
Set
\begin{align*}
M_{\lambda,1,t}:=L_{\lambda,1,t}^\dagger L_{\lambda,1,t},
\qquad
M_{\lambda,2,t'}:=L_{\lambda,2,t'}^\dagger L_{\lambda,2,t'}.
\end{align*}
Taking squared norms and using $J_U^{(t)}=(|U\rrangle\llangle U|)^{\ox t}$, with the fixed register ordering absorbed into the maps, gives
\begin{align}
\|\ket{v_{\lambda,t,t'}(U,V)}\|^2
=\tr\left[M_{\lambda,1,t}J_U^{(t)}\right]
\tr\left[M_{\lambda,2,t'}J_V^{(t')}\right].
\label{eq:phase branch coefficient contraction}
\end{align}
Substituting this identity into Eq.~\eqref{eq:phase averaged branch probabilities} proves Eq.~\eqref{eq:actual phase degree expansion}.

Summing Eq.~\eqref{eq:actual phase degree expansion} over all records $\ell$ gives the normalization identity in Eq.~\eqref{eq:phase degree mass}, since $p^{U,V}$ is a probability distribution.

It remains to bound the degree-weighted sum.
Initially, only degree $(0,0)$ occurs, and the squared norms summed over the shared-random seeds equal one.
For a fixed partial refined branch $\lambda$, use the same coefficient expansion and write $\ket{v_{t,t'}}$ for $\ket{v_{\lambda,t,t'}(U,V)}$.
Let $\ket{v^{E}_{t,t'}}$ and $\ket{v^{F}_{t,t'}}$ denote the degree-$(t,t')$ coefficient vectors after applying the Kraus operators $E_{\theta_1,U}$ and $F$, respectively, with the overall phase $e^{i\theta_1}$ extracted in the $E$ case.
We set coefficients with negative indices or absent degrees to zero.
Then, we have
\begin{align}
\ket{v^{E}_{t,t'}}
=\sqrt{1-\varepsilon}(\ketbra{0}_A\ox I_B)\ket{v_{t,t'}} +\sqrt\varepsilon(\ket1\!\bra0_A\ox U_B)\ket{v_{t-1,t'}}, \quad 
\ket{v^{F}_{t,t'}} = F\ket{v_{t,t'}}.
\end{align}
Their squared norms satisfy
\begin{align*}
\|\ket{v^{E}_{t,t'}}\|^2+\|\ket{v^{F}_{t,t'}}\|^2
&=\|\ket{v_{t,t'}}\|^2
+\varepsilon\|(\ketbra{0}_A\ox I_B)\ket{v_{t-1,t'}}\|^2
-\varepsilon\|(\ketbra{0}_A\ox I_B)\ket{v_{t,t'}}\|^2.
\end{align*}
We sum over all nonnegative degrees, taking coefficients at absent degrees to be zero.
The last two terms cancel after summation, giving
\begin{align*}
\sum_{t,t'}\left(\|\ket{v^{E}_{t,t'}}\|^2+\|\ket{v^{F}_{t,t'}}\|^2\right)
=\sum_{t,t'}\|\ket{v_{t,t'}}\|^2.
\end{align*}
For the degree-weighted sum, substituting $s=t-1$ in the shifted term gives
\begin{align*}
\sum_{t,t'}(t+t')\|(\ketbra{0}_A\ox I_B)\ket{v_{t-1,t'}}\|^2
&=\sum_{s,t'}(s+t'+1)\|(\ketbra{0}_A\ox I_B)\ket{v_{s,t'}}\|^2\\
&=\sum_{t,t'}(t+t'+1)\|(\ketbra{0}_A\ox I_B)\ket{v_{t,t'}}\|^2.
\end{align*}
Thus, after multiplying the identity for the squared norms by $t+t'$ and summing, the difference $(t+t'+1)-(t+t')=1$ leaves
\begin{align*}
\sum_{t,t'}(t+t')\left(\|\ket{v^{E}_{t,t'}}\|^2+\|\ket{v^{F}_{t,t'}}\|^2\right)
&=\sum_{t,t'}(t+t')\|\ket{v_{t,t'}}\|^2
+\varepsilon\sum_{t,t'}\|(\ketbra{0}_A\ox I_B)\ket{v_{t,t'}}\|^2 \\
&\leq \sum_{t,t'}(t+t'+\varepsilon)\|\ket{v_{t,t'}}\|^2. 
\end{align*}
Let $\lambda'$ label the branches obtained by advancing each current branch $\lambda$ through its next channel call or to termination.
The intervening known operations have Kraus operators $\{K_{\lambda,j}\}_j$ and transform each coefficient to $K_{\lambda,j}\ket{v_{\lambda,t,t'}(U,V)}$, with the identity used on terminated branches.
Summing the single-call bound over all resulting branches gives
\begin{align*}
\sum_{\lambda',t,t'}(t+t')\|\ket{v_{\lambda',t,t'}(U,V)}\|^2
&\leq\sum_{\lambda,j,t,t'}(t+t'+\varepsilon)\|K_{\lambda,j}\ket{v_{\lambda,t,t'}(U,V)}\|^2\\
&=\sum_{\lambda,t,t'}(t+t'+\varepsilon)\|\ket{v_{\lambda,t,t'}(U,V)}\|^2\\
&=\sum_{\lambda,t,t'}(t+t')\|\ket{v_{\lambda,t,t'}(U,V)}\|^2+\varepsilon.
\end{align*}
Here the first equality uses $\sum_jK_{\lambda,j}^\dagger K_{\lambda,j}=I$.
Thus each step increases the total degree-weighted sum by at most $\varepsilon$.
Every branch contains at most $2T$ channel calls, so at most $2T$ such steps suffice to process all queries; any remaining known operations preserve the weighted sum.
Starting from degree zero, and writing $\lambda$ for the complete branches at termination, we obtain
\begin{align*}
\sum_{\lambda,t,t'}(t+t')\|\ket{v_{\lambda,t,t'}(U,V)}\|^2\leq2T\varepsilon
\end{align*}
at the end of the protocol.
\end{proof}

\section{Upper Bound}
\label{sec:upper_bound}

We next give upper bounds. The first algorithm uses coherent access and applies to unitary channels; the second uses only incoherent access and applies to arbitrary channels.

\subsection{Coherent Access}

We first consider the case where both unknown channels are unitary, i.e., $\cE(\cdot) = U(\cdot)U^\dagger$ and $\cF(\cdot) = V(\cdot)V^\dagger$.
For this setting, the algorithm is based on symmetric collective measurements defined in Eq.~\eqref{eq:symmetric collective measurement}.
Our algorithm proceeds as follows for one input state:

\begin{enumerate}
\item Randomly generate pure state $\ket{\psi}^{\ox T}$ on each quantum device, where $\ket{\psi}$ from a state $4$-design ensemble. 
\item Apply the unitary operators $U^{\ox T}$ and $V^{\ox T}$ on $\ket{\psi}^{\ox T}$. 
\item Measure $(U\ket{\psi})^{\ox T}$ and $(V\ket{\psi})^{\ox T}$ with the POVM $\cM_{T}$. 
Record the measurement results on two devices as $\ket{\phi_A}$ and $\ket{\phi_B}$, and compute $\tilde{f} = |\braket{\phi_A}{\phi_B}|^2$. 
\item Compute the unbiased estimator: 
\begin{align}
\tilde{\chi} 
:= \frac{(d+1)(d+T)^2}{T^2 d} \tilde{f}  - \frac{(d+1)(d+2T) + T^2}{T^2 d}. 
\label{eq:coherent estimator}
\end{align}
\end{enumerate}
Repeat these steps independently $K$ times, using a fresh shared input state in each repetition, and return
\begin{align}
\sum_{k=1}^K\tilde{\chi}_{k}/K,
\label{eq:coherent averaged estimator}
\end{align}
where $\tilde{\chi}_{k}$ is the estimator from repetition $k$. The total query count is $KT$ per device.

\begin{algorithm}[t]
\caption{Distributed Similarity Estimation of Unitary Channels with \textbf{Coherent Access}}\label{alg:coherent access}
\KwInput{$KT$ queries to unknown \emph{unitary} channels $\cE$ and $\cF$ acting on a $d$-dimensional Hilbert space $\cH$.}
\KwOutput{an estimation of $\tr[J_\cE J_\cF] / d^2$.}

\For{$k=1,\cdots,K$}{
Randomly generate pure states $\ket{\psi}^{\ox T}$ on each device, where $\ket{\psi}$ is sampled from a state $4$-design ensemble.

Apply the unitary channels $\cE^{\ox T}$ and $\cF^{\ox T}$ on $\ket{\psi}^{\ox T}$.

Measure $(\cE(\ketbra{\psi}))^{\ox T}$ with the POVM $\cM_T$ and obtain result $\ket{\phi_A}$.

Measure $(\cF(\ketbra{\psi}))^{\ox T}$ with the POVM $\cM_T$ and obtain result $\ket{\phi_B}$.

Compute $\tilde{f} = |\braket{\phi_A}{\phi_B}|^2$ and store 
\begin{align}
\tilde{\chi}_{k} 
:= \frac{(d+1)(d+T)^2}{T^2 d} \tilde{f}  - \frac{(d+1)(d+2T) + T^2}{T^2 d}. 
\end{align}
}
Return $\sum_{k=1}^K\tilde{\chi}_{k}/K$.
\end{algorithm}

The algorithm uses shared randomness only to choose the input states; no shared measurement setting is required. 
Pseudocode is given in Algorithm~\ref{alg:coherent access}.
We summarize the query complexity of this algorithm in the following theorem, which is proven in Appendix~\ref{app:upper bound for coherent access}. 

\begin{theorem}
\label{the:coherent upper bound}
For unitary channels $\cE$ and $\cF$, the estimator $\tilde{\chi}$ defined in Eq.~\eqref{eq:coherent estimator} is unbiased and satisfies
\begin{align}
\bE\tilde{\chi} = \frac{1}{d^2} \tr[J_\cE J_\cF], \quad 
\Var(\tilde{\chi}) \leq \cO\left(\frac{1}{T} + \frac{d}{T^2} + \frac{d^2}{T^4} +\frac1d\right). 
\end{align}
The average estimator is unbiased and has variance $\Var(\tilde{\chi})/K$.
For $d\geq2$ and $0<\varepsilon<1$, there exist positive integers
\begin{align*}
T=\left\lceil\min\left\{d,\frac{\sqrt d}{\varepsilon}\right\}\right\rceil, \quad
K=\cO\left(\max\left\{1,\frac1{d\varepsilon^2}\right\}\right), 
\end{align*}
such that the average estimator achieves additive error $\varepsilon$ with success probability at least $2/3$, using $KT=\cO(\max\{1/\varepsilon^2,\sqrt d/\varepsilon\})$ queries per device.
\end{theorem}

Combined with Theorem~\ref{the:coherent lower bound}, this coherent algorithm is optimal in $d$ at constant accuracy for unitary channels.
However, this algorithm fundamentally requires unitary channels since the proof relies on pure output states, see Appendix~\ref{app:upper bound for coherent access} for details. 
Rather than extending the coherent algorithm through random purification techniques~\cite{tang2025conjugate, pelecanos2025mixed, mele2025optimal}, we show next that incoherent access alone already achieves the optimal scaling for general channels.

\subsection{Incoherent Access}

We now present our \emph{main algorithm}, which works for \emph{general} channels under incoherent access. 
The incoherent algorithm is the experimentally most relevant part of the work. 
Each device is used in a simple prepare--evolve--measure cycle: 
prepare a local input state, apply the unknown channel once, and measure immediately. 
No entangled input across devices, quantum memory, collective coherent control, or adaptive feedback is required.
The concrete procedure is as follows.

Its basic subroutine, $\texttt{CIPE}(m,Q,\ket{\alpha},\ket{\beta},\cE,\cF)$, applies $\cE$ and $\cF$ to input states $\ket{\alpha}$ and $\ket{\beta}$ on the two devices, respectively, and then measures $m$ output copies from each device in the common basis $\{Q^\dagger\ketbra{a}Q\}_a$, where $Q$ is a random unitary from a unitary $4$-design ensemble and $\{\ketbra{a}\}_a$ is the computational basis.
Record the measurement outcomes on two devices as $\{a_i\}_{i=1}^m$ and $\{b_j\}_{j=1}^m$, and return the classical inner product estimator
\begin{align}
\tilde{c} = \frac{1}{m^2} \sum_{i,j=1}^m \bm{1} (a_i, b_j),
\end{align}
where $\bm{1}(a, b) = 1$ if $a=b$; otherwise $0$.
This $\texttt{CIPE}$ algorithm is described in pseudocode in Algorithm~\ref{alg:classical inner product}.

The estimator is based on a bias-correction identity.
Averaging over the shared random basis and random input states, the same-input statistic $\tilde{g}$ contains the desired Choi overlap together with a nuisance term depending on $\bE_{\psi,\phi}\tr[\cE(\ketbra{\psi})\cF(\ketbra{\phi})]$.
The independent-input statistic $\tilde{o}$ estimates the corresponding nuisance contribution.
Combining the two moments gives
\begin{align}
\frac{\tr[J_\cE J_\cF]}{d^2}
= \frac{(d+1)^2}{d}\bE\tilde{g} - (d+1)\bE\tilde{o} - \frac{1}{d}.
\end{align}
The algorithm below implements this identity by estimating both moments with shared measurement randomness.

Then, our main algorithm repeats the following procedure $T$ times independently:
\begin{enumerate}
\item Run $\texttt{CIPE}(m, Q, \ket{\psi}, \ket{\psi}, \cE, \cF)$ to obtain $\tilde{g}$, where $\ket{\psi}$ is sampled from a state $4$-design ensemble and $Q$ is sampled from a unitary $4$-design ensemble.
\item Run $\texttt{CIPE}(m, Q', \ket{\psi'}, \ket{\phi'}, \cE, \cF)$ to obtain $\tilde{o}$, where $\ket{\psi'}$ and $\ket{\phi'}$ are independently sampled from a state $4$-design ensemble and $Q'$ is sampled from a unitary $4$-design ensemble.
\item Compute the unbiased estimator:
\begin{align}
\tilde{\omega} = \frac{(d+1)^2}{d}\tilde{g} - (d+1)\tilde{o} - \frac{1}{d}. 
\label{eq:incoherent estimator}
\end{align}
\end{enumerate}

\begin{algorithm}[t]
\caption{Classical Inner Product Estimation: $\texttt{CIPE}(m, Q, \ket{\alpha}, \ket{\beta}, \cE, \cF)$}\label{alg:classical inner product}
\KwInput{number of measurements $m$, \\
\hspace{0.88cm} a random unitary $Q$ from a unitary $4$-design ensemble, \\
\hspace{0.88cm} two random pure states $\ket{\alpha}$ and $\ket{\beta}$ from a state $4$-design ensemble, \\
\hspace{0.88cm} unknown channels $\cE$ and $\cF$ acting on a $d$-dimensional Hilbert space $\cH$.}
\KwOutput{an estimator $\tilde{c}$ for the classical inner product.}

On the first device, prepare the state $\ket{\alpha}$ and apply the unknown channel $\cE$.
Measure $m$ copies of $\cE(\ketbra{\alpha})$ in the basis $\{Q^\dagger\ketbra{a}Q\}_a$ and obtain $A = \{a_i\}_{i=1}^m$.

On the second device, prepare the state $\ket{\beta}$ and apply the unknown channel $\cF$.
Measure $m$ copies of $\cF(\ketbra{\beta})$ in the basis $\{Q^\dagger\ketbra{b}Q\}_b$ and obtain $B = \{b_i\}_{i=1}^m$.

Compute $\tilde{c} = \sum_{i, j=1}^m \bm{1}(a_i, b_j) / m^2$, 
where $\bm{1}(a, b) = 1$ if $a=b$; otherwise $0$.

Return $\tilde{c}$.
\end{algorithm}

\begin{algorithm}[t]
\caption{Distributed Similarity Estimation of General Channels with \textbf{Incoherent Access}}\label{alg:incoherent access}
\KwInput{number of SPAM settings $T$, \\
\hspace{0.88cm} number of measurements for each SPAM setting $m$, \\
\hspace{0.88cm} $2Tm$ queries to each of the unknown channels $\cE$ and $\cF$ acting on a $d$-dimensional Hilbert space $\cH$.}
\KwOutput{an estimate of $\tr[J_\cE J_\cF] / d^2$.}
\For{$t = 1, \cdots, T$}{



Run Algorithm~\ref{alg:classical inner product} $\texttt{CIPE}(m, Q_t, \ket{\psi_t}, \ket{\psi_t}, \cE, \cF)$ to obtain $\tilde{g}_t$,
where $\ket{\psi_t}$ is sampled from a state $4$-design ensemble and $Q_t$ is sampled from a unitary $4$-design ensemble.

Run Algorithm~\ref{alg:classical inner product} $\texttt{CIPE}(m, Q'_t, \ket{\psi_t'}, \ket{\phi_t'}, \cE, \cF)$ to obtain $\tilde{o}_t$,
where $\ket{\psi_t'}$ and $\ket{\phi_t'}$ are two independent random pure states sampled from a state $4$-design ensemble and $Q'_t$ is sampled from a unitary $4$-design ensemble, independently of the settings used for $\tilde g_t$.

Compute
\begin{align}
\tilde{\omega}_t 
= \frac{(d+1)^2}{d}\tilde{g}_t - (d+1)\tilde{o}_t - \frac{1}{d}. 
\end{align}
}
Return ${\omega} := \sum_t \tilde{\omega}_t / T$.
\end{algorithm}

This incoherent algorithm is summarized in Algorithm~\ref{alg:incoherent access}.
The first $\texttt{CIPE}$ quantifies the similarity between $\cE$ and $\cF$ when they act on the same input state, while the second $\texttt{CIPE}$ characterizes the behavior when they act on different input states.
If one of the channels is unital, for example a unitary or Pauli channel, the term involving $\tilde{o}$ is fixed and $\tr[J_\cE J_\cF]/d^2$ can be estimated directly by $(d+1)^2\tilde{g}/d-(d+2)/d$.
That is, we only need to run $\texttt{CIPE}$ once in each round.
The corresponding query complexity is shown in the following theorem, which is provided in Appendix~\ref{app:upper bound for incoherent access}.

\begin{theorem}
\label{the:incoherent upper bound}
For general channels $\cE$ and $\cF$, the expectation and variance of estimator $\tilde{\omega}$ defined in Eq.~\eqref{eq:incoherent estimator} are given by 
\begin{align}
\bE\tilde{\omega} = \frac{1}{d^2} \tr[J_\cE J_\cF], \quad 
\Var(\tilde{\omega}) \leq \cO\left(\frac{d}{m^2} + \frac{1}{m} + \frac{1}{d}\right). 
\end{align}
Consequently, for $d\geq2$ and $0<\varepsilon<1$, there exist positive integers
\begin{align}
m=\left\lceil\min\left\{d,\frac{\sqrt d}{\varepsilon}\right\}\right\rceil,
\qquad
T=\cO\left(\max\left\{1,\frac{1}{d\varepsilon^2}\right\}\right)
\label{eq:incoherent parameters}
\end{align}
such that Algorithm~\ref{alg:incoherent access} achieves additive error $\varepsilon$ with success probability at least $2/3$, using $2Tm=\cO(\max\{1/\varepsilon^2,\sqrt d/\varepsilon\})$ queries per device.
\end{theorem}

Together with Theorem~\ref{the:lower bound for DSEC}, this proves the exact DSEC query complexity $\Theta(\max\{\sqrt d/\varepsilon,1/\varepsilon^2\})$ for general channels.
This incoherent algorithm enjoys several notable advantages.
First, it is \emph{ancilla-free} and \emph{non-adaptive}, making it immediately implementable on near-term devices. 
Second, because it does not assume unitarity, the same primitive can be used for unitarity estimation~\cite{chen2023unitarity}.
Third, despite incoherent access being strictly weaker than coherent access in general, both algorithms achieve the same query complexity for DSEC.
This surprising phenomenon parallels the result for DIPE~\cite{anshu2022distributed}, which shows single-copy measurements can achieve the same sample complexity as symmetric collective measurements.

The key to these advantages lies in the different types of shared randomness employed.
The coherent algorithm shares only the randomness for state preparation, whereas the incoherent algorithm additionally shares measurement settings. 
This shared measurement randomness effectively boosts the performance of the incoherent algorithm, enabling it to achieve the optimal query complexity for general channels.
This finding highlights that shared randomness, rather than coherent access or collective measurements, is the key resource for efficient distributed learning.
On the other hand, the coherent algorithm has its own advantage: less shared randomness translates to less communication overhead. 
Thus, there exists a trade-off between classical communication resources and the power of the learning algorithm.

\section{Comparison with Independent Classical Shadow}
\label{sec:comparison_classical_shadow}

To isolate the effect of shared randomness, we compare with independent channel classical shadows~\cite{kunjummen2023shadow, levy2024classical, li2025nearly}, where the two devices sample their SPAM settings independently. We use the construction of~\cite{li2025nearly} as the baseline because it is near-optimal for unitary channels and has an incoherent variant for unital channels.
This comparison is useful because independent shadows and our algorithms use very similar local experimental primitives. 
The essential difference is whether the random SPAM settings at the two devices are correlated.

\begin{algorithm}[t]
\caption{Distributed Similarity Estimation of Unitary Channels with \textbf{Independent Channel Classical Shadows}}\label{alg:dseu with independent classical shadow}
\KwInput{number of SPAM settings $T$, \\
\hspace{0.8cm} the size of symmetric collective measurement $s$, \\
\hspace{0.8cm} $Ts$ queries to unknown unitary channels $\cE$ and $\cF$ acting on a $d$-dimensional Hilbert space $\cH$.}
\KwOutput{an estimation of $\tr[J_\cE J_\cF] / d^2$.}
\For{$t = 1, \cdots, T$}{

Randomly generate pure states $\ket{\psi_{t, A}}^{\ox s}$ sampled from a state $4$-design ensemble and apply the unitary channels $\cE^{\ox s}$. 

Measure $(\cE(\ketbra{\psi_{t, A}}))^{\ox s}$ with the POVM $\cM_s$ and obtain the result $\ket{\phi_{t, A}}$. 

Randomly generate pure states $\ket{\psi_{t, B}}^{\ox s}$ sampled from a state $4$-design ensemble and apply the unitary channels $\cF^{\ox s}$. 

Measure $(\cF(\ketbra{\psi_{t, B}}))^{\ox s}$ with the POVM $\cM_s$ and obtain the result $\ket{\phi_{t, B}}$.

Compute and store the classical snapshots with Eq.~\eqref{eq:classical snapshot}:
\begin{align}
\tilde{X}_t 
= {\rm Snap}\left(\psi_{t,A}, \phi_{t,A}, s\right), \quad 
\tilde{Y}_t 
= {\rm Snap}\left(\psi_{t,B}, \phi_{t,B}, s\right). 
\end{align}

}
Return
\begin{align}
\tilde{\gamma} = \frac{1}{T^2 d^2} \sum_{i,j=1}^{T} \tr\left[\tilde{X}_i \tilde{Y}_j\right]. 
\end{align}
\end{algorithm}

The independent classical shadow approach for DSEC proceeds as follows (see Algorithm~\ref{alg:dseu with independent classical shadow} for details).
Suppose $\cE(\cdot) = U(\cdot)U^\dagger$ and $\cF(\cdot) = V(\cdot)V^\dagger$ are unitary channels.
On each device \emph{independently}, we randomly select a state $\ket{\psi}$ from a state $4$-design ensemble, apply $U^{\ox s}$ (or $V^{\ox s}$) to $\ket{\psi}^{\ox s}$, and perform the symmetric collective measurement $\cM_s$. 
Repeating this procedure $T$ times on each device yields two independent sets of classical snapshots, $\{\tilde{X}_t\}$ and $\{\tilde{Y}_t\}$.
Each classical snapshot is defined as 
\begin{align}
\tilde{X} = {\rm Snap}(\psi,\phi, s)
:= \frac{d(d+1)(d+s){\phi}\ox {\psi}^T - (d+1+s)(I\ox I)}{s}, 
\label{eq:classical snapshot}
\end{align}
similar to $\tilde{Y}$, 
where $\ket{\psi}$ is sampled from a state $4$-design ensemble, and $\ket{\phi}$ is the measurement result of the POVM $\cM_s$.
Here and below, $\phi:=\ketbra{\phi}$ and $\psi:=\ketbra{\psi}$.
The inner product is then estimated via
\begin{align}
\tilde{\gamma} 
= \frac{1}{T^2 d^2}\sum_{i,j=1}^T \tr[\tilde{X}_i \tilde{Y}_j]. 
\label{eq:cseu estimator}
\end{align}

For $s>1$, each device uses symmetric collective measurements, as in our coherent algorithm; for $s=1$, each output is measured separately. 
The following proposition, proven in Appendix~\ref{app:dseu with cseu}, establishes the query complexity.

\begin{proposition}
\label{the:independent classical shadow upper bound}
For unitary channels $\cE$ and $\cF$, the expectation and variance of estimator $\tilde{\gamma}$ defined in Eq.~\eqref{eq:cseu estimator} are given by
\begin{align}
\bE\tilde{\gamma} = \frac{1}{d^2} \tr[J_\cE J_\cF], \quad 
\Var(\tilde{\gamma}) \leq \cO\left(\frac{(d+s)^4}{T^2s^4} + \frac{1}{T} \left(\frac1s+\frac1d\right)\right). 
\end{align}
Consequently, taking $s=d$ and $T=\Omega(\max\{1/\varepsilon,1/(d\varepsilon^2)\})$ suffices to achieve additive error $\varepsilon$ with success probability at least $2/3$. 
Then, the query complexity is $\cO(\max\{d/\varepsilon,1/\varepsilon^2\})$.

For $s=1$, the same estimator extends to unital channels, with query complexity $\cO(\max\{d^2/\varepsilon,1/\varepsilon^2\})$.
\end{proposition}

For $s=1$, the independent channel-shadow estimator uses incoherent access and has a dimension-dependent query upper bound of $d^2/\varepsilon$ for unital channels.
Our incoherent algorithm improves this term to $\sqrt d/\varepsilon$ while applying to \emph{arbitrary} channels.
For unitary channels, collective measurements with $s=d$ give the independent-shadow bound $d/\varepsilon$, compared with $\sqrt d/\varepsilon$ for our coherent algorithm.
These improvements arise from the shared randomness that coordinates the two devices' SPAM settings. Our coherent algorithm shares randomness for state preparation, while our incoherent algorithm additionally shares measurement settings.

\section{Discussion}
\label{sec:discussion}

In this work, we characterized the optimal query complexity 
of distributed similarity estimation of quantum channels (DSEC). 
We proved that $\Theta(\max\{\sqrt{d}/\varepsilon, 1/\varepsilon^2\})$ queries are both necessary and sufficient for this task. 
The lower bound holds even in the strongest model considered here (multi-round LOCC, coherent access, arbitrary ancillas, and adaptive strategies), while the upper bound is achieved by a non-adaptive, ancilla-free incoherent protocol.
For the upper bound, we proposed two randomized measurement-based algorithms that effectively utilize shared randomness between devices.
The incoherent algorithm achieves optimal query complexity for general channels.
Relative to the independent channel-shadow estimators analyzed here, our incoherent protocol improves the dimension-dependent query upper bound from $d^2/\varepsilon$ to $\sqrt d/\varepsilon$ on unital channels.
Our coherent protocol improves the corresponding bound for unitary channels from $d/\varepsilon$ to $\sqrt d/\varepsilon$.
These advantages arise from the shared randomness employed in our algorithms,
which allows the two devices to coordinate SPAM settings.

Several directions remain open for future research.
First, for unitary channels, unitary tomography~\cite{haah2023queryoptimal} yields an $\cO(d^2/\varepsilon)$ upper bound for DSEC.
Can a distributed protocol with coherent access achieve $\cO(\sqrt d/\varepsilon)$ queries per device, combining inverse-linear accuracy dependence with the optimal dimension dependence?
Under the same coherent local access and LOCC model, what is the optimal query complexity for estimating the diamond distance between two unknown unitary channels?
In particular, is it $\Theta(\sqrt d/\varepsilon)$ per device for additive error $\varepsilon$?
Second, existing cross-platform verification algorithms typically assume identical operations on both platforms, which may fail in practice; developing robust verification algorithms that relax this assumption is an important goal.
Third, since classical shadow methods utilize information about sampled unitaries, investigating whether this information can be better exploited to reduce query complexity is an intriguing question.

\paragraph*{Note added.} 
While completing this manuscript, we became aware of a related work by Ananth \emph{et al.}~\cite{ananth2026limitations}, which also studied the lower bound for completing Problem~\ref{pro:distinguish problem} with coherent access.
While both works identify the $\Omega(\sqrt{d})$ barrier, the proof techniques are distinct and complementary.
Our approach relies on Le Cam's two-point method and connections to optimal quantum cloning~\cite{harrow2013churcha}, 
whereas Ananth \emph{et al.} utilize linear programming. 
This diversity of techniques enriches our understanding of the fundamental hardness of this problem.
In addition, our result proves the full $\varepsilon$-dependent lower bound for DSEC.

\paragraph{Acknowledgements.} 
We thank Kean Chen for helpful discussions on the incoherent algorithm.
We also thank the anonymous reviewers for their careful reading and constructive comments, which helped improve the rigor and clarity of the manuscript.
This work was supported by 
the National Natural Science Foundation of China (Grant Nos. 62471126 and 12504584), 
the Jiangsu Frontier Technology Research and Development Plan (Grant No. BF2025066), 
the Fundamental Research Funds for the Central Universities (Grant No. 2242022k60001), 
the SEU Innovation Capability Enhancement Plan for Doctoral Students (Grant No. CXJH\_SEU 24083), and 
the Quantum Science and Technology-National Science and Technology Major Project (Grant No. 2025ZD0300300). 


\appendix
\addtocontents{toc}{\protect\setcounter{tocdepth}{1}} 

\section{Detailed proofs}

\subsection{Proof of Lemma~\ref{lem:upper bound for two terms in incoherent}}
\label{app:upper bound for two terms in incoherent lower bound}

Before the proof, we introduce some useful notation.
Each root-to-leaf path in the learning tree $\cT$ is uniquely specified by a sequence of vertices $v_0, v_1, \cdots, v_T$, where $v_0 = r$ is the root and $v_T = \ell$ is a leaf. 
Since $\ell$ uniquely determines the shortest path back to the root, specifying the leaf $\ell$ is equivalent to specifying the entire path $v_0 = r, v_1, \cdots, v_{T-1}, v_T = \ell$. 
Therefore, the probabilities of reaching leaf $\ell$ can be expressed as
\begin{align}
p^\cD(\ell) 
&= \prod_{t=1}^T w_t d^2d'^2 \bra{\psi_{t,1}} \left(\cD\ox\cI_{\rm aux}\right) (\ketbra{\phi_{t,1}}) \ket{\psi_{t,1}} \bra{\psi_{t,2}} \left(\cD \ox \cI_{\rm aux}\right) (\ketbra{\phi_{t,2}}) \ket{\psi_{t,2}}, \\
p^{U,V}(\ell) 
&= \prod_{t=1}^T w_t d^2d'^2 \bra{\psi_{t,1}} \left(\cU\ox\cI_{\rm aux}\right) (\ketbra{\phi_{t,1}}) \ket{\psi_{t,1}} \bra{\psi_{t,2}} \left(\cV \ox \cI_{\rm aux}\right) (\ketbra{\phi_{t,2}}) \ket{\psi_{t,2}}. 
\end{align}
An analogous expression holds for $p^{U,U}(\ell)$, with the second occurrence of $\cV$ replaced by $\cU$.
As shown in~\cite{chen2023unitarity}, we can decompose the $\ket{\phi_{t, i}}$ and $\ket{\psi_{t, i}}$ as 
\begin{align}
\ket{\phi_{t, i}} = \sum_{j=0}^{d'-1} \ket{\phi_{t, i, j}}\ox \ket{j}, \quad 
\ket{\psi_{t, i}} = \sum_{j=0}^{d'-1} \ket{\psi_{t, i, j}}\ox \ket{j}, 
\quad i = 1, 2, 
\end{align}
where the first tensor factor corresponds to the main system $\cH_{\rm main}$ and the second to the auxiliary system $\cH_{\rm aux}$.
Note that the vectors $\ket{\phi_{t,i,j}}$ and $\ket{\psi_{t,i,j}}$ are not required to be normalized.
For convenience, let $\bm{j} := (\bm{j}_1, \cdots, \bm{j}_T)$ range over all sequences of length $T$, with each $j_t \in \{0,\ldots,d'-1\}$.
Define $W_\ell := \prod w_t$, and 
\begin{align}
\ket{\Phi_{\ell, i, \bm{j}}} 
:= \bigotimes_{t=1}^T \ket{\phi_{t,i,\bm{j}_t}}, \quad 
\ket{\Psi_{\ell, i, \bm{j}}} 
:= \bigotimes_{t=1}^T \ket{\psi_{t,i,\bm{j}_t}}, \quad 
i = 1, 2, 
\end{align}
we have 
\begin{align}
\bE_{U, V\sim\mu_H} p^{U,V}(\ell) 
&= \bE_{U, V\sim\mu_H} \sum_{\bm{j}, \bm{k}, \bm{x}, \bm{y}} \prod_{t=1}^T 
\left[ w_t d^2d'^2 \times \bra{\psi_{t,1,\bm{j}_t}} \cU (\ket{\phi_{t,1,\bm{j}_t}} \! \bra{\phi_{t,1,\bm{k}_t}}) \ket{\psi_{t,1,\bm{k}_t}} \right. \notag \\
&\quad\quad\quad\quad\quad\quad\quad\quad\quad 
\left. \times \bra{\psi_{t,2,\bm{x}_t}} \cV (\ket{\phi_{t,2,\bm{x}_t}}\! \bra{\phi_{t,2,\bm{y}_t}}) \ket{\psi_{t,2,\bm{y}_t}} \right]\\
&= (dd')^{2T} W_\ell \sum_{\bm{j}, \bm{k}, \bm{x}, \bm{y}} \bE_{U,V\sim\mu_H} 
\left[ \bra{\Psi_{\ell ,1,\bm{j}}} U^{\ox T}\ket{\Phi_{\ell,1,\bm{j}}} \! 
\bra{\Phi_{\ell,1,\bm{k}}} U^{\dagger\ox T} \ket{\Psi_{\ell,1,\bm{k}}}\right. \notag \\
&\quad\quad\quad\quad\quad\quad\quad\quad\quad 
\left. \times \bra{\Psi_{\ell,2,\bm{x}}} V^{\ox T} \ket{\Phi_{\ell,2,\bm{x}}}\! 
\bra{\Phi_{\ell,2,\bm{y}}} V^{\dagger\ox T} \ket{\Psi_{\ell,2,\bm{y}}} \right] \notag \\
&= (dd')^{2T} W_\ell \prod_{i=1}^2 \left( \sum_{\bm{j}, \bm{k}}  \bra{\Psi_{\ell, i, \bm{j}}\ox \Phi_{\ell, i, \bm{j}}^*} J_H^{(T)} \ket{\Psi_{\ell, i, \bm{k}}\ox \Phi_{\ell, i, \bm{k}}^*} \right), 
\end{align}
where $J_H^{(T)}$ is the Choi operator of Haar random channel defined in Eq.~\eqref{eq:choi state of haar random channel}. 
Similarly, we have 
\begin{align}
\bE_{U\sim \mu_H} p^{U,U}(\ell) 
&= (dd')^{2T} W_\ell \sum_{\bm{j}, \bm{k}} \bra{\Psi_{\ell, \bm{j}}\ox \Phi_{\ell, \bm{j}}^{*}} J_H^{(2T)} \ket{\Psi_{\ell, \bm{k}}\ox \Phi_{\ell, \bm{k}}^{*}}, \\
p^{\cD}(\ell) 
&= d'^{2T} W_\ell \sum_{\bm{j}, \bm{k}} \braket{\Psi_{\ell, \bm{j}}}{\Psi_{\ell, \bm{k}}} \braket{\Phi_{\ell, \bm{k}}}{\Phi_{\ell, \bm{j}}}, 
\end{align}
where 
\begin{align}
\ket{\Phi_{\ell,\bm{j}}} 
:= \ket{\Phi_{\ell,1,\bm{j}'}} \ox \ket{\Phi_{\ell,2,\bm{j}''}}, \quad 
\ket{\Psi_{\ell,\bm{k}}} 
:= \ket{\Psi_{\ell,1,\bm{k}'}}\ox \ket{\Psi_{\ell,2,\bm{k}''}}, 
\end{align}
$\bm{j}' = (\bm{j}_1, \cdots, \bm{j}_{T})$, $\bm{j}'' = (\bm{j}_{T+1}, \cdots, \bm{j}_{2T})$, similar to $\bm{k}'$ and $\bm{k}''$. 
Now, we are ready to prove the upper bound for two terms defined in Eq.~\eqref{eq:incoherent lower bound two terms} as follows.

\begin{enumerate}
\item The first bound is trivial when $4T^2\geq d$, since TVD is at most one. Assume $4T^2<d$; then the degree-$2T$ prefactor in Lemma~\ref{lem:lemma 2 of chen2023unitarity} is greater than $1/2$. 
For the first term, we use the standard variational form of total variation distance:
\begin{align}
\norm{\bE_{U\sim \mu_H} p^{U, U} - p^\cD}{{\rm TV}}
&= \frac{1}{2}\sum_{\ell} \left\vert \bE_{U\sim \mu_H} p^{U,U}(\ell) - p^{\cD}(\ell)\right\vert \\
&= \sum_{\ell: p^{\cD}(\ell) \geq \bE_{U\sim \mu_H} p^{U,U}(\ell)} p^{\cD}(\ell) \left[1 - \frac{\bE_{U\sim \mu_H} p^{U,U}(\ell)}{p^{\cD}(\ell)}\right]. 
\end{align}
Therefore, we can focus on the lower bound for $\bE_{U\sim\mu_H} p^{U,U}(\ell) / p^\cD(\ell)$ and have 
\begin{align}
\bE_{U\sim \mu_H} p^{U,U}(\ell) 
&= (dd')^{2T} W_\ell \sum_{\bm{j}, \bm{k}} \bra{\Psi_{\ell, \bm{j}}\ox \Phi_{\ell, \bm{j}}^{*}} J_H^{(2T)} \ket{\Psi_{\ell, \bm{k}}\ox \Phi_{\ell, \bm{k}}^{*}} \\
&\geq \frac{(dd')^{2T} W_\ell}{d(d+1)\cdots (d+2T-1)} 
\sum_{\bm{j}, \bm{k}} \bra{\Psi_{\ell, \bm{j}}\ox \Phi_{\ell, \bm{j}}^{*}} \left(\sum_{\sigma\in\cS_{2T}} P_\sigma\ox P_\sigma\right) \ket{\Psi_{\ell, \bm{k}}\ox \Phi_{\ell, \bm{k}}^{*}} \tag*{Lemma~\ref{lem:lemma 2 of chen2023unitarity}} \\
&= \frac{(dd')^{2T} W_\ell}{d(d+1)\cdots (d+2T-1)} \sum_{\bm{j}, \bm{k}} 
\bra{\Psi_{\ell, \bm{j}}} \left(\sum_{\sigma\in\cS_{2T}} P_\sigma \ket{\Psi_{\ell,\bm{k}}}\!\bra{\Phi_{\ell, \bm{k}}} P_\sigma^\dagger \right) \ket{\Phi_{\ell, \bm{j}}}. 
\end{align}
Define $\cP_\pi: |X_1\rrangle\ox|X_2\rrangle\ox \cdots\ox |X_{k}\rrangle \mapsto |X_{\pi^{-1}(1)}\rrangle\ox|X_{\pi^{-1}(2)}\rrangle\ox \cdots\ox |X_{\pi^{-1}(k)}\rrangle$ for $\pi \in \cS_k$~\cite{chen2023unitarity}, then, we have 
\begin{align}
\bE_{U\sim \mu_H} p^{U,U}(\ell) 
&\geq \frac{(dd')^{2T} W_\ell}{d(d+1)\cdots (d+2T-1)} \sum_{\bm{j}, \bm{k}} 
\llangle \ket{\Psi_{\ell, \bm{j}}}\!\bra{\Phi_{\ell, \bm{j}}} | \left(\sum_{\sigma\in\cS_{2T}} \cP_\sigma \right) | \ket{\Psi_{\ell, \bm{k}}}\!\bra{\Phi_{\ell, \bm{k}}} \rrangle \\
&\geq \frac{(dd')^{2T} W_\ell}{d(d+1)\cdots (d+2T-1)} \sum_{\bm{j}, \bm{k}} \braket{\Psi_{\ell, \bm{j}}}{\Psi_{\ell, \bm{k}}} \braket{\Phi_{\ell, \bm{k}}}{\Phi_{\ell, \bm{j}}} 
\tag*{Lemma~\ref{lem:tensor product and permutation operators}}. 
\end{align}
Thus, for each leaf $\ell$, we have 
\begin{align}
\frac{\bE_{U\sim\mu_H} p^{U,U}(\ell)}{p^\cD(\ell)} 
&\geq \frac{d^{2T}}{d(d+1)\cdots (d+2T-1)} 
= \prod_{t=1}^{2T} \left(1 + \frac{t-1}{d}\right)^{-1} \\
&\geq \prod_{t=1}^{2T} \left(1 - \frac{t-1}{d}\right) 
\geq \left(1 - \frac{2T}{d}\right)^{2T} \geq 1 - \frac{4T^2}{d}. 
\end{align}
Therefore, we have 
\begin{align}
\norm{\bE_{U\sim \mu_H} p^{U, U} - p^\cD}{{\rm TV}}
&= \sum_{\ell: p^{\cD}(\ell) \geq \bE_{U\sim \mu_H} p^{U,U}(\ell)} p^{\cD}(\ell) \left[1 - \frac{\bE_{U\sim \mu_H} p^{U,U}(\ell)}{p^{\cD}(\ell)}\right] 
\leq \frac{4 T^2}{d}. 
\end{align}
\item The second bound is trivial when $2T^2\geq d$. Assume $2T^2<d$, we have 
\begin{align}
\bE_{U, V\sim\mu_H} p^{U,V}(\ell) 
&= (dd')^{2T} W_\ell \prod_{i=1}^2 \left( \sum_{\bm{j}, \bm{k}}  \bra{\Psi_{\ell, i, \bm{j}}\ox \Phi_{\ell, i, \bm{j}}^*} 
J_H^{(T)} \ket{\Psi_{\ell, i, \bm{k}}\ox \Phi_{\ell, i, \bm{k}}^*} \right) \\
&\geq \left(\frac{(dd')^T}{d(d+1)\cdots(d+T-1)}\right)^2 W_\ell \prod_{i=1}^2 
\left(\sum_{\bm{j}, \bm{k}} \braket{\Psi_{\ell,i,\bm{j}}}{\Psi_{\ell,i,\bm{k}}} \braket{\Phi_{\ell,i,\bm{k}}}{\Phi_{\ell,i,\bm{j}}} \right) \\
&= \left(\frac{(dd')^T}{d(d+1)\cdots(d+T-1)}\right)^2 W_\ell \sum_{\bm{j}, \bm{k}} \braket{\Psi_{\ell, \bm{j}}}{\Psi_{\ell, \bm{k}}} \braket{\Phi_{\ell, \bm{k}}}{\Phi_{\ell, \bm{j}}} 
\end{align}
Thus, for each leaf $\ell$, we have 
\begin{align}
\frac{\bE_{U, V\sim\mu_H} p^{U,V}(\ell)}{p^\cD(\ell)} 
&\geq \left(\frac{d^T}{d(d+1)\cdots(d+T-1)}\right)^2 
\geq \left(1 - \frac{T^2}{d}\right)^2. 
\end{align}
Therefore, we have 
\begin{align}
\norm{p^\cD - \bE_{U,V\sim \mu_H} p^{U, V}}{{\rm TV}}
\leq 1 - \left(1 - \frac{T^2}{d}\right)^2
\leq \frac{2T^2}{d}.
\end{align}
\end{enumerate}

\begin{lemma}[Lemma 5.12 in~\cite{chen2022exponential} and Lemma 3 in~\cite{chen2023unitarity}]
\label{lem:tensor product and permutation operators}
For any product vector $\ket{X} = \bigotimes_{t=1}^k \ket{x_t}$, we have 
\begin{align}
\bra{X} \sum_{\sigma\in\cS_k} P_\sigma \ket{X} 
\geq \braket{X}{X}. 
\end{align}
\end{lemma}

\subsection{Proof of Theorem~\ref{the:coherent upper bound}}
\label{app:upper bound for coherent access}

To obtain the expectation of $\tilde{\chi}$, we compute the expectation of $\tilde{f}$ first:
\begin{align}
\bE \tilde{f} 
&= \bE |\braket{\phi_A}{\phi_B}|^2 
= \bE \tr\left[\left(\ketbra{\phi_A}\right) \left(\ketbra{\phi_B}\right)\right] \\
&= \bE_{\psi} \tr\left[\left(\frac{I + T U\ketbra{\psi}U^\dagger}{d+T}\right) \left(\frac{I + T V\ketbra{\psi}V^\dagger}{d+T}\right)\right] 
\tag*{Lemma~\ref{lem:expectation and second moment of collective measurement result}}\\
&= \frac{d + 2T}{(d+T)^2} + \frac{T^2}{(d+T)^2} \bE_{\psi} \tr\left[\left(U^\dagger V \ox V^\dagger U\right) \ketbra{\psi}^{\ox 2}\right] \\
&= \frac{d + 2T}{(d+T)^2} + \frac{T^2\left(\left|\tr[U^\dagger V]\right|^2 + d\right)}{d(d+1)(d+T)^2}. 
\tag*{Lemma~\ref{lem:special cases of weingarten calculus}}
\end{align}
As we can see, the proof relies on Lemma~\ref{lem:expectation and second moment of collective measurement result}, which requires that the output states remain \emph{pure}.
For a general quantum channel, the output states might be mixed, and the key identities in the lemma no longer hold.
Therefore, the coherent algorithm is applicable only when the unknown channels are unitary channels.
Lastly, we can prove $\tilde{\chi}$ is an unbiased estimator for $\tr[J_U J_V] / d^2$, 
\begin{align}
\bE \tilde{\chi} 
= \frac{(d+1)(d+T)^2}{T^2 d} \bE \tilde{f} - \frac{(d+1)(d+2T) + T^2}{T^2 d} 
= \frac{\left|\tr[U^\dagger V]\right|^2}{d^2} = \frac{\tr[J_U J_V]}{d^2}.
\end{align}
This uses $\tr[J_UJ_V]=|\tr[U^\dagger V]|^2$, which is the unitary-channel special case of Eq.~\eqref{eq:choi and kraus}.

We now consider the variance of the estimator $\tilde{\chi}$. 
Let $\tilde{\chi} = X \tilde{f} - Y$, where 
\begin{align}
X = \frac{(d+1)(d+T)^2}{T^2 d}, \quad 
Y = \frac{(d+1)(d+2T) + T^2}{T^2 d}.
\end{align}
Then, we have 
\begin{align}
\Var (\tilde{\chi})
&= X^2 \bE \tilde{f}^2 - 2XY\bE\tilde{f} + Y^2 - \frac{\left|\tr[U^\dagger V]\right|^4}{d^4}
\leq X^2 \bE \tilde{f}^2 + Y^2 - \frac{\left|\tr[U^\dagger V]\right|^4}{d^4}. 
\end{align}
We consider the first two terms, respectively.
With~\cite[Proof of Lemma 5, Eqs.~(165) and~(166)]{anshu2022distributed}, 
we have 
\begin{align}
X^2 \bE \tilde{f}^2 
&\leq \left(\frac{d+1}{d}\right)^2 \bE_\psi \left(f_\psi^2 + \frac{8f_\psi - 2f_\psi^2}{T} + \frac{2df_\psi + f_\psi^2 + 8 + 2d}{T^2} + \frac{8d+4}{T^3} + \frac{2d^2 + 2d}{T^4}\right), 
\end{align}
where we define 
\begin{align}
f_\psi := \tr\left[\left(U\ketbra{\psi}U^\dagger\right) \left(V\ketbra{\psi}V^\dagger\right) \right]. 
\end{align}
With Lemma~\ref{lem:special cases of weingarten calculus}, we have 
\begin{align}
\bE_\psi f_\psi = \frac{\left|\tr[U^\dagger V]\right|^2 + d}{d(d+1)} \leq 1, 
\end{align}
and 
\begin{align}
\bE_\psi f_\psi^2 
&= \bE_\psi \tr\left[\left(U^\dagger V \ox V^\dagger U\right)^{\ox 2}\ketbra{\psi}^{\ox 4}\right] 
= \bE_\psi \tr\left[\left(U^\dagger V \ox V^\dagger U\right)^{\ox 2}\frac{\Pi_{\rm sym}^{(4)}}{\kappa_4}\right] \\
&= \frac{1}{d(d+1)(d+2)(d+3)} \sum_{\pi\in\cS_4} \tr\left[\left(U^\dagger V\ox V^\dagger U\right)^{\ox 2} P_\pi\right] 
\leq \frac{\left|\tr[U^\dagger V]\right|^4}{d^4} + \cO\left(\frac{1}{d}\right). 
\end{align}
For the second term, we have
\begin{align}
Y^2 
= \left(\frac{d}{T^2} + \frac{2T+1}{T^2} + \frac{2}{Td} + \frac{1}{d}\right)^2 
= \cO\left(\frac{d^2}{T^4} + \frac{1}{T^2} + \frac{1}{d^2}\right).
\end{align}
Thus, the variance of $\tilde{\chi}$ is upper bounded by 
\begin{align}
\Var (\tilde{\chi}) 
&\leq \cO\left(\frac{1}{T} + \frac{d}{T^2} + \frac{d}{T^3} + \frac{d^2}{T^4}+\frac1d\right) 
+ \cO\left(\frac{d^2}{T^4} + \frac{1}{T^2} + \frac{1}{d^2}\right) \notag \\
&= \cO\left(\frac{1}{T} + \frac{d}{T^2} + \frac{d^2}{T^4}+\frac1d\right). 
\end{align}
The $K$ repetitions use independent inputs and measurement outcomes, so $\sum_{k=1}^K\tilde{\chi}_{k}/K$ is unbiased and has variance $\Var(\tilde{\chi})/K$.
A sufficient number of repetitions is obtained as follows.
\begin{enumerate}
\item If $\varepsilon\geq d^{-1/2}$, then $T=\lceil\sqrt d/\varepsilon\rceil$, and the single-input variance is $\cO(\varepsilon^2)$. A sufficiently large constant number of independent repetitions makes the variance of the average at most $\varepsilon^2/3$, using $KT=\cO(\sqrt d/\varepsilon)$ queries per device.
\item If $0<\varepsilon<d^{-1/2}$, then $T=d$, and the single-input variance is $\cO(1/d)$. In this case, $\cO(1/(d\varepsilon^2))$ independent repetitions suffice to make the variance of the average at most $\varepsilon^2/3$, using $KT=\cO(1/\varepsilon^2)$ queries per device.
\end{enumerate}
Therefore, the total query count per device is
\begin{align*}
KT=\cO\left(\max\left\{\frac1{\varepsilon^2},\frac{\sqrt d}{\varepsilon}\right\}\right).
\end{align*}

\subsection{Proof of Theorem~\ref{the:incoherent upper bound}}
\label{app:upper bound for incoherent access}

We first show that $\tilde{\omega}_t$ is an unbiased estimator for $\tr[J_\cE J_\cF] / d^2$ with the following lemma.
\begin{lemma}
\label{lem:expectation of omega}
The expectation of $\tilde{\omega}_t$ in Algorithm~\ref{alg:incoherent access} is given by 
\begin{align}
\bE\; \tilde{\omega}_t 
= \frac{\tr[J_\cE J_\cF]}{d^2}.
\end{align}
That is, $\tilde{\omega}_t$ is an unbiased estimator for $\tr[J_\cE J_\cF] / d^2$.
\end{lemma}

\begin{proof}[Proof of Lemma~\ref{lem:expectation of omega}]
We prove this lemma by analyzing the expectation of $\tilde{g}_t$ and $\tilde{o}_t$, respectively.

\paragraph*{Expectation of $\tilde{g}$.}
Note that to obtain $\tilde{g}$, we generate the same random SPAM setting $\{Q, \psi\}$ for both devices.
For each SPAM setting $\{Q, \psi\}$, we define the following two probabilities: 
\begin{align}
p_{Q, \psi}(a) := \bra{a}Q \cE(\ketbra{\psi}) Q^\dagger \ket{a}, \quad 
q_{Q, \psi}(b) := \bra{b}Q \cF(\ketbra{\psi}) Q^\dagger \ket{b}, 
\end{align}
and their classical inner product function $g(Q, \psi) := \sum_a p_{Q, \psi}(a) q_{Q, \psi}(a)$.
Then, the expectation of $\tilde{g}$ can be represented as 
\begin{align}
\bE_{Q,\psi,A,B}\; \tilde{g} 
&= \bE_{Q, \psi}\; g(Q, \psi) 
= \bE_{Q, \psi}\; \sum_a \bra{a}Q \cE(\ketbra{\psi}) Q^\dagger \ket{a} \bra{a}Q \cF(\ketbra{\psi}) Q^\dagger \ket{a} \\
&= \bE_{\psi} \tr\left[\left(\sum_a \bE_{Q} Q^{\dagger\ox 2}\ketbra{aa}Q^{\ox 2}\right) \left(\cE(\ketbra{\psi})\ox \cF(\ketbra{\psi})\right)\right]\\
&= \frac{1}{d+1} + \frac{1}{d+1} \bE_{\psi} \tr\left[\cE(\ketbra{\psi}) \cF(\ketbra{\psi}) \right],  
\end{align}
where we have used Lemma~\ref{lem:special cases of weingarten calculus} and the property of unitary $2$-design.
For the last term, suppose that $\cE(\cdot) = \sum_i E_i(\cdot) E_i^\dagger$ and $\cF(\cdot) = \sum_j F_j(\cdot) F_j^\dagger$. Then we have 
\begin{align}
\bE_{\psi} \tr\left[\cE(\ketbra{\psi}) \cF(\ketbra{\psi})\right]
&= \sum_{i,j} \bE_\psi \tr\left[E_i\ketbra{\psi} E_i^\dagger F_j\ketbra{\psi} F_j^\dagger\right] \\
&= \sum_{i,j} \tr\left[\left(E_i^\dagger F_j \ox F_j^\dagger E_i\right)
\left( \bE_\psi \ketbra{\psi}^{\ox 2}\right)\right] \\
&= \sum_{i,j} \tr\left[\left(E_i^\dagger F_j \ox F_j^\dagger E_i\right) \frac{\Pi_{\rm sym}^{(d, 2)}}{\kappa_2}\right] \\
&= \frac{\tr[J_\cE J_\cF]}{d(d+1)} + 
\frac{1}{d(d+1)} \sum_{i,j} \tr\left[E_i E_i^\dagger F_j F_j^\dagger\right] \\
&= \frac{\tr[J_\cE J_\cF]}{d(d+1)} + \frac{d}{d+1} \bE_{\psi, \phi} \tr\left[\cE(\ketbra{\psi}) \cF(\ketbra{\phi})\right].
\end{align}
Therefore, we have 
\begin{align}
\bE_{Q,\psi,A,B}\; \tilde{g}
&= \frac{1}{d+1} + \frac{\tr[J_\cE J_\cF]}{d(d+1)^2} + \frac{d}{(d+1)^2} \bE_{\psi, \phi} \tr\left[\cE(\ketbra{\psi}) \cF(\ketbra{\phi})\right].
\label{eq:expectation of g}
\end{align}

\paragraph*{Remark.} 
Before analyzing the expectation of $\tilde{o}$, we show that if one of the two channels is a unital channel, then, we have 
\begin{align}
\bE_{\psi, \phi} \tr\left[\cE(\ketbra{\psi}) \cF(\ketbra{\phi})\right]
= \frac{1}{d^2} \sum_{i,j} \tr\left[E_i E_i^\dagger F_j F_j^\dagger\right]
= \frac{1}{d}.
\end{align}
Thus, in this case, we have 
\begin{align}
\frac{\tr[J_\cE J_\cF]}{d^2} = \frac{(d+1)^2}{d} \bE_{Q,\psi,A,B}\; \tilde{g} - \frac{d+2}{d}.
\end{align}
That is, we can directly estimate $\tr[J_\cE J_\cF] / d^2$ without computing $\tilde{o}$.

\paragraph*{Expectation of $\tilde{o}$.}
We now turn to analyze the expectation of $\tilde{o}$ and show that its expectation can be connected to $\bE_{\psi, \phi} \tr\left[\cE(\ketbra{\psi}) \cF(\ketbra{\phi})\right]$ in Eq.~\eqref{eq:expectation of g}.
Note that to obtain $\tilde{o}$, we generate two independent random states $\ket{\psi}$ and $\ket{\phi}$ for the two devices, respectively.
Thus, for each SPAM setting $\{Q, \psi, \phi\}$, we can define the following two probabilities:
\begin{align}
p_{Q, \psi}(a) := \bra{a}Q \cE(\ketbra{\psi}) Q^\dagger \ket{a}, \quad 
q_{Q, \phi}(b) := \bra{b}Q \cF(\ketbra{\phi}) Q^\dagger \ket{b},
\end{align}
and their classical inner product function $o(Q, \psi, \phi) := \sum_a p_{Q, \psi}(a) q_{Q, \phi}(a)$.
Likewise, the expectation of $\tilde{o}$ can be represented as
\begin{align}
\bE_{Q,\psi,\phi,A,B}\; \tilde{o}
&= \bE_{Q, \psi, \phi}\; o(Q, \psi, \phi)
= \bE_{Q, \psi, \phi}\; \sum_a \bra{a}Q \cE(\ketbra{\psi}) Q^\dagger \ket{a} \bra{a}Q \cF(\ketbra{\phi}) Q^\dagger \ket{a} \\
&= \bE_{\psi, \phi} \tr\left[\left(\sum_a \bE_Q Q^{\dagger\ox 2} \ketbra{aa}Q^{\ox 2} \right) \left(\cE(\ketbra{\psi}) \ox \cF(\ketbra{\phi})\right)\right] \\
&= \frac{1}{d+1} + \frac{1}{d+1} \bE_{\psi, \phi} \tr\left[\cE(\ketbra{\psi}) \cF(\ketbra{\phi}) \right],
\label{eq:expectation of o}
\end{align}
where we have used Lemma~\ref{lem:special cases of weingarten calculus} and the property of unitary $2$-design.

Therefore, combining Eqs.~\eqref{eq:expectation of g} and~\eqref{eq:expectation of o}, we have
\begin{align}
\bE_{Q,\psi,A,B}\; \tilde{g}
&= \frac{1}{d+1} + \frac{\tr[J_\cE J_\cF]}{d(d+1)^2} + \frac{d}{(d+1)^2}
\left[(d+1)\bE_{Q,\psi,\phi,A,B}\;\tilde{o} - 1\right] \\
&= \frac{1}{(d+1)^2} + \frac{\tr[J_\cE J_\cF]}{d(d+1)^2} + \frac{d}{d+1} \bE_{Q,\psi,\phi,A,B}\; \tilde{o}, \\
\Rightarrow\quad
\frac{\tr[J_\cE J_\cF]}{d^2}
&= \frac{(d+1)^2}{d}\bE_{Q,\psi,A,B}\; \tilde{g} - (d+1) \bE_{Q,\psi,\phi,A,B}\; \tilde{o} - \frac{1}{d} 
= \bE\; \tilde{\omega}. 
\end{align}
\end{proof}

Now, we consider the variance of the estimator so that we can obtain the query complexity of Algorithm~\ref{alg:incoherent access}. 
The following analysis is similar to~\cite{anshu2022distributed}.
With the definition of ${\omega}$, the variance of ${\omega}$ is given by
\begin{align}
\Var({\omega}) 
= \frac{1}{T} \Var(\tilde{\omega}_t), \quad \text{where} \quad
\Var(\tilde{\omega}_t)
= \frac{(d+1)^4}{d^2} \Var(\tilde{g}) + (d+1)^2 \Var(\tilde{o}),
\label{eq:variance of omega}
\end{align}
as $\tilde{g}$ and $\tilde{o}$ are independent.
In the following, we focus on the variances of these two random variables and establish the following lemma. 
\begin{lemma}
\label{lem:variance of g(Q,psi,S)}
The variance of $\tilde{g}$ and $\tilde{o}$ can be bounded as follows:
\begin{align}
\Var(\tilde{g}) 
\leq \cO\left(\frac{1}{m^2 d} + \frac{1}{md^2} + \frac{1}{d^3}\right), \quad
\Var(\tilde{o})
\leq \cO\left(\frac{1}{m^2 d} + \frac{1}{md^2} + \frac{1}{d^3}\right).
\end{align}
\end{lemma}

\begin{proof}[Proof of Lemma~\ref{lem:variance of g(Q,psi,S)}]
We consider these two variances separately as follows.

\paragraph*{Variance of $\tilde{g}$.}
It should be noted that $\tilde{g}$ is a function of the SPAM setting $\{Q, \psi\}$ and measurement results $S := \{A, B\}$. 
Thus, we can write $\tilde{g}$ as $\tilde{g}(Q,\psi,S)$, and the law of total variance gives~\cite{anshu2022distributed} 
\begin{align}
\Var(\tilde{g}(Q, \psi, S)) 
=  \bE_{Q, \psi} \Var\left[\tilde{g}(Q, \psi, S | Q, \psi)\right] + \Var_{Q, \psi} \bE\left[\tilde{g}(Q, \psi, S | Q, \psi)\right]. 
\end{align}
We consider these two terms separately as follows. 
\begin{enumerate}
\item For the first term, with Lemma~14 in~\cite{anshu2022distributed}, we have 
\begin{align}
\bE_{Q, \psi} \Var\left[\tilde{g}(Q, \psi, S | Q, \psi)\right] 
&\leq \bE_{Q, \psi} \left[ \frac{g(Q,\psi)}{m^2} + 
\frac{1}{m}\sum_a \left( p_{Q, \psi}^2(a) q_{Q, \psi}(a) + p_{Q, \psi}(a) q_{Q, \psi}^2(a)\right)\right]. 
\end{align}
As shown in Eq.~\eqref{eq:expectation of g}, we have 
\begin{align}
\bE_{Q, \psi} \frac{g(Q,\psi)}{m^2} 
&= \frac{1}{m^2(d+1)} + \frac{1}{m^2}\frac{\tr[J_\cE J_\cF]}{d(d+1)^2} 
+ \frac{d}{m^2(d+1)^2} \bE_{\psi, \phi} \tr\left[\cE(\ketbra{\psi}) \cF(\ketbra{\phi})\right] \notag \\
&= \cO\left(\frac{1}{m^2 d}\right). 
\end{align}
Additionally, with Eq.~(162) in~\cite{anshu2022distributed}, we have 
\begin{align}
\bE_{Q,\psi}\; \sum_a p_{Q, \psi}^2(a) q_{Q, \psi}(a)
&= \bE_{Q,\psi}\; \sum_a \bra{a}Q \cE(\ketbra{\psi}) Q^\dagger \ket{a}^2 
\bra{a}Q \cF(\ketbra{\psi}) Q^\dagger \ket{a} \\
&= d \bE_{\psi,\phi} \bra{\phi} \cE(\ketbra{\psi}) \ket{\phi}^2
\bra{\phi} \cF(\ketbra{\psi}) \ket{\phi} = \cO\left(\frac{1}{d^2} \right),
\end{align}
where we use the property of unitary $3$-design. 
Likewise, we have 
\begin{align}
\bE_{Q, \psi}\; \sum_a p_{Q, \psi}(a) q_{Q, \psi}^2(a) = \cO\left(\frac{1}{d^2}\right). 
\end{align}
Therefore, we have 
\begin{align}
\bE_{Q, \psi} \Var\left[\tilde{g}(Q, \psi, S | Q, \psi)\right] 
\leq \cO\left(\frac{1}{m^2 d} + \frac{1}{md^2}\right). 
\end{align}
\item 
For the second term, we have
\begin{align}
&\;\Var_{Q, \psi} \bE\left[\tilde{g}(Q, \psi, S | Q, \psi)\right] \notag \\
&= \Var_{Q, \psi}\; g(Q, \psi) 
= \bE_{Q, \psi}\; g^2(Q, \psi) - \left[\bE_{Q, \psi} \; g(Q, \psi) \right]^2 \\
&= \bE_{Q, \psi}\; g^2(Q, \psi) - \bE_\psi \left(\frac{1 + \tr[\cE(\ketbra{\psi})\cF(\ketbra{\psi})]}{d+1}\right)^2 +\frac{\Var_\psi\tr[\cE(\ketbra{\psi}) \cF(\ketbra{\psi})] }{(d+1)^2}. 
\end{align}
We have 
\begin{align}
\bE_{Q}\; g^2(Q, \psi) 
=&\; \bE_{Q}\; \left(\sum_a \bra{a}Q \cE(\ketbra{\psi}) Q^\dagger \ket{a} \bra{a}Q \cF(\ketbra{\psi}) Q^\dagger \ket{a}\right)^2 \\
=&\; d\; \bE_{\phi} \bra{\phi} \cE(\ketbra{\psi}) \ket{\phi}^2 
\bra{\phi} \cF(\ketbra{\psi}) \ket{\phi}^2 \notag \\
&+ d(d-1) \bE_{\phi, \phi^\bot} \bra{\phi\phi^\bot} (\cE(\ketbra{\psi}))^{\ox 2} \ket{\phi \phi^\bot} 
\bra{\phi\phi^\bot} (\cF(\ketbra{\psi}))^{\ox 2} \ket{\phi\phi^\bot}, 
\end{align}
where $\phi^\bot$ is randomly sampled from the orthogonal space of $\phi$, i.e., $\braket{\phi}{\phi^\bot} = 0$. 
Using the fourth-moment identity and the orthogonal-vector calculation
in~\cite[Appendix~C, Lemma~22 and Eqs.~(159), (169)]{anshu2022distributed},
we obtain the following bounds for arbitrary output states:
\begin{align}
&d\; \bE_{\phi} \bra{\phi} \cE(\ketbra{\psi}) \ket{\phi}^2
\bra{\phi} \cF(\ketbra{\psi}) \ket{\phi}^2 
 = d \cO\left(\frac{1}{d^4}\right) = \cO\left(\frac{1}{d^3}\right), \\
&d(d-1) \bE \bra{\phi\phi^\bot} (\cE(\ketbra{\psi}))^{\ox 2} \ket{\phi \phi^\bot} \bra{\phi\phi^\bot} (\cF(\ketbra{\psi}))^{\ox 2} \ket{\phi\phi^\bot} \notag \\
&\quad\quad\quad\quad\quad\quad\quad\quad\quad\quad\quad\quad 
\leq \frac{(1 + \tr[\cE(\ketbra{\psi}) \cF(\ketbra{\psi})])^2}{d(d+1)} + \cO\left(\frac{1}{d^4}\right).
\end{align}
With Lemma~\ref{lem:variance of input correlations}, we have 
\begin{align}
\frac{\Var_\psi\tr[\cE(\ketbra{\psi})\cF(\ketbra{\psi})]}{(d+1)^2}
\leq \frac{4}{(d+1)^3}. 
\end{align}
Therefore, we have 
\begin{align}
&\Var_{Q, \psi} \bE\left[\tilde{g}(Q, \psi, S | Q, \psi)\right] \notag \\
\leq &\; \cO\left(\frac{1}{d^3}\right) + \left(\frac{1}{d(d+1)} - \frac{1}{(d+1)^2}\right) \bE_\psi (1 + \tr[\cE(\ketbra{\psi}) \cF(\ketbra{\psi})])^2 + \frac{4}{(d+1)^3} \\
=&\; \cO\left(\frac{1}{d^3}\right). 
\end{align}
\end{enumerate}
Combining the above two terms, we have 
\begin{align}
\Var(\tilde{g}(Q, \psi, S)) 
\leq \cO\left(\frac{1}{m^2 d} + \frac{1}{md^2} + \frac{1}{d^3}\right). 
\end{align}

\paragraph*{Variance of $\tilde{o}$.}
We now analyze the variance of $\tilde{o}$, which is similar to the analysis of $\tilde{g}$. 
$\tilde{o}$ is a function of the SPAM setting $\{Q, \psi, \phi\}$ and measurement results $S := \{A, B\}$.
Thus, we can write $\tilde{o}$ as $\tilde{o}(Q,\psi,\phi,S)$, and the law of total variance gives
\begin{align}
\Var(\tilde{o}(Q, \psi, \phi, S)) 
&=  \bE_{Q, \psi, \phi} \Var\left[\tilde{o}(Q, \psi, \phi, S | Q, \psi, \phi)\right] + \Var_{Q, \psi, \phi} \bE\left[\tilde{o}(Q, \psi, \phi, S | Q, \psi, \phi)\right]. 
\end{align}
We also consider the above two terms separately as follows. 
\begin{enumerate}
\item For the first term, with Lemma~14 in~\cite{anshu2022distributed}, we have 
\begin{align}
\bE_{Q, \psi, \phi} \Var\left[\tilde{o}(Q, \psi, \phi, S | Q, \psi, \phi)\right]
&\leq \bE_{Q, \psi, \phi} \left[ \frac{o(Q,\psi,\phi)}{m^2} \right. \notag \\
&\quad\quad\quad \left. + \frac{1}{m}\sum_a \left( p_{Q, \psi}^2(a) q_{Q, \phi}(a) + p_{Q, \psi}(a) q_{Q, \phi}^2(a)\right)\right]. 
\end{align}
As shown in Eq.~\eqref{eq:expectation of o}, we have 
\begin{align}
\bE_{Q, \psi, \phi} \frac{o(Q,\psi,\phi)}{m^2} 
&= \frac{1}{m^2(d+1)} + \frac{1}{m^2(d+1)} \bE_{\psi, \phi} \tr\left[\cE(\ketbra{\psi}) \cF(\ketbra{\phi}) \right] \notag \\
&= \cO\left(\frac{1}{m^2 d}\right). 
\end{align}
Likewise, with Eq.~(162) in~\cite{anshu2022distributed}, we have 
\begin{align}
\bE_{Q,\psi,\phi} \sum_a p_{Q, \psi}^2(a) q_{Q, \phi}(a) = \cO\left(\frac{1}{d^2}\right), \quad 
\bE_{Q, \psi, \phi} \sum_a p_{Q, \psi}(a) q_{Q, \phi}^2(a) = \cO\left(\frac{1}{d^2}\right). 
\end{align}
Therefore, we have 
\begin{align}
\bE_{Q, \psi, \phi} \Var\left[\tilde{o}(Q, \psi, \phi, S | Q, \psi, \phi)\right] 
\leq \cO\left(\frac{1}{m^2 d} + \frac{1}{md^2}\right). 
\end{align}
\item For the second term, we have
\begin{align}
\Var_{Q, \psi, \phi}\; \bE\left[\tilde{o}(Q, \psi, \phi, S | Q, \psi, \phi)\right] 
&= \Var_{Q, \psi, \phi}\; o(Q, \psi, \phi) \\
&= \bE_{Q, \psi, \phi}\; o^2(Q, \psi, \phi) - \left[\bE_{Q, \psi, \phi} \; o(Q, \psi, \phi) \right]^2 \\
&= \bE_{Q, \psi, \phi}\; o^2(Q, \psi, \phi) \notag \\
&\quad\quad\quad - \bE_{\psi,\phi} \left(\frac{1 + \tr[\cE(\ketbra{\psi})\cF(\ketbra{\phi})]}{d+1}\right)^2\notag\\
&\quad\quad\quad + \frac{\Var_{\psi,\phi} \tr[\cE(\ketbra{\psi})\cF(\ketbra{\phi})]}{(d+1)^2}. 
\end{align}
Similar to the analysis of $\tilde{g}$, using Lemma~\ref{lem:variance of input correlations}, we have
\begin{align}
&\Var_{Q, \psi, \phi}\; \bE\left[\tilde{o}(Q, \psi, \phi, S | Q, \psi, \phi)\right] \notag \\
\leq&\; \cO\left(\frac{1}{d^3}\right) + \left(\frac{1}{d(d+1)} - \frac{1}{(d+1)^2}\right) \bE_{\psi,\phi} (1 + \tr[\cE(\ketbra{\psi}) \cF(\ketbra{\phi})])^2 + \frac{2}{(d+1)^3} \\
\leq&\; \cO\left(\frac{1}{d^3}\right)
\end{align}
\end{enumerate}
Combining the above two terms, we have 
\begin{align}
\Var(\tilde{o}(Q, \psi, \phi, S)) 
\leq \cO\left(\frac{1}{m^2 d} + \frac{1}{md^2} + \frac{1}{d^3}\right). 
\end{align}
\end{proof}

Combining Lemma~\ref{lem:variance of g(Q,psi,S)} with Eq.~\eqref{eq:variance of omega} gives
\begin{align}
\Var({\omega})\leq\cO\left(\frac{d/m^2+1/m+1/d}{T}\right).
\label{eq:incoherent averaged variance}
\end{align}
For $d\geq2$ and $0<\varepsilon<1$, set
\begin{align*}
m=\left\lceil\min\left\{d,\frac{\sqrt d}{\varepsilon}\right\}\right\rceil.
\end{align*}
A sufficient number of repetitions $T$ is obtained as follows.
\begin{enumerate}
\item If $\varepsilon\geq d^{-1/2}$, then $m=\lceil\sqrt d/\varepsilon\rceil$, and the single-repetition variance is $\Var(\tilde{\omega}_t)=\cO(\varepsilon^2)$. A sufficiently large constant number of independent repetitions makes $\Var(\omega)\leq\varepsilon^2/3$, using $2Tm=\cO(\sqrt d/\varepsilon)$ queries per device.
\item If $0<\varepsilon<d^{-1/2}$, then $m=d$, and the single-repetition variance is $\Var(\tilde{\omega}_t)=\cO(1/d)$. In this case, $\cO(1/(d\varepsilon^2))$ independent repetitions suffice to make $\Var(\omega)\leq\varepsilon^2/3$, using $2Tm=\cO(1/\varepsilon^2)$ queries per device.
\end{enumerate}
Chebyshev's inequality gives additive error $\varepsilon$ with success probability at least $2/3$. The total query count per device is
\begin{align*}
2Tm=\cO\left(\max\left\{\frac1{\varepsilon^2},\frac{\sqrt d}{\varepsilon}\right\}\right).
\end{align*}
This proves Theorem~\ref{the:incoherent upper bound}.

\begin{lemma}
\label{lem:variance of input correlations}
Let $\cE,\cF$ be CPTP maps on a $d$-dimensional system, and let $\ket{\psi},\ket{\phi}$ be independent samples from an exact state $4$-design. Then
\begin{align}
\Var_\psi\tr[\cE(\ketbra{\psi})\cF(\ketbra{\psi})]&\leq\frac4{d+1},\\
\Var_{\psi,\phi}\tr[\cE(\ketbra{\psi})\cF(\ketbra{\phi})]&\leq\frac2{d+1}.
\end{align}
\end{lemma}

\begin{proof}
For the same input state, define
\begin{align}
\bF_{\cE,\cF}:=\Pi_{\rm sym}^{(d,2)}(\cE^\dagger\ox\cF^\dagger)(\bF)\Pi_{\rm sym}^{(d,2)},
\end{align}
where $\bF$ swaps the two output systems and $\cE^\dagger,\cF^\dagger$ are the Hilbert--Schmidt adjoints.
These adjoints are completely positive and unital. 
Since $-I\preceq\bF\preceq I$,
\begin{align}
-\Pi_{\rm sym}^{(d,2)}\preceq \bF_{\cE,\cF}\preceq\Pi_{\rm sym}^{(d,2)}, \quad
\tr[\cE(\ketbra{\psi})\cF(\ketbra{\psi})]
=\tr[\bF_{\cE,\cF}\ketbra{\psi}^{\ox2}].
\end{align}
The second and fourth Haar moments give
\begin{align}
\bE_\psi\tr[\bF_{\cE,\cF}\ketbra{\psi}^{\ox2}]
&=\frac{2\tr[\bF_{\cE,\cF}]}{d(d+1)},\\
\bE_\psi\!\left[\bigl(\tr[\bF_{\cE,\cF}\ketbra{\psi}^{\ox2}]\bigr)^2\right] 
&=\frac{4(\tr[\bF_{\cE,\cF}])^2+16\tr[(\tr_2[\bF_{\cE,\cF}])^2]+4\tr[(\bF_{\cE,\cF})^2]}{d(d+1)(d+2)(d+3)}.
\end{align}
For the fourth moment, classify the $24$ permutations by how many registers cross between the two pairs.
The $4$ permutations with no crossing give $(\tr[\bF_{\cE,\cF}])^2$, the $16$ with one crossing give $\tr[(\tr_2[\bF_{\cE,\cF}])^2]$, and the $4$ with two crossings give $\tr[(\bF_{\cE,\cF})^2]$; permutations within each pair act trivially on the symmetric support of $\bF_{\cE,\cF}$.
Moreover,
\begin{align}
-\frac{d+1}{2}I\preceq\tr_2[\bF_{\cE,\cF}]\preceq\frac{d+1}{2}I, 
\quad
\tr[(\tr_2[\bF_{\cE,\cF}])^2]\leq\frac{d(d+1)^2}{4},
\quad
\tr[(\bF_{\cE,\cF})^2]\leq\frac{d(d+1)}2. \notag
\end{align}
Consequently,
\begin{align}
\Var_\psi\tr[\cE(\ketbra{\psi})\cF(\ketbra{\psi})]
&= \bE_\psi\!\left[\bigl(\tr[\bF_{\cE,\cF}\ketbra{\psi}^{\ox2}]\bigr)^2\right]  - \left(\bE_\psi\tr[\bF_{\cE,\cF}\ketbra{\psi}^{\ox2}]\right)^2 \\
&\leq\frac{16\tr[(\tr_2[\bF_{\cE,\cF}])^2]+4\tr[(\bF_{\cE,\cF})^2]}{d(d+1)(d+2)(d+3)} \\
&\leq\frac{4d+6}{(d+2)(d+3)}\leq\frac4{d+1}.
\end{align}
For independent inputs, fix state $\phi$ and set $\bF_{\cE, \cF}^\phi:=\cE^\dagger(\cF(\ketbra{\phi}))$.
Then $0\preceq \bF_{\cE, \cF}^\phi\preceq I$, and the second input moment implies
\begin{align}
\Var_\psi\tr[\cE(\ketbra{\psi})\cF(\ketbra{\phi})]
=\frac{\tr[(\bF_{\cE, \cF}^\phi)^2]-(\tr[\bF_{\cE, \cF}^\phi])^2/d}{d(d+1)}
\leq\frac1{d+1}.
\end{align}
Averaging over $\psi$ while keeping $\phi$ fixed gives
\begin{align*}
\bE_\psi\tr[\cE(\ketbra{\psi})\cF(\ketbra{\phi})]
&=\tr[\cE(I/d)\cF(\ketbra{\phi})]
=\bra\phi\cF^\dagger(\cE(I/d))\ket\phi,
\end{align*}
This conditional mean still depends on $\phi$.
Since $\cE(I/d)$ is a density matrix and $\cF^\dagger$ is positive and unital,
$0\preceq\cF^\dagger(\cE(I/d))\preceq I$.
Applying the same second-moment bound to the random state $\phi$ therefore yields
\begin{align*}
\Var_\phi\!\left[\bE_\psi\tr[\cE(\ketbra{\psi})\cF(\ketbra{\phi})]\right]
&=\Var_\phi\!\left(\bra\phi\cF^\dagger(\cE(I/d))\ket\phi\right)
\leq\frac1{d+1}.
\end{align*}
The law of total variance now combines the conditional variance and the variance of the conditional mean:
\begin{align*}
\Var_{\psi,\phi}\tr[\cE(\ketbra{\psi})\cF(\ketbra{\phi})] 
& = \bE_\phi\!\left[\Var_\psi\tr[\cE(\ketbra{\psi})\cF(\ketbra{\phi})]\right] + \Var_\phi\!\left[\bE_\psi\tr[\cE(\ketbra{\psi})\cF(\ketbra{\phi})]\right] \\
&\leq\frac1{d+1}+\frac1{d+1}=\frac2{d+1}.
\end{align*}
\end{proof}

\subsection{DSEC with Independent Classical Shadow}
\label{app:dseu with cseu}

\subsubsection{Proof of Proposition~\ref{the:independent classical shadow upper bound}}
First, we prove that $\tilde{\gamma}$ is an unbiased estimator of $\tr[J_U J_V]/d^2$: 
\begin{align}
\bE \tilde{\gamma} 
&= \frac{1}{d^2} \bE \tr\left[\tilde{X} \tilde{Y}\right]
= \frac{1}{d^2} \tr\left[J_U J_V\right]. 
\end{align}
Then, we analyze the query complexity of this classical shadow-based algorithm. The variance of estimator $\tilde{\gamma}$ is given by 
\begin{align}
\Var(\tilde{\gamma}) 
&= \frac{1}{T^4 d^4} \bE \left(\sum_{i,j=1}^T \tr\left[\tilde{X}_i \tilde{Y}_j\right]\right)^2 - \frac{\left|\tr[U^\dagger V]\right|^4}{d^4} \\
&= \frac{1}{T^4 d^4} \bE \left(\sum_{i,j,k,l=1}^T \tr\left[\tilde{X}_i \tilde{Y}_j\right] \tr\left[\tilde{X}_k \tilde{Y}_l\right] \right)
- \frac{\left|\tr[U^\dagger V]\right|^4}{d^4}.
\end{align}
After expanding the expectation in the above equation, there are four terms:
\begin{enumerate}
\item For $i = k$ and $j = l$, there are $T^2$ terms, and using Lemma~\ref{lem:tr2 XY}, we have 
\begin{align}
\bE \left(\sum_{i,j} \tr^2\left[\tilde{X}_i \tilde{Y}_j\right] \right) 
= T^2 \bE \tr^2\left[X Y\right] 
\leq T^2 \cO\left(\frac{d^4(d+s)^4}{s^4} \right). 
\end{align}
\item For $i \neq k$ and $j \neq l$, there are $T^2(T-1)^2$ terms, and we have 
\begin{align}
\bE \left(\sum_{i \neq k, j \neq l}  \tr\left[\tilde{X}_i \tilde{Y}_j\right] \tr\left[\tilde{X}_k \tilde{Y}_l\right] \right)
= T^2 (T-1)^2 \left|\tr[U^\dagger V]\right|^4. 
\end{align}
\item For $i = k$ and $j \neq l$, there are $T^2(T-1)$ terms, and using Lemma~\ref{lem:tr XY tr XY} we have 
\begin{align}
\bE \left(\sum_{i = k, j \neq l}  \tr\left[\tilde{X}_i \tilde{Y}_j\right] \tr\left[\tilde{X}_k \tilde{Y}_l\right] \right)
&= T^2(T-1) \bE \tr\left[X \tilde{Y}_1\right] \tr\left[X \tilde{Y}_2\right] \\
&\leq T^2(T-1)\left(\left|\tr[U^\dagger V]\right|^4+\cO\left(d^4\left(\frac1s+\frac1d\right)\right)\right).
\end{align}
\item For $i \neq k$ and $j = l$, there are $T^2(T-1)$ terms, and using Lemma~\ref{lem:tr XY tr XY} we have 
\begin{align}
\bE \left(\sum_{i \neq k, j = l}  \tr\left[\tilde{X}_i \tilde{Y}_j\right] \tr\left[\tilde{X}_k \tilde{Y}_l\right] \right)
&= T^2(T-1) \bE \tr\left[\tilde{X}_1 Y\right] \tr\left[\tilde{X}_2 Y\right] \\
&\leq T^2(T-1)\left(\left|\tr[U^\dagger V]\right|^4+\cO\left(d^4\left(\frac1s+\frac1d\right)\right)\right).
\end{align}
\end{enumerate}
Therefore, we have 
\begin{align}
\Var(\tilde{\gamma}) 
&\leq \cO\left(\frac{(d+s)^4}{T^2s^4}+\frac1T\left(\frac1s+\frac1d\right)\right)
-\frac{\left|\tr[U^\dagger V]\right|^4}{T^2d^4}\\
&\leq \cO\left(\frac{(d+s)^4}{T^2s^4}+\frac1T\left(\frac1s+\frac1d\right)\right).
\end{align}
For fixed $s$ and $0<\varepsilon\leq1$, controlling the two terms in the variance bound gives a sufficient total query complexity of
\begin{align*}
\cO\left( \max\left\{\frac{(d+s)^2}{s\varepsilon},\frac{1+s/d}{\varepsilon^2}\right\}\right).
\end{align*}
Taking $s=d$ and $T=\Omega(\max\{1/\varepsilon,1/(d\varepsilon^2)\})$ suffices to make this variance at most $\varepsilon^2/3$. 
Then, we have 
\begin{align}
Ts=\cO\left(\max\left\{\frac d\varepsilon,\frac1{\varepsilon^2}\right\}\right).
\end{align}

\paragraph*{Incoherent case.}
For $s=1$, the unbiasedness and moment estimates used above extend to unital channels, so the same variance bound applies.
Substituting $s=1$ gives
\begin{align*}
\Var(\tilde{\gamma})
&\leq\cO\left(\frac{(d+1)^4}{T^2}+\frac1T\left(1+\frac1d\right)\right)
=\cO\left(\frac{d^4}{T^2}+\frac1T\right).
\end{align*}
Consequently, taking $T=\Omega(\max\{d^2/\varepsilon,1/\varepsilon^2\})$ suffices to achieve additive error $\varepsilon$ with success probability at least $2/3$.
Since each snapshot uses one query, the query complexity per device is $\cO(\max\{d^2/\varepsilon,1/\varepsilon^2\})$.

Therefore, we complete the proof of Proposition~\ref{the:independent classical shadow upper bound}.

\subsubsection{Technical Lemmas}
In the following, we present some technical lemmas that are used in the above analysis.

\begin{lemma}[Lemma~D1 in~\cite{li2025nearly}]
\label{lem:lemma d1 in li2025}
Suppose that $\psi$ and $\phi$ are the random input states and measurement outcomes of learning a unitary channel $\cU$, then we have 
\begin{align}
\bE \phi^{\ox 2} \ox \psi^{\ox 2}
&= \frac{2}{(d+s)(d+s+1)}\sum_{i=1}^4 \Delta_{U,i}, 
\end{align}
where 
\begin{align}
\Delta_{U, 1} &:= \frac{1}{\kappa_2} \Pi_{\rm sym}^{(d, 2)} \ox \Pi_{\rm sym}^{(d, 2)}, \\
\Delta_{U, 2} &:= \frac{s}{\kappa_3} (I\ox U\ox I\ox I)\left[I_1\ox \left(\Pi_{\rm sym}^{(d, 3)}\right)_{2,3,4}\right] (I\ox U^\dagger \ox I\ox I) \left(\Pi_{\rm sym}^{(d,2)}\ox I\ox I\right), \\
\Delta_{U, 3} &:= \frac{s}{\kappa_3} (U\ox I\ox I\ox I)\left[I_2\ox \left(\Pi_{\rm sym}^{(d, 3)}\right)_{1,3,4}\right] (U^\dagger \ox I\ox I\ox I) \left(\Pi_{\rm sym}^{(d, 2)}\ox I\ox I\right), \\
\Delta_{U, 4} &:= \frac{s(s-1)}{2\kappa_4} (U\ox U\ox I\ox I) \Pi_{\rm sym}^{(d, 4)} (U^\dagger\ox U^\dagger\ox I\ox I). 
\end{align}
\end{lemma}

\begin{lemma}
\label{lem:tr2 XY}
Let $X$ and $Y$ be independent classical snapshots defined in Eq.~\eqref{eq:classical snapshot} with POVM $\cM_s$ for unitary channels $\cU$ and $\cV$, we have 
\begin{align}
\bE \tr^2 \left[X Y\right] 
\leq \cO\left(\frac{d^4(d+s)^4}{s^4} \right). 
\end{align}
\end{lemma}
\begin{proof}
Using the snapshot definition, we have
\begin{align}
\tr[XY]
&=\frac{d^2(d+1)^2(d+s)^2}{s^2}\tr[\phi_1\phi_2]\tr[\psi_1\psi_2] - \frac{2d(d+1)(d+s)(d+s+1)}{s^2}
+\frac{d^2(d+s+1)^2}{s^2}. \notag
\end{align}
The last two terms are $\cO(d^2(d+s)^2/s^2)$. To bound the square of the first term, let
$M_U:=\bE(\phi^{\ox2}\ox\psi^{\ox2})$ for a snapshot of $\cU$, and define $M_V$ similarly.
These are density operators. Lemma~\ref{lem:lemma d1 in li2025} and the fact that each projector and unitary has operator norm at most one give
\begin{align}
\norm{M_U}{\infty}
&\leq\frac{2}{(d+s)(d+s+1)}
\left(\frac1{\kappa_2}+\frac{2s}{\kappa_3}+\frac{s(s-1)}{2\kappa_4}\right)\notag\\
&\leq\frac{2}{(d+s)^2}
\left(\frac2{d^2}+\frac{12s}{d^3}+\frac{12s^2}{d^4}\right)\notag\\
&=\frac{4d^2+24ds+24s^2}{d^4(d+s)^2}
\leq\frac{24}{d^4}=\cO(d^{-4}).
\end{align}
Here we used $\kappa_r=\binom{d+r-1}{r}\geq d^r/r!$ for $r=2,3,4$ and $s(s-1)\leq s^2$.
The last inequality follows from $4d^2+24ds+24s^2\leq24(d+s)^2$, so the bound holds uniformly for all $s\geq1$.
Independence of the two snapshots therefore implies
\begin{align}
\bE\left[(\tr[\phi_1\phi_2]\tr[\psi_1\psi_2])^2\right]
=\tr[M_UM_V]\leq\norm{M_U}{\infty}\tr[M_V]=\cO(d^{-4}).
\end{align}
Squaring the snapshot inner product and applying $(a+b)^2\leq2a^2+2b^2$ now gives
\begin{align}
\bE\tr^2[XY]
&\leq\cO\left(\frac{d^4(d+1)^4(d+s)^4}{s^4d^4}+\frac{d^4(d+s)^4}{s^4}\right)
=\cO\left(\frac{d^4(d+s)^4}{s^4}\right).
\end{align}
\end{proof}

\begin{lemma}
\label{lem:tr XY tr XY}
Let $\tilde{Y}_1$ and $\tilde{Y}_2$ be independent classical snapshots defined in Eq.~\eqref{eq:classical snapshot} with POVM $\cM_s$ for unitary channel $\cV$, and let $X$ be an independent classical snapshot for unitary channel $\cU$. Then
\begin{align}
\bE\tr\left[X\tilde{Y}_1\right]\tr\left[X\tilde{Y}_2\right]
&\leq \left|\tr[U^\dagger V]\right|^4
+\cO\left(d^4\left(\frac1s+\frac1d\right)\right).
\end{align}
Similarly, let $\tilde{X}_1$ and $\tilde{X}_2$ be independent classical snapshots for $\cU$, and let $Y$ be an independent classical snapshot for $\cV$. Then
\begin{align}
\bE\tr\left[\tilde{X}_1Y\right]\tr\left[\tilde{X}_2Y\right]
&\leq \left|\tr[U^\dagger V]\right|^4
+\cO\left(d^4\left(\frac1s+\frac1d\right)\right).
\end{align}
\end{lemma}
\begin{proof}
We first prove the bound involving the shared snapshot $X$.
We have 
\begin{align}
\bE\tr\left[X\tilde{Y}_1\right]\tr\left[X\tilde{Y}_2\right]
&=\bE\!\left[(\tr[XJ_V])^2\right]
=\left|\tr[U^\dagger V]\right|^4+\Var(\tr[XJ_V]),
\end{align}
where the last equality uses $\bE X=J_U$ and $\tr[J_UJ_V]=|\tr[U^\dagger V]|^2$.
For the input $\psi$ and outcome $\phi$ in $X$, define
$f_\psi:=|\bra\psi U^\dagger V\ket\psi|^2$ and $r:=\bra{V\psi}\phi\ket{V\psi}$.
Lemma~\ref{lem:expectation and second moment of collective measurement result} gives
\begin{align}
\bE[r\mid\psi]&=\frac{1+sf_\psi}{d+s},\\
\bE[r^2\mid\psi]&=\frac{2+4sf_\psi+s(s-1)f_\psi^2}{(d+s)(d+s+1)},\\
\Var(r\mid\psi)&=\frac{(d+s-1)(1+2sf_\psi)-s(d+2s)f_\psi^2}{(d+s)^2(d+s+1)}
\leq\frac{3s}{(d+s)^2}.
\end{align}
Since
\begin{align}
\frac{\tr[XJ_V]}{d^2}
=\frac{(d+1)(d+s)}{sd}r-\frac{d+s+1}{sd},
\qquad
\bE\!\left[\frac{\tr[XJ_V]}{d^2}\,\middle|\,\psi\right]
=\frac{d+1}{d}f_\psi-\frac1d,
\end{align}
the conditional variance of $\tr[XJ_V]/d^2$ is $\cO(1/s)$.
Lemma~\ref{lem:variance of input correlations} gives $\Var_\psi(f_\psi)=\cO(1/d)$, so the law of total variance yields
\begin{align}
\Var\!\left(\frac{\tr[XJ_V]}{d^2}\right)=\cO\left(\frac1s+\frac1d\right).
\end{align}
Multiplying this variance bound by $d^4$ and substituting into the product-moment identity above proves the first bound.
The second follows by exchanging the two devices.
\end{proof}

\end{document}